\documentclass{llncs}
\usepackage{latexsym}
\usepackage{amssymb}
\usepackage{amsmath}
\usepackage{graphicx, subfigure}
\usepackage{enumitem}
\usepackage{subfigure}
\usepackage{multirow}
\usepackage[left=2cm,top=3cm,right=2cm]{geometry} 

\def\wedge{\mathcal{W}}
\def\obj{\mathcal{O}}
\newcommand{\flag}{\mathord{\it flag}}
\newcommand{\Ob}{\mathcal{O}}
\newcommand{\TRUE}{\mathord{\it true}}
\newcommand{\FALSE}{\mathord{\it false}}

\title{Probing Convex Polygons with a Wedge\thanks{This work was supported by NSERC and FQRNT.}}
\author{Prosenjit Bose and Jean-Lou De Carufel and Alina Shaikhet and Michiel Smid}
\institute{Carleton University}

\begin{document}

\newtheorem{observation}[question]{Observation}

\newtheorem{algorithm}[question]{Algorithm}

\maketitle

\abstract{Minimizing the number of probes is one of the main challenges in reconstructing geometric objects with probing devices. In this paper, we investigate the problem of using an \emph{$\omega$-wedge} probing tool to determine the exact shape and orientation of a convex polygon. An $\omega$-wedge consists of two rays emanating from a point called the \emph{apex} of the wedge and the two rays forming an angle $\omega$. To probe with an $\omega$-wedge, we set the direction that the apex of the probe has to follow, the line $\overrightarrow L$, and the initial orientation of the two rays. A \emph{valid $\omega$-probe} of a convex polygon $\obj$ contains $\obj$ within the $\omega$-wedge and its outcome consists of the coordinates of the apex, the orientation of both rays and the coordinates of the closest (to the apex) points of contact between $\obj$ and each of the rays.

We present algorithms minimizing the number of probes and prove their optimality. In particular, we show how to reconstruct a convex $n$-gon (with all internal angles of size larger than $\omega$) using $2n-2$ $\omega$-probes; if $\omega = \pi/2$, the reconstruction uses $2n-3$ $\omega$-probes. We show that both results are optimal. Let $N_B$ be the number of vertices of $\obj$ whose internal angle is at most $\omega$, (we show that $0 \leq N_B \leq 3$). We determine the shape and orientation of a general convex $n$-gon with $N_B=1$ (respectively $N_B=2$, $N_B=3$) using $2n-1$ (respectively $2n+3$, $2n+5$) $\omega$-probes. We prove optimality for the first case. Assuming the algorithm knows the value of $N_B$ in advance, the reconstruction of $\obj$ with $N_B=2$ or $N_B=3$ can be achieved with $2n+2$ probes,- which is optimal. 
}

\section{Introduction}
\label{sec:introduction}

Geometric probing is a branch of computational geometry looking into how to determine shapes of unknown geometric objects by using special measurements - \emph{probes}. It is an area of research with various applications in robotics, automated manufacturing and tomography to mention only a few. Geometric probing was introduced by Cole and Yap~\cite{DBLP:journals/jal/ColeY87} in 1983. Many probing tools, together with reconstruction algorithms, have been developed - finger probes~\cite{DBLP:journals/jal/ColeY87}, hyperplane (or line) probes~\cite{DBLP:conf/stoc/DobkinEY86,DBLP:journals/ipl/Li88a}, diameter probes~\cite{DBLP:journals/ijrr/RaoG94}, x-ray probes~\cite{DBLP:journals/siamcomp/EdelsbrunnerS88,Gardner92}, histogram (or parallel x-ray) probes~\cite{MeijerSkiena96}, half-plane probes~\cite{DBLP:journals/jal/Skiena91}, composite probes~\cite{DBLP:journals/amai/BrucksteinL91,DBLP:journals/ipl/Li88a,Skiena88Dissertation} among others. 
Another closely related area, called geometric testing, studies a verification problem: find a set of probes to determine if a given set of geometric objects contains one which is equivalent to a certain query object (see~\cite{RomanikSurvey} for a survey).

\begin{table}[ht]
\caption{Summary of results for reconstructing convex polygons. $N_B$ stands for the number of angles smaller than $\omega$ in the convex polygon.}
\centering
\begin{tabular}{l|c|c|c|c}
\hline\hline 
Probing Tool & Notes & Lower Bound & Upper Bound & References \\[0.5ex]
\hline
finger & & $3n-1$ & $3n$ & ~\cite{DBLP:journals/jal/ColeY87} \\[0.3ex]
\hline
hyperplane/line & & $3n+1$ & $3n+1$ & ~\cite{DBLP:journals/ipl/Li88a} \\[0.3ex]
\hline
diameter & Reconstruction is not possible & * & * & ~\cite{DBLP:journals/ijrr/RaoG94} \\[0.3ex]
\hline
x-ray & & $2n$ & $5n+19$ & ~\cite{DBLP:journals/siamcomp/EdelsbrunnerS88} \\[0.3ex]
\hline
half-plane & & $2n$ & $7n+7$ & ~\cite{DBLP:journals/jal/Skiena91}\\[0.3ex]
\hline
\multirow{5}{*}{$\omega$-wedge} & $\omega = \pi/2$, $N_B = 0$ & $2n-3$ & $2n-3$ & \multirow{5}{*}{our results} \\[0.5ex]
  & $0 < \omega < \pi/2$, $N_B = 0$ & $2n-2$ & $2n-2$ \\[0.3ex]
  & $0 < \omega \leq \pi/2$, $N_B = 1$ & $2n-1$ & $2n-1$ \\[0.3ex]
  & $0 < \omega \leq \pi/2$, $N_B = 2$ & $2n+2$ & $2n+3$ \\[0.3ex]
  & $0 < \omega \leq \pi/2$, $N_B = 3$ & $2n+2$ & $2n+5$ \\[0.3ex]
\hline
\end{tabular}
\label{table:table1}
\end{table}

Various probing schemes along with upper and lower bounds on performance are summarized in the upper part of Table~\ref{table:table1} (see the review paper of Skiena~\cite{DBLP:journals/algorithmica/Skiena89}, discussing the methodology as well the challenges of geometric probing).

\vspace{0.5em}
For a given probing tool, we seek to reconstruct an object with as few probes as possible; hence much effort has gone to finding an algorithm with sufficiently tight upper and lower bounds or, ultimately, an \emph{optimal} algorithm. We introduce a new probing tool - an $\omega$-wedge. It consists of two rays emanating from a point called the {\it apex} of the wedge, forming an angle $\omega$, where $\omega$ is a fixed real number with $0 < \omega \leq \pi/2$. To probe a convex polygon $\obj$ we choose a directed line 
$\overrightarrow L$ and position the apex of the $\omega$-wedge on $\overrightarrow L$ such that $\obj$ and the positive direction of $\overrightarrow L$ are contained in the wedge. Initially, $\obj$ does not touch the rays of the wedge. Informally, imagine the $\omega$-wedge moving from its initial position along and in the direction of $\overrightarrow L$ until both rays contact $\obj$ in which case the wedge cannot move any further. A valid $\omega$-probe of a convex polygon $\obj$ is a placement of an $\omega$-wedge, such that $\obj$ is contained in the wedge and touches both of its rays. An outcome of a valid $\omega$-probe consists of the coordinates of the apex, the orientation of both rays and the coordinates of the points of contact between $\obj$ and each of the rays closest to the apex. Our probing tool generalizes previous work (see e.g.,~\cite{DBLP:conf/cccg/BoseC10},~\cite{DBLP:conf/isaac/FleischerW09}). Our method of angular probing was inspired by problems on enclosing triangles~\cite{DBLP:conf/cccg/BoseC10}, namely searching for all $\omega$-angle triangles of minimum area which enclose a  given set of points in the plane. At the same time, the way of probing with a wedge can be thought of as a photographer locating himself and orienting his camera such that $\obj$ fits exactly in the field of view of his camera~\cite{DBLP:conf/isaac/FleischerW09}.

\vspace{0.5em}
We present an algorithm that reconstructs a convex $n$-gon with all internal angles of size bigger than $\omega$ using $2n-2$ $\omega$-probes (for $0 < \omega < \pi/2$). When $\omega = \pi/2$, the reconstruction uses $2n-3$ $\omega$-probes. We prove optimality for both cases. We show that $2n-1$ probes are necessary and sufficient to reconstruct a convex polygon with exactly one vertex whose internal angle is at most $\omega$. The main difficulty that is common to reconstruction problems, both in theory and in practice, is the reconstruction of sharp corners (or sharp edges if we talk about $3$ dimensions); see, e.g.,~\cite{DBLP:conf/siggraph/AmentaBK98}. Those areas of so called high-curvature cause difficulties for our probing tool as well. When the polygon has $2$ or $3$ angles smaller than or equal to $\omega$, our results are \emph{almost} optimal. In particular, for polygons with exactly $2$ (respectively, $3$) angles of size at most $\omega$, we show a reconstruction strategy that uses $2n+3$ (respectively, $2n+5$) $\omega$-probes, while $2n+2$ probes are necessary.  Our results are summarized in Table~\ref{table:table1}. Our reconstruction strategies are adaptive, i.e., the choice of parameters for each probe depends on all previous outcomes. We also prove lower bounds for all of our algorithms using adversarial arguments.

\vspace{0.5em}
Our probing tool is similar to finger probing~\cite{DBLP:journals/jal/ColeY87}, where each probe is formally defined to be a directed line $\overrightarrow L$. In finger probing there is a point moving along and in the direction of $\overrightarrow L$ until it contacts $\obj$. In our research the point is equipped with two rays. The movement of the apex $q$ along $\overrightarrow L$ terminates when both rays contact $\obj$ ($q$ is not necessarily touches $\obj$). Our model of probing achieves better results than finger probing (consider Table~\ref{table:table1}) mainly because the outcome of a valid $\omega$-probe contains one or two points of contact with $\obj$ (which are vertices of $\obj$), while in finger probing the outcome of a successful probe is a single point (not necessarily a vertex of $\obj$). Moreover, almost every edge of a polygon, reconstructed via finger probing, contains two points of contact in its interior. In our model only one probe is sufficient to verify an edge.
 
\vspace{0.5em}
Our model of probing should not be confused with diameter or "parallel-jaw gripper`` model of probing~\cite{DBLP:journals/ijrr/RaoG94}, where the shape of a polygon is determined from a sequence of projections. Rao and Goldberg showed that reconstruction of $\obj$  with a diameter probing tool is not possible, -- for a given set of diameter measurements, there is an infinite set of polygonal shapes consistent with these measurements. We define the angle between the rays of our probing tool to be strictly bigger than $0$. Thus, the rays of an $\omega$-wedge are never parallel. Another difference is that a valid $\omega$-probe returns points of contact with $\obj$, while a diameter probe does not,- it returns only the distance between supporting lines.

\vspace{0.5em}
Our probing tool is quite simple compared to other probing devices such as x-ray or histogram. It can be manufactured from two tactile whisker-like sensors and does not require complicated software.
 
\vspace{0.5em}
The organization of the paper is as follows. In Section~\ref{sec:definitions} we present definitions related to our probing method. In Section~\ref{sec:preliminaries} we give some observations, review known facts and prove some fundamental properties related to $\omega$-probing. We concentrate on the reconstruction of convex polygons in Sections~\ref{sec:Polygons_without_B-vertices} and~\ref{sec:Polygons_with_B-vertices}. In Section~\ref{sec:Polygons_without_B-vertices} we present an optimal algorithm for reconstructing convex polygons that do not have \emph{narrow vertices} (refer to Definition~\ref{def:B-vertex}). Reconstruction of general convex polygons together with the discussion on upper and lower bounds on the performance are given in Section~\ref{sec:Polygons_with_B-vertices}. Concluding remarks and possible directions for future research are presented in Section~\ref{sec:Conclusion}.

\section{Definitions}
\label{sec:definitions}

Given a convex polygon $\obj$ in the plane, the goal is to probe $\obj$ to determine its exact shape and orientation with a minimum number of probes.  We probe $\obj$ with a wedge $W$ consisting of two infinite arms making an angle of $\omega$ (for some fixed angle $0 < \omega \leq \pi/2$). An $\omega$-probe of $\obj$ returns the coordinates of the apex of $W$ together with the orientation of the arms and one point of contact of each arm with $\obj$, such that the returned point of contact is the point of $\obj$ that is closest to the apex. This is formalized in the following definitions.
For the rest of this paper, $\obj$ is a convex polygon in the plane and $\omega$ is an angle such that $0 < \omega \leq \pi/2$. At the end of the paper, we discuss the case where $\pi/2 < \omega < \pi$ (refer to Section~\ref{sec:Conclusion}).

\begin{definition}[$\omega$-Wedge]
\label{def:omega-wedge}
Let $q$ be a point in the plane. Let
$H_1$ and $H_2$ be two rays emanating from $q$ such that the angle
between $H_1$ and $H_2$ is $\omega$ and $(q, H_1, H_2)$ forms a left-turn.  The closed set formed
by $q$, $H_1$, $H_2$ and the points between $H_1$ and $H_2$
is an \emph{$\omega$-wedge}, denoted $\wedge(\omega,q,H_1,H_2)$.
The point $q$ is the \emph{apex} of the $\omega$-wedge. (refer to Figure~\ref{fig:omega_wedge}).
\end{definition}

\begin{figure}
    \begin{center}
        \subfigure[]{
            \label{fig:omega_wedge}
            \includegraphics[width=0.15\textwidth]{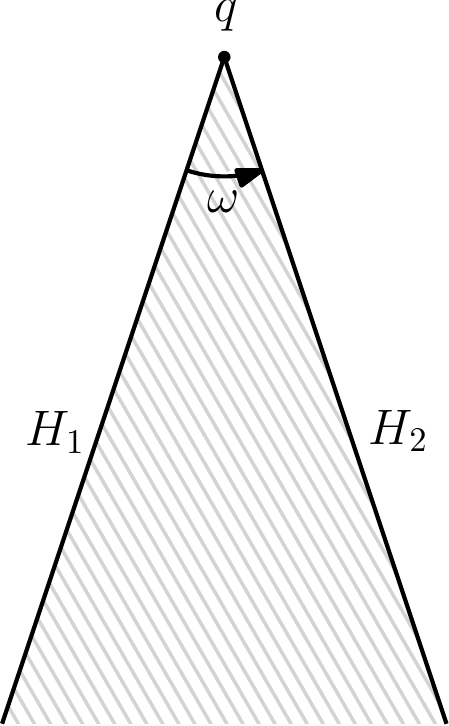}
        }%
         \subfigure[]{
            \label{fig:probe_example}
            \includegraphics[width=0.4\textwidth]{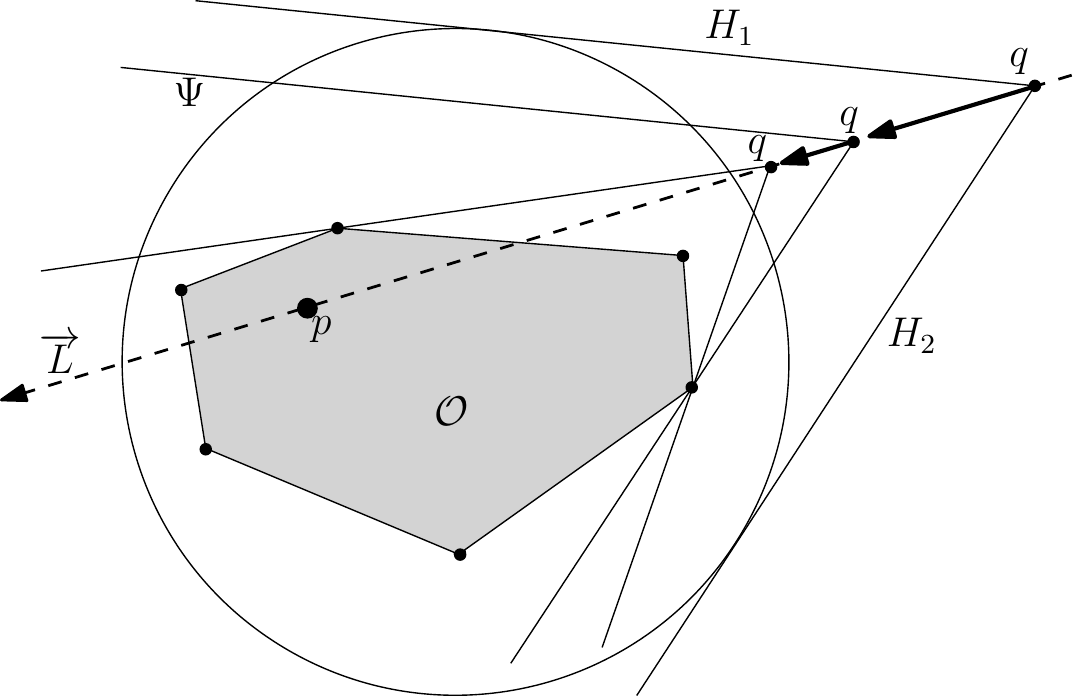}
        }%
        \subfigure[]{
            \label{fig:miss_O}
            \includegraphics[width=0.4\textwidth]{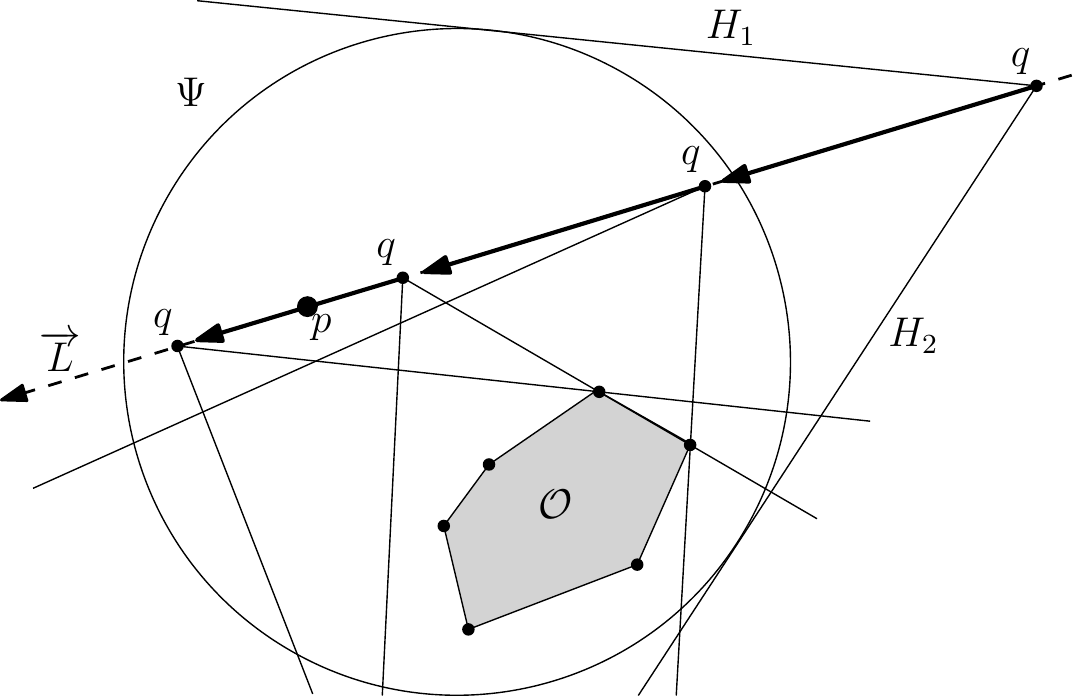}
        }%
    \end{center}
    \vspace{-1.5em}
    \caption{Definition and overview of the probing tool. \textbf{(a)} $\omega$-wedge; \textbf{(b)} Probing along a directed line $\protect \overrightarrow{L}$ that intersects $\obj$; \break \textbf{(c)} Despite the fact that $\obj$ is inside $W$, the probe will miss $\obj$ because $p \notin \obj$.}
\end{figure}

To probe with an $\omega$-wedge, we set the direction that the apex $q$ of the probe has to follow, the line $L$, and the initial orientation of the probe. As the apex $q$ of the probe moves along $L$, the arms can move freely provided that they maintain an angle of $\omega$ and $\obj$ stays between the arms. We do not have any control over rotation of the arms.
When $q$ cannot move anymore, then we have a valid $\omega$-probe of $\obj$ (refer to Definition~\ref{def:omega_probe}).

\begin{definition}[valid $\omega$-Probe of $\obj$]
\label{def:omega_probe}
A valid $\omega$-probe of $\obj$ is a quadruple ($\omega$, $q$, $H_1$, $H_2$) such that $W=\wedge(\omega,q,H_1,H_2)$ is an $\omega$-wedge; $\obj \subseteq W$; $H_1$ (respectively $H_2$) contains at least one point of $\obj$.
\end{definition}

\begin{definition}[Probing along a directed line $\overrightarrow L$]
\label{def:probe_along_L}
Let $\overrightarrow L$ be a directed line and ($\omega$, $q$, $H_1$, $H_2$) be a valid $\omega$-probe of $\obj$  such that the apex $q$ is on $\overrightarrow L$, and the ray from $q$ in the direction of $\overrightarrow L$ intersects $\obj$. We say that ($\omega$, $q$, $H_1$, $H_2$) is a valid $\omega$-probe of $\obj$ \emph{with respect to $\overrightarrow L$}. Refer to Figure~\ref{fig:probe_example}.
\end{definition}

We assume that we know a point $p$ and a circle $\Psi$ such that $p \in \obj$ and $\obj \subset \Psi$. The hypothesis $\obj \subset \Psi$ guarantees that $\obj$ is inside $W$ (refer to~\cite{DBLP:conf/isaac/FleischerW09} for instance). The hypothesis $p\in \obj$ guarantees that $H_1$ and $H_2$ will eventually touch $\obj$ as $q$ approaches $p$ (refer to~\cite{DBLP:journals/jal/ColeY87} for instance). Otherwise, we could miss the object.
Indeed, suppose that for some probe we choose a point $p \in \Psi \setminus \obj$ to aim at. Suppose that we choose a direction $\overrightarrow L$ containing $p$ that does not intersect $\obj$
(refer to Figure~\ref{fig:miss_O}). As $q$ moves along $\overrightarrow L$,
$H_2$ will eventually touch $\obj$. Then $H_2$ will turn counter-clockwise around $\obj$ as $q$ goes further. But $q$ can go forever along $\overrightarrow L$ and $H_1$ will never touch $\obj$. So we miss $\obj$ even though $\obj$ is inside $W$. This shows that for every probe, we need to aim towards a point in $\obj$.

\begin{definition}[Outcome of probing along a directed line $\overrightarrow L$]
\label{def:outcome_of_probing}
The outcome of probing with an $\omega$-wedge along a directed line $\overrightarrow L$ consists of the following values: $q$, $H_1$, $H_2$, $p_1$ and $p_2$, such that $\wedge(\omega,q,H_1,H_2)$ is a valid $\omega$-probe and the point $p_1$ (respectively $p_2$) is the point of $\obj \cap H_1$ (respectively $\obj \cap H_2$) that is closest to $q$. 
\end{definition}

Notice that $p_1$ is on or to the right of $\overrightarrow{L}$, $p_2$ is on or to the left of $\overrightarrow{L}$ and we can also have $p_1 = p_2 = q$ (refer to Figure~\ref{fig:points_of_contact_3}).

Consider three consecutive vertices $v_{i+1}$, $v_i$ and $v_{i-1}$ of $\obj$ such that $\angle (v_{i+1}, v_i, v_{i-1}) < \omega$.
(refer to Figure~\ref{fig:points_of_contact_3}).

If $\obj$ is contained in the $\omega$-wedge $W$ and the apex $q$ of $W$ is equal to $v_i$, then there are infinitely many different probes with $q = v_i$. We say that $W$ encloses $\obj$ even if $H_1$ and $H_2$ touch $\obj$ only at $q=v_i$. This is precisely the situation that causes difficulty since a probe returning a value of $q=v_i$ gives no additional information.

\begin{figure}[h]
    \begin{center}
        \subfigure[The probe $W$ touches the polygon $\obj$ at its vertices.]{
            \label{fig:points_of_contact_1}
            \includegraphics[width=0.29\textwidth]{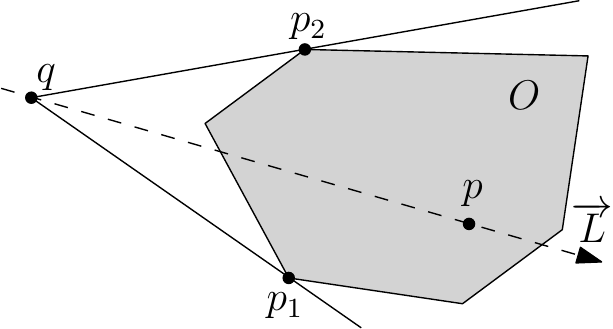}
        }%
        \hspace{1cm}%
        \subfigure[One or more polygon edges coincide with $H_1$ and $H_2$. Only the points of contact $p_1$ and $p_2$ closest to $q$ are reported.]{
            \label{fig:points_of_contact_2}
            \includegraphics[width=0.32\textwidth]{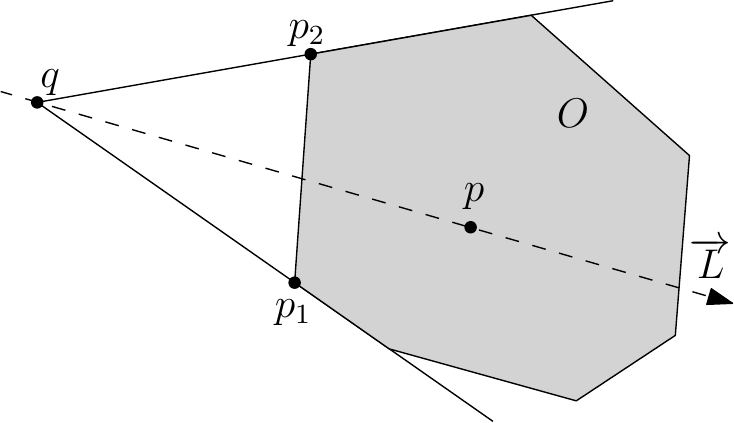}
        }%
        \hspace{0.5cm}%
        \subfigure[The apex $q$ of the probe touches $\obj$. In these cases $q$ is considered to be the only point of contact with the polygon.]{
            \label{fig:points_of_contact_3}
            \includegraphics[width=0.21\textwidth]{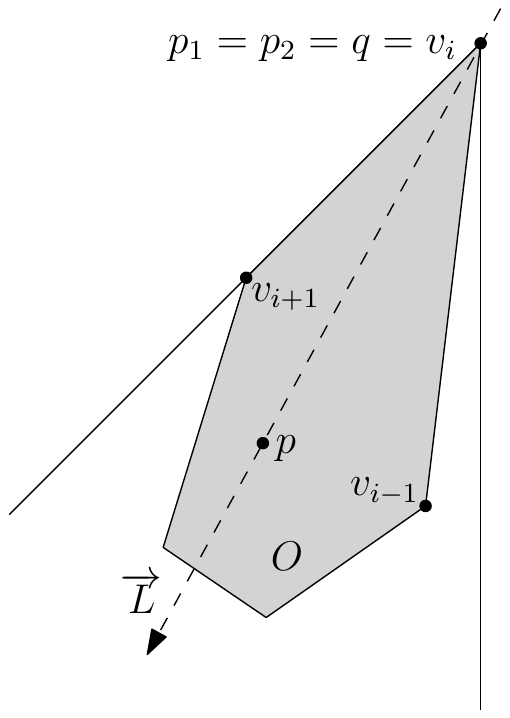}
            \hspace{0.3cm}%
            \includegraphics[width=0.21\textwidth]{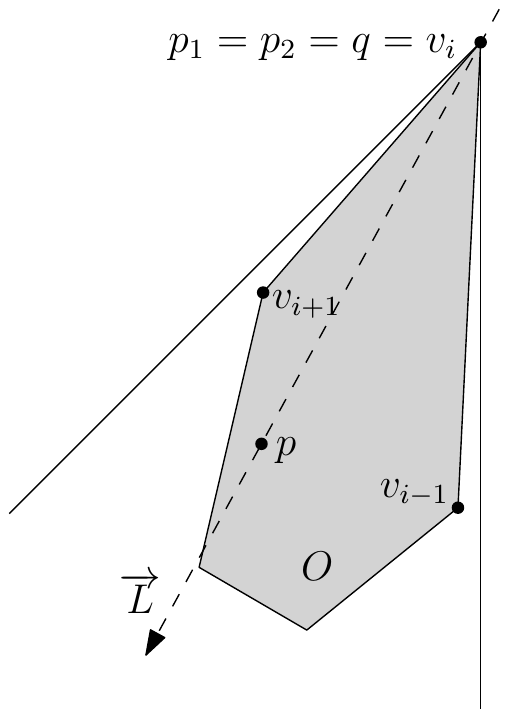}
            \hspace{0.3cm}%
            \includegraphics[width=0.21\textwidth]{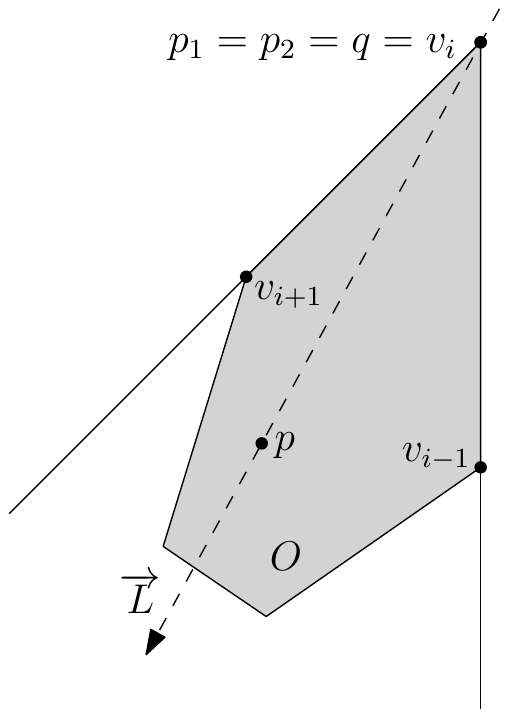}
        }%
    \end{center}
    \vspace{-1.5em}
    \caption{Possible outcomes of the probing regarding orientation and shape of the polygon.}
   \vspace{-0.8em}
\label{fig:points_of_contact}
\end{figure}

In this paper we study the reconstruction of $\obj$ by using an $\omega$-wedge model of probing. When it is clear from the context, we write ``probe" instead of $\omega$-probe.

\begin{definition}[$\omega$-Cloud]
\label{def:omega_cloud}
By rotating an $\omega$-wedge around $\obj$, such that it contains $\obj$ and both its rays are touching $\obj$, the apex traces a sequence of circular arcs. This sequence is called an \emph{$\omega$-cloud}, denoted $\Omega$ (refer to Figure~\ref{fig:simple_example}). In other words, an \emph{$\omega$-cloud} is the set of apices of all the valid $\omega$-probes of $\obj$. The circular arcs of $\Omega$ are labelled in counterclockwise order by $\Gamma_j$ for $0\leq j \leq n'-1$. We note that $n'=O(n)$~\cite{journals/ijcga/BoseMSS09}. The intersection point of a pair of consecutive circular arcs is called a \emph{pivot point}.
\end{definition}

Each point $x$ on the $\omega$-cloud has the property that there exists a valid $\omega$-probe with the apex placed at $x$. Moreover, if $x$ is not on the polygon, then this probe is unique (refer to Lemma~\ref{lem:uniqueness_lemma_2}).

\begin{figure}
\centerline{\resizebox{!}{7.0cm}{\includegraphics{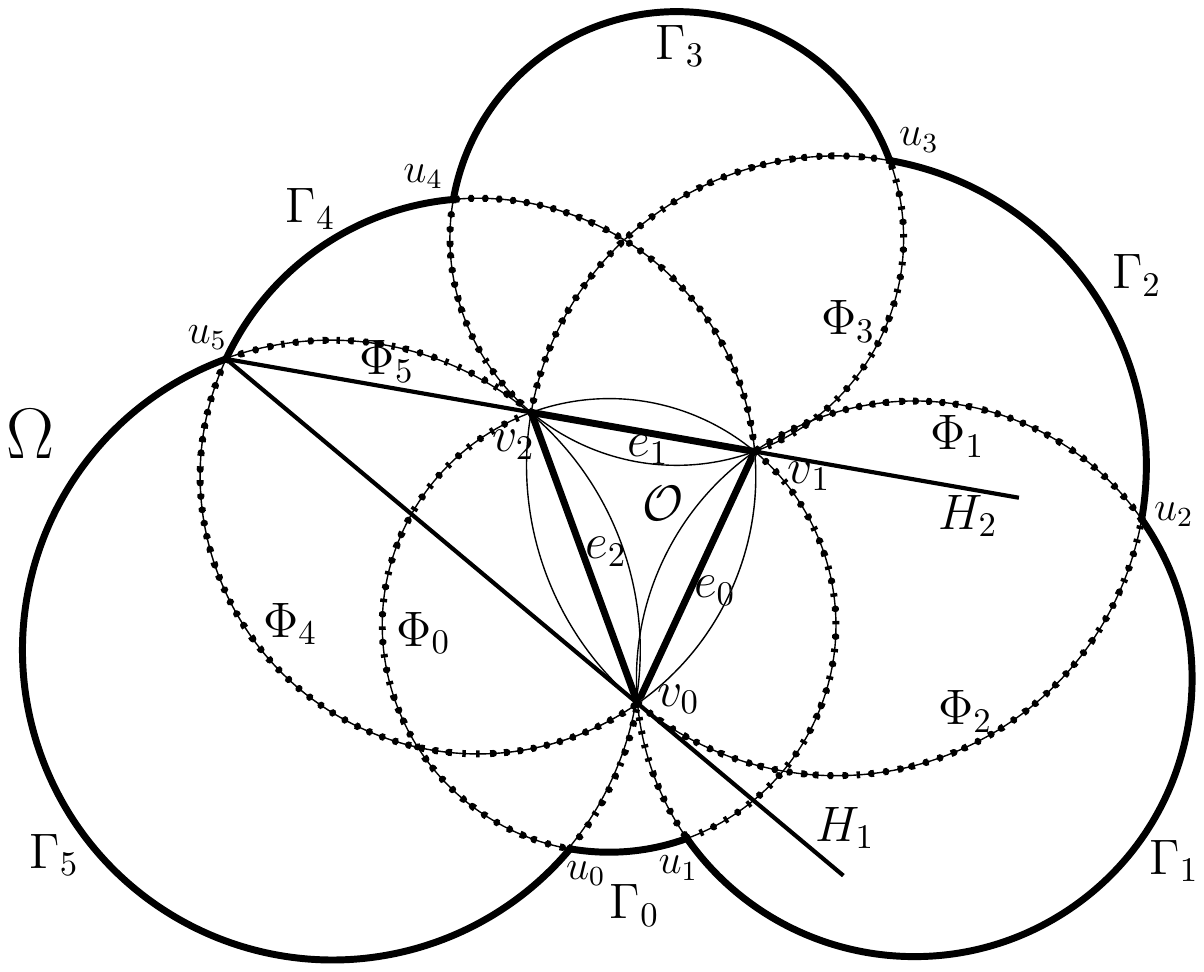}}}
\caption{$\Omega$ is the $\frac{\pi}{6}$-cloud of $\obj = \triangle( v_0, v_1, v_2)$.
A black line represents a circular arc of $\Omega$.
A dotted line represents the remaining part of the corresponding $\frac{\pi}{6}$-arc.
A grey line represents the remaining part of the corresponding circle.}
\label{fig:simple_example}
\end{figure}

 If a convex polygon $\obj$ is initially contained within $W$, the movement of $q$ on $\overrightarrow L$ will stop when both $H_1$ and $H_2$ have a non-empty intersection with $\obj$. There are two cases to consider:

\vspace{0.5em}
\emph{Case $1$}: $q$ is not on $\obj$. In this case, both arms $H_1$ and $H_2$ touch the convex polygon $\obj$ at points $p_1 \in H_1$ and $p_2 \in H_2$, such that $p_1 \neq p_2$ (which implies: $p_1 \neq q$ and $ p_2 \neq q$). The apex of the probe $q$ cannot move any further and the probe outputs the result together with the coordinates of the apex (which is a point of the $\omega$-cloud of $\obj$).

\vspace{0.5em}
\emph{Case $2$}: $q$ is on $\obj$. In the second case, when $q$ touches $\obj$ ($q$ can touch $\obj$ only at its vertices), the arms $H_1$ and $H_2$ of $W$ may or may not touch $\obj$ at points other than $q$. In fact, since the arms $H_1$ and $H_2$ are free to rotate around $q$, if the internal angle of $\obj$ at the vertex touched by $q$ is smaller than $\omega$, then there are infinitely many different positions for $W$ (refer to the first two examples of Figure~\ref{fig:points_of_contact_3}). If $q$ touches a vertex whose internal angle is $\omega$, there is only one position for $W$ (refer to the last example of Figure~\ref{fig:points_of_contact_3}).
Since we know neither $\obj$ nor the $\omega$-cloud $\Omega$ of $\obj$, it is not clear whether the arms $H_1$ and $H_2$ touch $\obj$ at points other than $q$.
We explain how to deal with this technicality in Section~\ref{sec:Polygons_with_B-vertices}. In this case, since $0 < \omega \leq \pi/2$, $q$ is a vertex of $\obj$ and a pivot point of the $\omega$-cloud $\Omega$ of $\obj$. 

\begin{definition}[Narrow Vertex]
\label{def:B-vertex}
A vertex of the convex polygon $\obj$ is called a \emph{narrow vertex} if its internal angle is at most $\omega$. 
\end{definition}

If the ray in the direction of an $\omega$-probe enters $\obj$ via a narrow vertex $v_B$, then the apex of this probe is on a vertex $v_B$ of $\obj$. Refer to Figure~\ref{fig:points_of_contact_3}. However, there can be no more than 3 vertices on the polygon where this happens (refer to Observation~\ref{observ:N_B_general}). 

\begin{definition}[$\omega$-Arc]
\label{def:omega_arc}
Let $a$ and $b$ be two points on a circle $\Phi$.  Denote by $\Gamma = \widehat{ba}$ the counterclockwise circular arc from $b$ to $a$.  By elementary geometry, for any point $q\in\Gamma$, $\angle (a, q, b )= \frac{\widehat{ab}}{2}$. (Here and throughout the paper we use counterclockwise angle naming convention). If $\angle (a, q, b) =
\frac{\widehat{ab}}{2} = \omega$, we say that $\Gamma$ is an
\emph{$\omega$-arc with respect to $a$ and $b$}.  The points $a$ and $b$ are called the \emph{supporting points} of $\Gamma$.
\end{definition}

Therefore, given two distinct points $a$ and $b$, there are two possible $\omega$-arcs: the one with respect to $a$ and $b$ and the one with respect to $b$ and $a$.

\begin{definition}[Feasible]
\label{def:feasible}
Let $Q$ be a convex polygon. Let $e$ be an edge of $Q$. Let $\mathcal{W}$ be the set of all $\omega$-wedges, each one containing $Q$. The feasible region $F_Q$ of $Q$ is the intersection of all wedges in $\mathcal{W}$. 
For any edge $e$ of $Q$, let $H_e$ be the half-plane containing $e$ on its boundary that does not contain $Q$. The feasible region $F_{Q,e}$ is defined as $F_Q \cap H_e$. (refer to Figure~\ref{fig:feasible}).
\end{definition}

\begin{figure}[h]
\centerline{\resizebox{!}{4.9cm}{\includegraphics{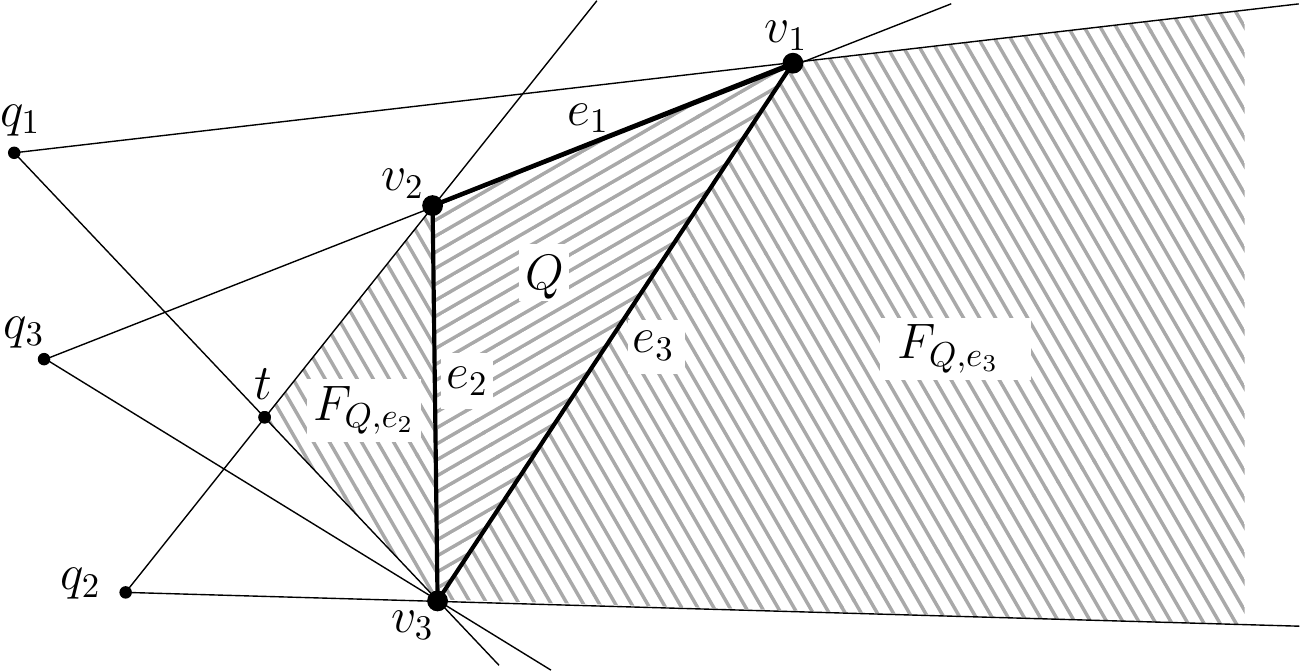}}}
\caption{$W$ consists of three $\omega$-wedges, each one containing $Q = \{ v_1, v_2, v_3\}$. The feasible region $F_Q$ is shown in a tiling pattern. The feasible region $F_{Q, e_2}$ of the edge $e_2$ is the triangle $\triangle(v_2, t, v_3)$. $F_{Q, e_1}$ degenerates into the line segment $\overline{v_1 v_2}$; in other words, $F_{Q, e_2} = e_1$ and thus $e_1 \in \obj$. $F_{Q, e_3}$ is unbounded.}
    \vspace{-1.0em}
\label{fig:feasible}
\end{figure}

\section{Preliminaries}
\label{sec:preliminaries}

When an $\omega$-probe is shot along a directed line that enters $\obj$ via a narrow vertex $v_B$, the probe stops only when the apex touches $v_B$ (refer to Figure~\ref{fig:points_of_contact_3}). However, there can be no more than 3 vertices on a convex polygon, that is not a rectangle, where this happens (refer to Observation~\ref{observ:N_B_general}), since the angle of an $\omega$-probe is $0 < \omega \leq \pi/2$. The number of acute angles is $4$ when $\obj$ is a rectangle and cannot exceed 3 otherwise.

\begin{observation}
\label{observ:N_B_general}
If $\obj$ is a convex polygon that is not a rectangle, then the number of narrow vertices cannot exceed 3.
\end{observation}

\begin{proof}
We prove that the maximum number of acute angles of a convex polygon is 3.
The sum of the exterior angles of any convex polygon is $2\pi$. If an inner angle is acute, then the corresponding exterior angle is greater than $\pi/2$. Suppose that 4 or more acute angles are possible. Since the polygon is convex, all the exterior angles are positive, so the sum of the exterior angles is at least the sum of the exterior angles supplementary to the four or more acute angles, which is strictly greater than $4 \pi/2 = 2\pi$, which is a contradiction.
\qed
\end{proof}

\begin{observation}
\label{observ:N_B_small_omega}
If $\obj$ is a convex polygon that is not a rectangle and $\omega < \pi / 3$, then the number of narrow vertices is at most 2.
\end{observation}

\begin{proof}
Assume that $\obj$ is not a rectangle and has three narrow vertices $v_{B_1}$, $v_{B_2}$ and $v_{B_3}$, that define a triangle $\triangle ( v_{B_1}, v_{B_2}, v_{B_3})$. At least one of the angles of the triangle is bigger or equal to $ \pi / 3$. Since the triangle completely resides inside $\obj$, at least one of the internal angles of $\obj$ at $v_{B_1}$, $v_{B_2}$ or $v_{B_3}$ should be bigger or equal to $ \pi / 3$. Thus if $\omega < \pi / 3$ then it is not possible for $\obj$ to have three narrow vertices.
\qed
\end{proof}

\begin{lemma}
\label{lem:uniqueness_lemma_1}
Assume that we are given a directed line $\overrightarrow L$. The apices of all valid $\omega$-probes of $\obj$ along $\overrightarrow L$ are located at the same point in the plane.
\end{lemma}
\begin{proof}
Assume to the contrary that there are two different valid $\omega$-probes of $\obj$ along $\overrightarrow{L}$ with apices located at different places. We denote them by ($q$, $H_1$, $H_2$, $p_1$, $p_2$) and ($q'$, $H'_1$, $H'_2$, $p'_1$, $p'_2$) and we have $q \neq q'$. Since $q \in \overrightarrow L$ and $q' \in \overrightarrow L$, the apex of one of the probes is contained in the $\omega$-wedge of the other probe. Assume, without loss of generality, that the apex $q'$ of the second  probe is contained in the $\omega$-wedge of the first probe, rooted at $q$. Both probes are taken along the same line $\protect \overrightarrow{L}$ (refer to Figure~\ref{fig:lemma1_proof}). Since the angle at both $\omega$-wedges is the same, at least one arm of the $\omega$-wedge at $q$ is completely outside of the interior of the $\omega$-wedge at $q'$ and thus does not touch $\obj$. This contradicts the definition of a valid $\omega$-probe. 
\qed
\end{proof}

\begin{figure}[h]
\centerline{\resizebox{!}{4.5cm}{\includegraphics{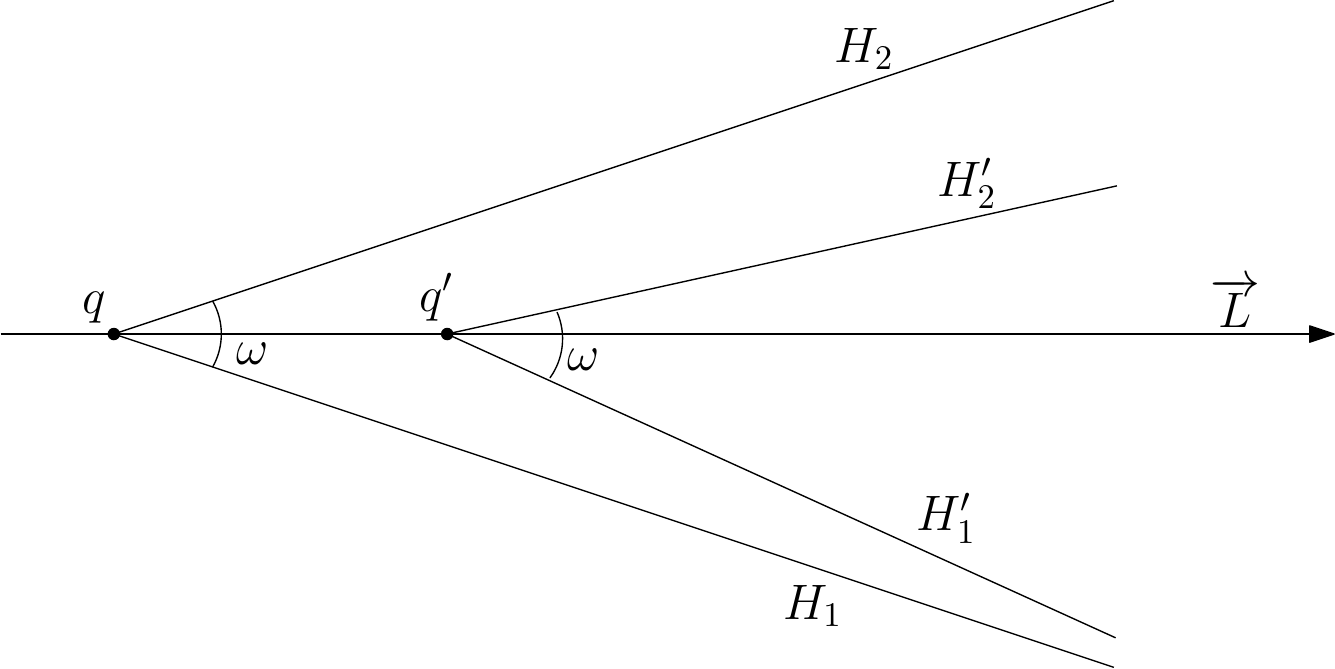}}}
\caption{Illustration of the proof of Lemma~\ref{lem:uniqueness_lemma_1}. The arm $H_2$ is completely outside of the interior of the $\omega$-wedge at $q'$ and thus does not touch $\obj$.}
\vspace{-1.5em}
\label{fig:lemma1_proof}
\end{figure}

\begin{lemma}
\label{lem:uniqueness_lemma_2}
Assume that we are given a directed line $\overrightarrow L$. If the apex of a valid $\omega$-probe of $\obj$ along $\overrightarrow L$ is not a narrow vertex, then this probe is unique.
\end{lemma}
\begin{proof}
Assume to the contrary that there are two different valid $\omega$-probes of $\obj$ along $\overrightarrow{L}$ such that the apices of those probes do not touch $\obj$: ($q$, $H_1$, $H_2$, $p_1$, $p_2$) and ($q'$, $H'_1$, $H'_2$, $p'_1$, $p'_2$). By Lemma~\ref{lem:uniqueness_lemma_1} we have $q = q'$. 

If $H_1 \neq H'_1$ (and thus $H_2 \neq H'_2$) then at least one arm of every probe does not touch $\obj$. This is a contradiction to the validity of the probes. 

If $q = q'$, $H_1 = H'_1$ and $H_2 = H'_2$, then $p_1 = p'_1$ and $p_2 = p'_2$ according to the definition of a valid $\omega$-probe.
\qed
\end{proof}

\begin{lemma}
\label{lem:feasible_triangle}
Let $Q$ be a convex polygon. Let $e$ be an edge of $Q$.
Suppose, $F_{Q,e}$ (see Definition~\ref{def:feasible}) is not a subset of the line through $e$. Then, $F_{Q,e}$ contains a triangle with base $e$.
\end{lemma}

In a more intuitive way, Lemma~\ref{lem:feasible_triangle} states the following: Let $Q$ be the convex hull of the vertices of $\obj$ that have been discovered so far, and let $e = (v_1, v_2)$ be an edge of $Q$. As long as $e$ has not been confirmed to be an edge of $\obj$, there is a possibility for an additional vertex to exist between $v_1$ and $v_2$ on the boundary of $\obj$. 

\vspace{0.5em}
To demonstrate the difficulty of the reconstruction problem, and to present a thorough study of our probing tool, we start with polygons without narrow vertices.

\section{Polygons without narrow vertices}
\label{sec:Polygons_without_B-vertices}

In this section we investigate the reconstruction of convex polygons that do not have narrow vertices. This guarantees that during the reconstruction we will not encounter any case where the apex of the wedge touches the polygon. Let $\obj$ be a convex polygon whose exact shape and orientation are unknown. However, we are given a point $p$ and a circle $\Psi$ such that  $p \in \obj \subset \Psi$. Each probe is initialized in an orientation where it encloses $\Psi$. A probe returns the coordinates of the points of contact with the polygon $\obj$ (as shown in Figure~\ref{fig:points_of_contact}), the orientation of the arms and the coordinates of the apex $q$.

\subsection{Algorithm for $0 < \omega \leq \pi/2$}
\label{subsec:algorithm}

The algorithm, given below, knows that all the internal angles of $\obj$ are larger than $\omega$. We will explain how to deal with general convex polygons in Section~\ref{sec:Polygons_with_B-vertices}. Notice, that at the start of the algorithm, the algorithm does not know $n$ - the number of vertices of $\obj$.

As part of the input the algorithm requires a polygon $Q \subseteq \obj$, such that each vertex of $Q$ is a vertex of $\obj$ and $Q$ can contain as little as two vertices of $\obj$. The input for the algorithm depends on the value of $\omega$. In the subsection~\ref{subsec:ImprovedAnalysis} we describe a specific case when $\omega = \pi/2$. In this case our input can be a bit more complicated ($Q$ can be quadrilateral or pentagon), and we can reduce the number of spent probes by one. The input for the algorithm is constructed during initial probing, which can consist of one probe only (for general $\omega$) or three to four probes (for $\omega = \pi/2$). Each vertex $u$ of the convex polygon $Q$ has a Boolean variable $\flag(u)$ that is $\TRUE$ if $u$ and its counter-clockwise successor in $Q$ form an edge of $\Ob$. The algorithm uses variable $F$ that is equal to the number of vertices of $Q$ whose $\flag$ is $\FALSE$. The role of the $\flag$s together with the construction of the input for the algorithm is given below. Notice, this input is suitable for any $\omega$ in the range $0 < \omega \leq \pi/2$.

\vspace{0.5em}
\rule{0.88\textwidth}{1pt}

{\bf Input 1:}

\begin{quotation}
\noindent
{\bf Invariant:} 
\begin{enumerate} 
\item $Q$ is a convex polygon. 
\item Each vertex of $Q$ is a vertex of $\Ob$ (and, therefore, $Q$ is contained in $\Ob$). 
\item Each vertex $u$ of $Q$ stores a Boolean variable $\flag(u)$. If $\flag(u) = \TRUE$, then $u$ and its counter-clockwise successor in $Q$ form an edge of $\Ob$. 
\item $F$ is a variable such that $F = | \{ u \in Q : \flag(u) = \FALSE \} |$. 
\end{enumerate}  
\vspace{0.5em} 
\noindent
{\bf Initialization:} 
To initialize the algorithm, choose an arbitrary line 
$\overrightarrow{L}$ that contains $p$. Let $(q,H_1,H_2,p_1,p_2)$ be the outcome of the probe along $\overrightarrow{L}$. Initialize $Q$ as the polygon consisting of the vertices $p_1$ and $p_2$, set 
$\flag(p_1) = \FALSE$, $\flag(p_2) = \FALSE$, and $F = 2$. 
\end{quotation}
   \vspace{-0.5em}
\rule{0.88\textwidth}{1pt}
 
\vspace{0.7em}
Since $\obj$ has no narrow vertices, every probe returns two distinct points of contact $p_1$ and $p_2$ that are vertices of $\obj$. The input convex polygon $Q$ consists of already discovered vertices of $\obj$. In every iteration, the algorithm takes two consecutive vertices of $Q$, $v_1$ and $v_2$, and checks whether they share an edge in $\obj$ by shooting a probe along the line that contains $v_1$ and $v_2$. This probe either confirms an edge or, if there is no edge that connects $v_1$ and $v_2$ in $\obj$, the probe reveals a vertex of $\obj$ that belongs to the boundary between $v_1$ and $v_2$. This new vertex is then added to $Q$. In both cases the other arm of the probe may also reveal new information about $\obj$ or may return a vertex or edge of $\obj$ that have already been discovered. The polygon $Q$ grows from within $\obj$. When all the edges of $Q$ are confirmed, the algorithm stops, meaning that $Q$ is equal to $\obj$.

The basic algorithm for the reconstruction of the polygon without narrow vertices is given below.

\rule{0.88\textwidth}{1pt}

{\bf Algorithm (Input):} 
\vspace{1.0em}

While $F \neq 0$, do the following:
\vspace{-1.0em}

\begin{quotation} 
~\begin{enumerate} 
\item Take an arbitrary vertex $u$ in $Q$ for which $\flag(u)=\FALSE$, and let $v$ be the counter-clockwise successor of $u$ in $Q$.  
\item Let $\overrightarrow{L}$ be the line through $u$ and $v$ that is directed from $u$ to $v$. Let $(q,H_1,H_2,p_1,p_2)$ be the outcome of the probe along $\overrightarrow{L}$.
\item Case 1: $p_1 \neq u$ and $p_2 \neq u$.
\begin{itemize}
\item[$\bullet$] Insert $p_1$ into $Q$ between $u$ and $v$, and set $\flag(p_1) = \FALSE$ and $F=F+1$.
\item[$\bullet$] If $p_2$ is not a vertex of $Q$, then replace $Q$ by the convex hull of $Q \cup \{p_2\}$, and set $\flag(p_2) = \FALSE$ and $F=F+1$.
\end{itemize}

 Case 2: Either $p_1 = u$ or $p_2 = u$. 
\begin{itemize}
\item[$\bullet$] If $p_1 = u$, then set $\flag(u) = \TRUE$ and $F=F-1$; if $p_2 \notin Q$, then replace $Q$ by the convex hull of $Q \cup \{p_2\}$, and set $\flag(p_2) = \FALSE$ and $F=F+1$.
\item[$\bullet$] If $p_2 = u$, then set $\flag(v) = \TRUE$, insert $p_1$ into $Q$ between $u$ and $v$, and set $\flag(p_1) = \FALSE$. 
\end{itemize}      
\end{enumerate}  
\end{quotation}
    \vspace{-0.5em}
\rule{0.88\textwidth}{1pt}

\vspace{1.0em} 
In \emph{Case 1} the pair of vertices of interest $u$ and $v$ are not connected by an edge in $\obj$ and thus neither $p_1 = u$ nor $p_2 = u$. In \emph{Case 2} $\overline{uv}$ is an edge of $\obj$. So, one arm of the probe coincides with $\overrightarrow{L}$. This results in either $p_1 = u$ or $p_2 = u$ (the latter can happen at most once and only during the first probe after the initialization). Because $\omega > 0$ the case when $p_1 = u$ and $p_2 = u$  is impossible.

\vspace{0.5em}
Let us analyze the performance of the algorithm on \emph{Input 1}. It follows from the algorithm that the invariant is correctly maintained. The following theorem shows that the algorithm terminates.

\begin{theorem}
\label{theo:model_1_sum-up_theorem_upper_bound}
Given $\omega$-wedge, with $0 < \omega \leq \pi/2$, and a convex polygon $\obj$ (whose angles are strictly bigger than $\omega$), the above algorithm reconstructs $\obj$ by using at most $2n - 2$ $\omega$-probes. 
\end{theorem}
\begin{proof}
Let $n$ denote the number of vertices of $\Ob$. Consider the quantity 
\[ \Phi = 2|Q| - F , 
\]
where $|Q|$ denotes the number of vertices of $Q$. After initialization, we have $\Phi = 2$. After the first iteration of the algorithm, we have $\Phi = 4$. This follows from the fact that $\obj$ has no narrow vertices. In every subsequent iteration of the algorithm, the value of $\Phi$ increases by at least one. Notice that every vertex of $Q$ is a vertex of $\obj$. Since, at any moment, $\Phi \leq 2n$, it follows that the algorithm makes at most $2n-3$ iterations and thus, the algorithm terminates.  
After termination, we have $F=0$. It then follows from the invariant 
that $Q = \Ob$. Finally, the number of probes made by the algorithm together with the probe spent to create the input is at most $2n-2$.  
\qed
\end{proof}

\subsection{An improved analysis for $\omega = \pi/2$} 
\label{subsec:ImprovedAnalysis}
We now show how to reconstruct $\obj$ with one fewer probe when $\omega = \pi / 2$. The above probing algorithm with the \emph{Input 1} can be used successfully for $\omega = \pi / 2$. As we will show, in this particular case,  the input for the algorithm can be improved (by a complex initialization step) to save one probe during the whole reconstruction process. We present a modification of \emph{Input 1}, on which $\obj$ can be reconstructed by using at most $2n - 3$ $\omega$-probes. As before, we assume that all internal angles of the polygon $\obj$ are larger than $\omega = \pi / 2$. The initialization step is more complex than the initialization of \emph{Input 1} and may consist of two or three probes. The initialization step is followed by one special probe (which we call the \emph{hit probe}) that returns two new pieces of information about $\obj$: either two newly discovered vertices or one new vertex and one new edge of $\obj$. This is precisely where we save one probe compared to the previous technique. 

\begin{figure}[h]
    \begin{center}
        \subfigure[First case of the initialization step.]{
            \label{fig:omega90_alg_case1}
            \includegraphics[width=0.5\textwidth]{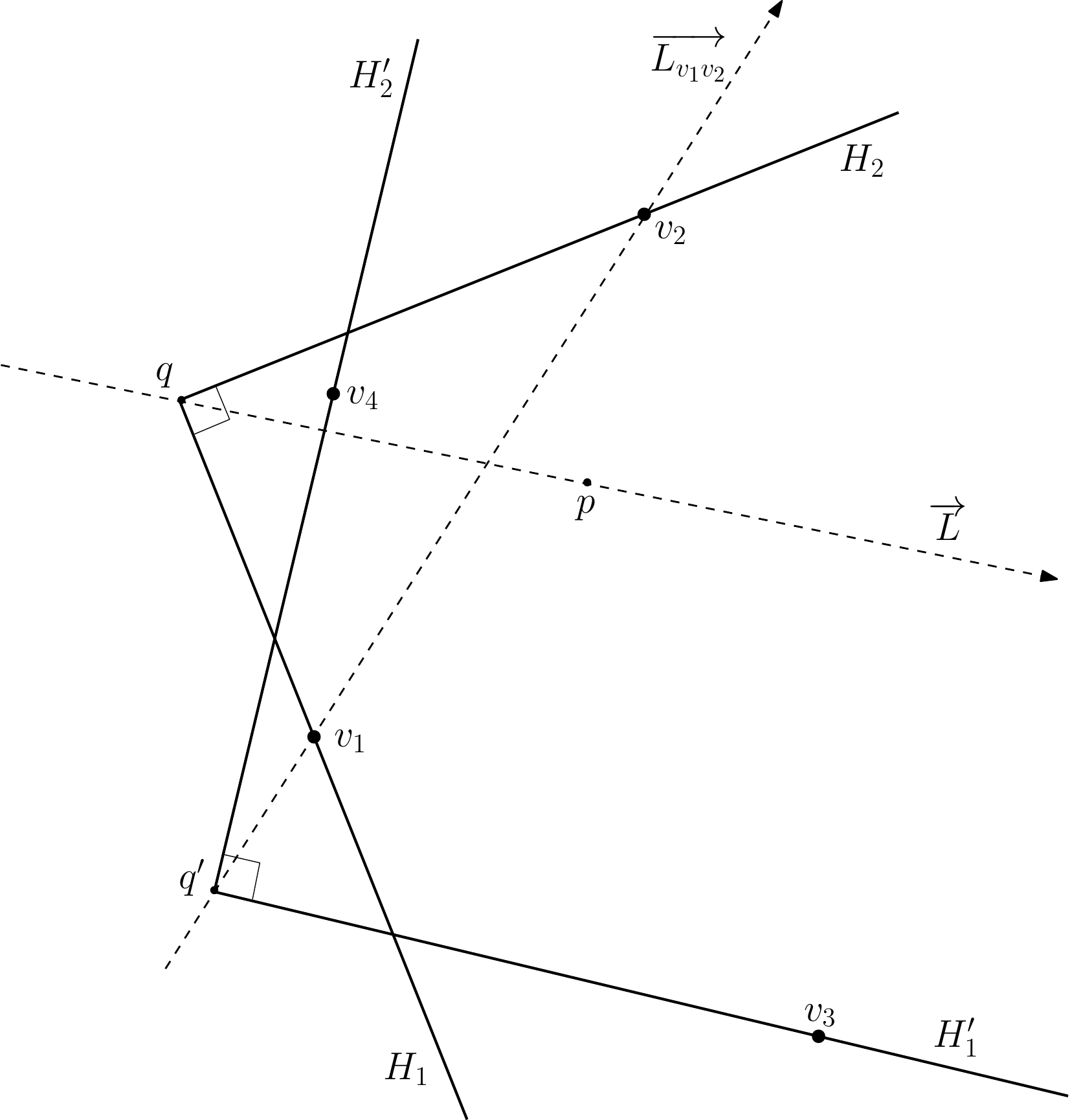}
        }%
        \hspace{0.4cm}%
        \subfigure[Second case of the initialization requires three probes. The figure shows $v = v_1$, but $v$ can be anywhere inside the black triangle.]{
            \label{fig:omega90_alg_case2}
            \includegraphics[width=0.45\textwidth]{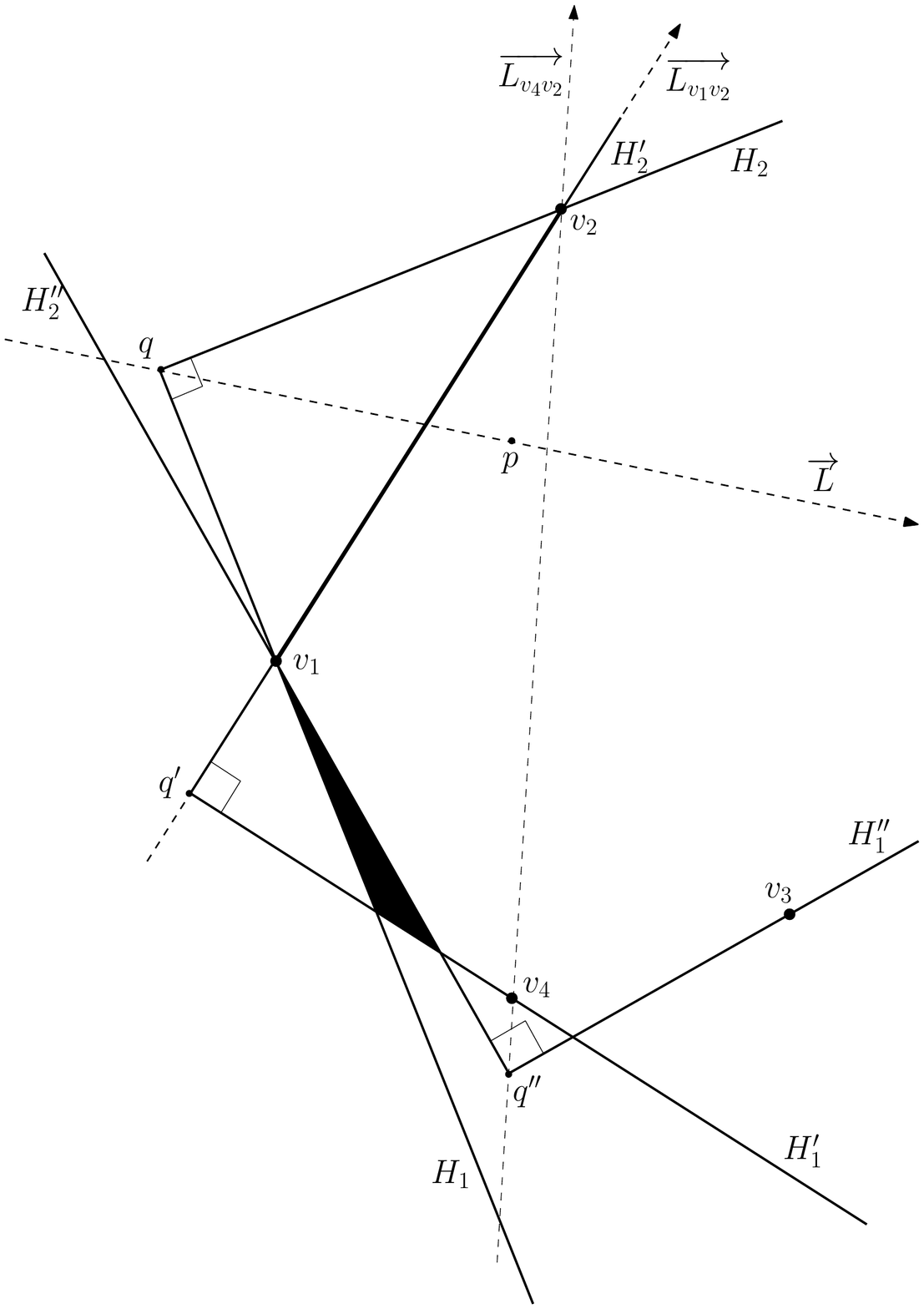}
        }%
        
    \end{center}
    \vspace{-1.5em}
    \caption{Initialization steps of the algorithm given $\omega = \pi / 2$.}

\label{fig:omega90_alg}
\end{figure}

\vspace{0.5em}

\rule{0.88\textwidth}{1pt}

{\bf Input 2:} 
\begin{quotation}
\noindent
{\bf Invariant:} Invariant of \emph{Input 1} with addition of $\omega = \pi / 2$.

\noindent
{\bf Initialization:} Choose an arbitrary line $\overrightarrow{L}$ that contains $p$. Let $(q,H_1,H_2,v_1,v_2)$ be the outcome of the probe along $\overrightarrow{L}$ (refer to Figure~\ref{fig:omega90_alg}). Shoot another probe along the line $\overrightarrow{L_{v_1 v_2}}$ (that is directed from $v_1$ to $v_2$). Since $\obj$ cannot be fully contained inside the triangle $\triangle (q, v_1, v_2)$, two different outcomes of this probe are possible:
\begin{itemize}
\item[$\star$] \textbf{Case $1$:} $(q',H'_1,H'_2,v_3,v_4)$ - the probe returns two new vertices $v_3$ and $v_4$. Refer to Figure~\ref{fig:omega90_alg_case1}. 

Initialize $Q$ as the polygon consisting of the vertices $v_1$, $v_2$, $v_3$ and $v_4$, set $\flag(v_1) = \FALSE$, $\flag(v_2) = \FALSE$, $\flag(v_3) = \FALSE$, $\flag(v_4) = \FALSE$ and $F = 4$.

\item[$\star$] \textbf{Case $2$:} $(q',H'_1,H'_2,v_4,v_1)$ - the arm $H'_2$ of the second probe is flush with the line $\overrightarrow{L_{v_1 v_2}}$ and thus confirms the edge $v_1 v_2$ of $\obj$. The other arm of the probe $H'_1$ reveals a new vertex $v_4$ of $\obj$. In this case, we proceed with a third probe (refer to Figure~\ref{fig:omega90_alg_case2}). We probe along the line $\overrightarrow{L_{v_4 v_2}}$ (directed from $v_4$ to $v_2$). The line segment $v_4 v_2$ cannot be an edge of $\obj$, otherwise $v_2$ would be a narrow vertex. So, the outcome $(q'',H''_1,H''_2,v_3,v)$ of the third probe returns a new vertex $v_3$ of $\obj$. The vertex $v$ can be a new vertex of $\obj$ or it can be an already discovered vertex $v_1$.  

Initialize $Q$ as the polygon consisting of the vertices $v_1$, $v_2$, $v_3$ and $v_4$, and set $\flag(v_1) = \FALSE$, $\flag(v_2) = \TRUE$, $\flag(v_3) = \FALSE$ and $\flag(v_4) = \FALSE$. If $v \neq v_1$, add $v$ to $Q$, set $\flag(v) = \FALSE$ and $F = 4$. Otherwise, if $v = v_1$, set $F = 3$.
\end{itemize}

\vspace{0.5em}
\noindent 
{\bf Hit Step:} 
We choose the direction $\overrightarrow{L}$ of the \emph{hit step} probe as follows:
\begin{itemize}
\item If the initialization step unfolded according to the first scenario (refer to Figures~\ref{fig:omega90_alg_case1} and~\ref{fig:omega90_alg_case1_hit_step}), then :
\begin{itemize}
\item If $\angle (v_2, v_3, q') < \angle (q, v_2, v_3)$, set $\overrightarrow{L} = \overrightarrow{ L_{v_3v_1}}$. 
\item If $\angle (v_2, v_3, q') \geq \angle (q, v_2, v_3)$, set $\overrightarrow{L} = \overrightarrow{ L_{v_2v_4}}$.
\end{itemize}
\item If the initialization step followed the second scenario and required three probes (refer to Figures~\ref{fig:omega90_alg_case2} and ~\ref{fig:omega90_alg_case2_hit_step}), then set $\overrightarrow{L} = \overrightarrow{ L_{v_3v_4}}$.
\end{itemize} 

\vspace{0.5em}
\noindent
Let ($q^*$, $H^*_1$, $H^*_2$, $p^*_1$, $p^*_2$) be the outcome of the \emph{hit step} probe. Three cases are possible:
\begin{itemize}
\item[$\star$] \textbf{Case 1:} $\overrightarrow{L} = \overrightarrow{ L_{v_3v_1}}$. Add $p^*_1$ to $Q$, set $\flag(p^*_1) = \FALSE$; 

if $p^*_2 = v_3$, then set $\flag(v_1) = \TRUE$; 

if $p^*_2 \neq v_3$, then add $p^*_2$ to $Q$, set $\flag(p^*_2) = \FALSE$ and $F = 6$.
\item[$\star$] \textbf{Case 2:} $\overrightarrow{L} = \overrightarrow{ L_{v_2v_4}}$. Add $p^*_2$ to $Q$, set $\flag(p^*_2) = \FALSE$; 

if $p^*_1 = v_2$, then set $\flag(v_2) = \TRUE$; 

if $p^*_1 \neq v_2$, then add $p^*_1$ to $Q$, set $\flag(p^*_1) = \FALSE$ and $F = 6$.
\item[$\star$] \textbf{Case 3:} $\overrightarrow{L} = \overrightarrow{ L_{v_3v_4}}$. Add $p^*_1$ to $Q$, set $\flag(p^*_1) = \FALSE$; 

if $p^*_2 = v_3$, then set $\flag(v_4) = \TRUE$; 

if $p^*_2 \neq v_3$, then add $p^*_2$ to $Q$, set $\flag(p^*_2) = \FALSE$ and $F = F+2$.
\end{itemize}
\end{quotation}
\rule{0.88\textwidth}{1pt}

\begin{lemma}[Justification of initialization steps]
\label{lemma:not_an_edge}
Given a convex polygon $\obj$ with $n\geqslant 5$ vertices and no narrow vertices; $\omega = \pi / 2$; polygon $Q = \{v_1, v_2, v_3, v_4$, and possibly $v \}$ of vertices of $\obj$ discovered during the initialization step of \emph{Input $2$}, the line segment $v_3v_2$ is \textbf{not} an edge of $\obj$.  
\end{lemma}
\begin{proof}
We prove this lemma by contradiction. Let us assume that the line segment $v_3v_2$ is an edge of $\obj$. We consider the two outcomes of the probe along the line $\overrightarrow{L_{v_1 v_2}}$ (the second probe of the initialization step) separately.

\begin{enumerate}
\item Assume that we are in the first case of the initialization (refer to Figure~\ref{fig:omega90_alg_case1}). The $\omega$-wedges at $q$ and $q'$ are valid $\omega$-probes of $\obj$. Thus the vertices of $\obj$ belong to the intersection of these wedges. By construction, the probe with apex $q'$ was shot along the line $\overrightarrow {L_{v_1v_2}}$. Moreover, $\obj$ has no narrow vertices. Therefore, $q \neq v_1$, $q \neq v_2$, $q' \neq v_1$ (and thus $q \neq q'$). The line segment $qq'$ is disjoint with the interiors of the $\omega$-wedges at $q$ and $q'$, otherwise, $v_1 = q'$. Therefore, the angles $\angle (v_3, q', q)$ and $\angle (q', q, v_2)$ are strictly bigger than $\pi / 2$. 

Since $\obj$ has no narrow vertices, the internal angles of $\obj$ at $v_2$ and $v_3$ are bigger than $\pi / 2$. Since $v_3v_2$ is an edge of $\obj$, the angles $\angle (q, v_2, v_3)$ and $\angle (v_2, v_3, q')$ are strictly bigger than $\pi / 2$. Consider the quadrilateral $\lbrace q,q',v_3,v_2 \rbrace$. All its internal angles are bigger than $\pi / 2$, which is a contradiction. Therefore, $v_3v_2$ is not an edge of $\obj$. 

\item Assume that we are in the second case of the initialization (refer to Figure~\ref{fig:omega90_alg_case2}). Note that, in this case, the arm $H_2''$ of the third probe (with apex $q''$) contains the vertex $v$ of $\obj$, that may or may not be equal to $v_1$. The vertex $v$ belongs to the intersection of the $\omega$-wedges at $q$ and $q'$ and is to the left of the line $\overrightarrow{ L_{q''v_1}}$ or on it. The region where $v$ belongs appears in black in Figure~\ref{fig:omega90_alg_case2}; the arm $H_2''$ intersects this area. Because $\obj$ has no narrow vertices, $v_1 \neq q$, $v_1 \neq q'$ and $v_4 \neq q''$. Therefore, the line segment $q'q''$ is disjoint with the interiors of the $\omega$-wedges at $q'$ and $q''$. Thus, the angles $\angle (v_3, q'', q')$ and $\angle (q'', q', v_2)$ are strictly bigger than $\pi / 2$. 

Since $\obj$ has no narrow vertices and $v_1v_2$ is an edge of $\obj$, the angles $\angle (q', v_2, v_3)$ and $\angle (v_2, v_3, q'')$ are strictly bigger than $\pi / 2$. Consider the quadrilateral $\lbrace q',q'',v_3,v_2 \rbrace$. All its internal angles are bigger than $\pi / 2$, which is a contradiction. Therefore, $v_3v_2$ cannot be an edge of $\obj$. 
\qed
\end{enumerate}
\end{proof}

\begin{figure}[h]
    \begin{center}
        \subfigure[\textit{Hit Step} after the first case of the initialization.]{
            \label{fig:omega90_alg_case1_hit_step}
            \includegraphics[width=0.5\textwidth]{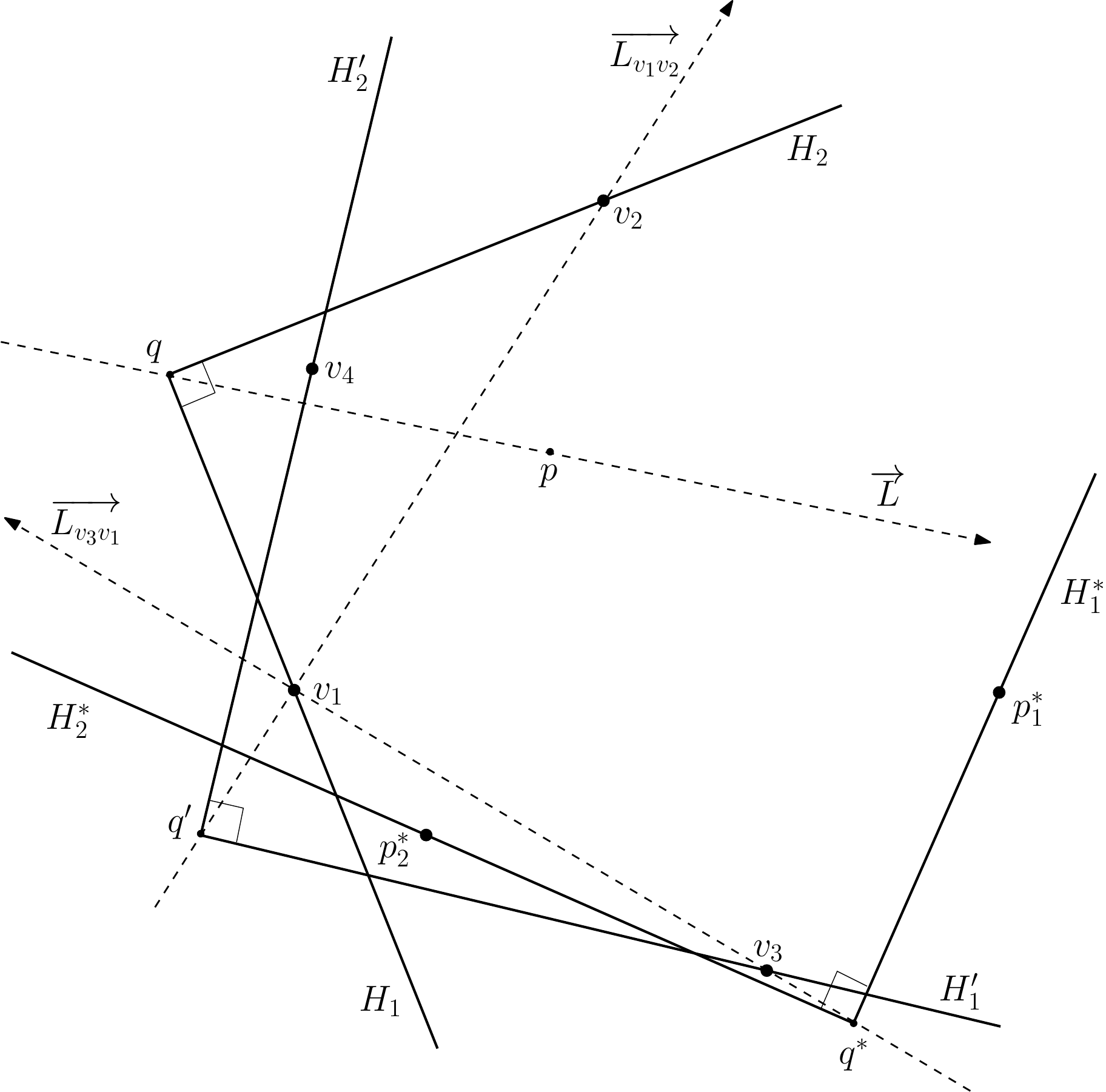}
        }%
        \hspace{0.3cm}%
        \subfigure[\textit{Hit Step} after the second case of the initialization. The vertex $v$ is shown to be equal to the vertex $v_1$.]{
            \label{fig:omega90_alg_case2_hit_step}
            \includegraphics[width=0.45\textwidth]{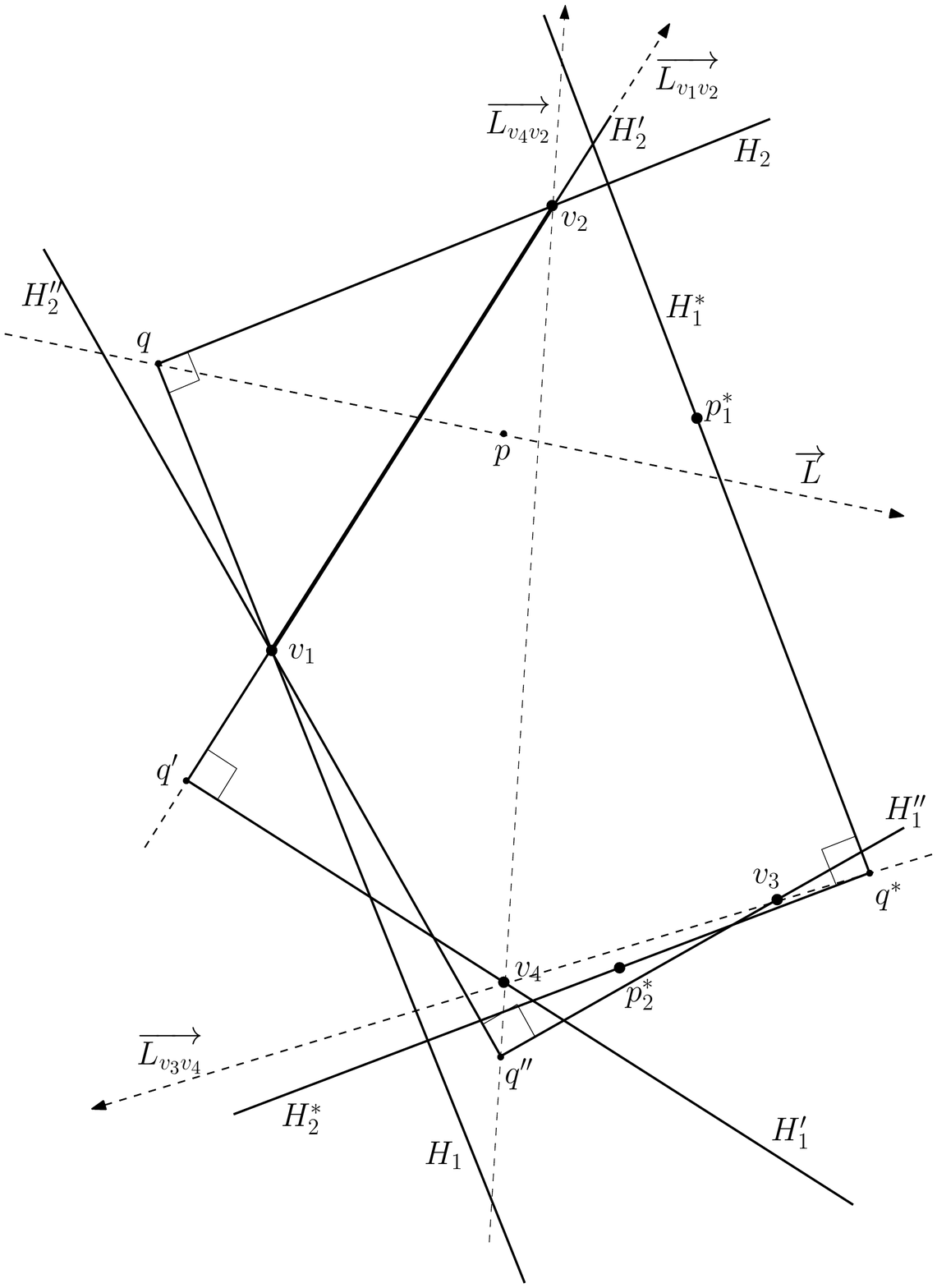}
        }%
    \end{center}
    \vspace{-1.5em}
    \caption{Possible outcomes of the \textit{Hit Step} of the Input $2$.}
\label{fig:omega90_alg_hit_step}
\end{figure}

\begin{lemma}[Hit Step justification]
\label{lemma:hit_step}
Given a convex polygon $\obj$ with $n\geqslant 5$ vertices and no narrow vertices; $\omega = \pi / 2$; polygon $Q = \{v_1, v_2, v_3, v_4$, and possibly $v \}$, of vertices of $\obj$ discovered during the initialization step of \emph{Input $2$}. The \textit{Hit Step} of Input $2$ (that consists of one probe only), results in \textbf{two} new pieces of information: either two newly discovered vertices of $\obj$, or one newly discovered vertex and one newly confirmed edge of $\obj$. 
\end{lemma}
\begin{proof}
The \textit{Hit Step} of \emph{Input $2$} consists of one probe only. We choose the direction of this probe to be one of the following two:
\begin{enumerate}
\item $\overrightarrow{ L_{v_3x_1}}$, where $x_1$ is a clockwise neighbour
 of $v_3$ in $Q$.
\item $\overrightarrow{ L_{v_2x_2}}$, where $x_2$ is a counter-clockwise neighbour of $v_2$ in $Q$ (refer to Figure~\ref{fig:omega90_alg}).
\end{enumerate}

\vspace{0.5em}
There are two possible outcomes of the probe along the line $\overrightarrow{L_{v_1 v_2}}$ (the second probe of the initialization step of Input $2$). We consider those cases separately.

\begin{itemize}
\item[$\bullet$] Case 1: Assume, that the probe along $\overrightarrow{L_{v_1 v_2}}$ returns two new vertices $v_3$ and $v_4$ (refer to Figure~\ref{fig:omega90_alg_case1}). We showed in the proof of Lemma~\ref{lemma:not_an_edge}, that the angles $\angle (v_3, q', q)$ and $\angle (q', q, v_2)$ are strictly bigger than $\pi / 2$. Thus, in the quadrilateral $\lbrace q,q',v_3,v_2 \rbrace$ at least one of the angles $\angle (v_2, v_3, q')$ or $\angle (q, v_2, v_3)$ is strictly smaller than $\pi / 2$. Consider the smallest one, which is smaller than $\pi / 2$. If $\angle (v_2, v_3, q') < \angle (q, v_2, v_3)$, then we choose $\overrightarrow{ L_{v_3v_1}}$ as a direction for the \textit{Hit Step} probe. Otherwise, we choose $\overrightarrow{ L_{v_2v_4}}$.

We focus on the case where $\angle (v_2, v_3, q') < \angle (q, v_2, v_3)$. The other case is similar. We shoot our probe along the line $\overrightarrow{ L_{v_3v_1}}$. Let ($q^*$, $H^*_1$, $H^*_2$, $p^*_1$, $p^*_2$) be the outcome of the probe. Refer to Figure~\ref{fig:omega90_alg_case1_hit_step}. Since $v_3$ is not a narrow vertex, $q^* \neq v_3$. Similarly to what was shown in the proof of Lemma~\ref{lemma:not_an_edge}, the line segment $q' q^*$ is disjoint with the interiors of the $\omega$-wedges at $q'$ and $q^*$, leading to the conclusion that the angles $\angle (p^*_1, q^*, q')$ and $\angle (q^*, q', v_4)$ are strictly bigger than $\pi / 2$. The arm $H^*_1$ contains the vertex $p^*_1$ of $\obj$ and it is to the right of the line $\overrightarrow{ L_{v_3v_1}}$ (refer to Definition~\ref{def:probe_along_L}). The vertex $p^*_1$ satisfies the following properties:

\vspace{0.5em}
\begin{enumerate}
\item $p^*_1 \neq v_3$. Otherwise, the vertices $v_2$ and $v_4$ would not belong to the interior of the wedge at $q^*$, which is a contradiction to the validity of the probe. Similarly, $p^*_1 \neq v_1$.
\item $p^*_1 \neq v_2$.  We showed previously, that $\angle (p^*_1, q^*, q') > \pi / 2$. Assume, to the contrary, that $p^*_1 = v_2$. Thus, $\angle (v_2, q^*, q') > \pi / 2$. The vertex $v_3$ is inside the wedge induced by $q'$, $q^*$ and $v_2$ according to the construction. Therefore, the angle $\angle (v_2, v_3, q')$ is bigger than $\angle (v_2, q^*, q')$. This is a contradiction because $\angle (v_2, v_3, q') < \pi / 2$.
\item $p^*_1 \neq v_4$. Otherwise, the line segment $q' q^*$ would be completely inside the interior of the $\omega$-wedges at $q'$ and $q^*$, which is impossible according to the construction.
\end{enumerate}
Consequently, $p^*_1$ is a newly discovered vertex of $\obj$.

\vspace{0.5em}
The other arm $H^*_2$ of the probe contains the vertex $p^*_2$ of $\obj$ and it is to the left of the line $\overrightarrow{ L_{v_3v_1}}$ or on it (by Definition~\ref{def:probe_along_L}). If $p^*_2 \in \overrightarrow{ L_{v_3v_1}}$ then the arm $H^*_2$ coincides with the line $\overrightarrow{ L_{v_3v_1}}$ and as a result the edge $v_1v_3$ of $\obj$ is confirmed. If $p^*_2 \notin \overrightarrow{ L_{v_3v_1}}$ then $p^*_2$ must be a newly discovered vertex of $\obj$, since there were no vertices of $\obj$ discovered so far to the left of $\overrightarrow{ L_{v_3v_1}}$. In either way, the arm $H^*_2$ reveals new information about $\obj$.

\vspace{0.5em}
We conclude, that the probe of the \textit{Hit Step}, that follows the first case of the initialization step, contributes two new pieces of information.

\item[$\bullet$] Case 2: Assume, that the probe along $\overrightarrow{ L_{v_1 v_2}}$ confirms the edge $v_1 v_2$ of $\obj$ (refer to Figure~\ref{fig:omega90_alg_case2}). In this case the direction for the \textit{Hit Step} probe will be $\overrightarrow{ L_{v_3v_4}}$ (refer to Figure~\ref{fig:omega90_alg_case2_hit_step}). Let ($q^*$, $H^*_1$, $H^*_2$, $p^*_1$, $p^*_2$) be the outcome of this probe. Similarly to the previous case, we can show that the line segment $q'' q^*$ is disjoint with the interiors of the $\omega$-wedges at $q''$ and $q^*$, leading to the following conclusion: $\angle (p^*_1, q^*, q'') \geq \pi / 2$ and $\angle (q^*, q'', v) \geq \pi / 2$. Note, that in Figure~\ref{fig:omega90_alg_case2_hit_step}, the vertex $v$ is shown to be equal to the vertex $v_1$. Let $t$ be the intersection point between $H'_2$ and $H''_2$. Note, that if $v = v_1$, then  $t =  v_1$. The angle $\angle (q'', t, v_2) > \pi / 2$, otherwise $q''$ would be inside the $\omega$-wedge at $q'$ together with the line segment $q' q''$, which is a contradiction (refer to the proof of the Lemma~\ref{lemma:not_an_edge}).

The arm $H^*_1$ contains the vertex $p^*_1$ of $\obj$ and it is to the right of $\overrightarrow{ L_{v_3v_4}}$ (refer to Definition~\ref{def:probe_along_L}). The vertex $p^*_1$ satisfies the following properties:

\begin{enumerate}
\item $p^*_1 \neq v_3$. Otherwise, the vertices $v_2$, $v_1$ and $v$ would not belong to the interior of the wedge at $q^*$, which is a contradiction to the validity of the probe. Similarly, $p^*_1 \neq v_4$.
\item $p^*_1 \neq v_2$.   Assume, to the contrary, that $p^*_1 = v_2$. Consider the quadrilateral $\lbrace q'', q^*, v_2, t \rbrace$. We showed before, that $\angle (v_2, q^*, q'') = \angle (p^*_1, q^*, q'') \geq \pi / 2$, $\angle (q^*, q'', t) = \angle (q^*, q'', v) \geq \pi / 2$ and $\angle (q'', t, v_2) > \pi / 2$. This implies that the fourth internal angle $\angle (t, v_2, q^*) = \angle (v_1, v_2, q^*)$ of the quadrilateral must be smaller than $\pi / 2$. All the probes made so far are valid $\omega$-probes of $\obj$ and thus $\obj$ should reside inside the intersection of all $\omega$-wedges. Thus, the internal angle of $\obj$ at $v_2$ is inside the wedge induced by $q^*$, $v_2$ and $v_1$ and therefore it is smaller than $\pi / 2$. This is a contradiction, because $v_2$ is not a narrow vertex.

\item $p^*_1 \neq v$. Otherwise, the line segment $q'' q^*$ is completely inside the interiors of the $\omega$-wedges at $q''$ and $q^*$, which is impossible according to the construction.

\item $p^*_1 \neq v_1$.  The $\omega$-wedge at $q^*$ is a valid $\omega$-probe and thus must contain the edge $v_1v_2$. In order for the arm $H^*_1$ to contain the vertex $v_1$, the apex of the probe $q^*$ should be to the left of $\overrightarrow{ L_{v_1v_2}}$, which is impossible according to the construction.

\end{enumerate}
Consequently, $p^*_1$ is a newly discovered vertex of $\obj$.
Similarly to the previous case, the arm $H^*_2$ reveals new information about $\obj$ (it may be an edge or a vertex). We showed that the probe of the \textit{Hit Step}, that follows the second case of the initialization step, contributes two new pieces of information. 
\qed
\end{itemize}
\end{proof}

Let us analyze the performance of the algorithm on Input $2$. It follows from the algorithm that the invariant is correctly maintained. The following theorem shows that the algorithm terminates.

\begin{theorem}
\label{theo:model_1_sum-up_theorem_upper_bound_90}
Given $\omega$-wedge such that $\omega = \pi / 2$ and a convex polygon $\obj$ (whose angles are strictly bigger than $\omega$), the algorithm on Input $2$ reconstructs $\obj$ using at most $2n - 3$ $\omega$-probes. 
\end{theorem}
\begin{proof}
Let $n$ denote the number of vertices of $\Ob$. Consider the quantity 
\[ \Phi = 2|Q| - F , 
\]
where $|Q|$ denotes the number of vertices of $Q$. At the start of the 
algorithm, we have $\Phi = 2$.  At the end of the initialization step, if we are in the first case (the probe along $\overrightarrow{L_{v_1 v_2}}$ returns two new vertices), then $|Q| = 4$, $F = 4$ and thus $\Phi = 4$. Otherwise,
\begin{itemize}
\item if $v = v_1$, then $|Q| = 4$, $F = 3$ and thus $\Phi = 5$.
\item if $v \neq v_1$, then $|Q| = 5$, $F = 4$ and thus $\Phi = 6$.
\end{itemize}
The \textit{Hit Step} increases $\Phi$ by two. The \textit{hit probe} either discovers two new vertices (and thus $|Q|$ and $F$ increase by two), or one new vertex and one new edge ($|Q|$ increases by one and $F$ does not change).
In every iteration of the algorithm, the value of $\Phi$ increases by at least one. Notice, that every vertex of $Q$ is a vertex of $\obj$. Since, at any moment, $\Phi \leq 2n$, it follows that the algorithm makes at most
\begin{itemize}
\item[$\bullet$] $2n-6$ iterations, if the initialization unfolded according to the first case, and
\item[$\bullet$] $2n-7$ iterations, if the initialization followed the second case (at most $2n-7$ iterations if $v = v_1$, and at most $2n-8$ iterations if $v \neq v_1$).
\end{itemize} 
Thus, the algorithm terminates. After termination, we have $F=0$. It then follows from the invariant that, at that moment, $Q$ is equal to $\Ob$. Finally, the number of probes made by the overall reconstruction process is at most $2n-3$.  
\qed
\end{proof}

\subsection{Lower Bound}
\label{subsec:lower_bound_model1}

In this section, we prove a lower bound on the number of probes required to reconstruct an $n$-gon that is equal to the number of probes used by our algorithm (refer to Theorem~\ref{theo:model_1_sum-up_theorem_upper_bound} and~\ref{theo:model_1_sum-up_theorem_upper_bound_90}) in the worst case. We present an adversarial argument that forces any probing strategy to make a given number of probes to determine the exact shape and orientation of $\obj$. Our goal is to show that in order to be correct, any algorithm must confirm every edge and probe for every vertex of $\obj$. Given any sequence of probes by an algorithm, we show that the adversary can force the algorithm to spend at least $2n-2$ probes for any $\omega$ in the range $0 < \omega < \pi/2$, and at least $2n-3$ probes for $\omega = \pi/2$. We outline the adversary's strategy below.

\vspace{0.5em}
Consider an algorithm that wants to reconstruct an object $\obj$ that is ``hidden" by the adversary. The algorithm provides a direction of its next probe to the adversary. If the direction is valid, the adversary returns an output that is coherent with the information already given to the algorithm. This output consists of: two points of contact with $\obj$, the orientation of both arms and the coordinates of the apex. For a probe that is not valid, the adversary replies that the probe missed the object and gives no additional information. The strategy of the adversary is to answer each probe by revealing new information only if it is unavoidable. We showed in Lemma~\ref{lem:feasible_triangle} that as long as an edge $e$ of $Q \subseteq \obj$ is not confirmed (by a probe directed through that edge), the feasible region of $F_{Q,e}$ is not empty and does not degenerate into a line segment, meaning that the adversary has infinitely many points inside that region where it can place a vertex to deny the existence of the edge $e$. Similarly, when the algorithm aims for some place where, according to its strategy, a vertex of $\obj$ is most likely to be hidden, the adversary orients the arms of the probe to ``shrink" the feasible region of that place. The probe then reveals only known vertices and, possibly confirms an edge. For directions that do not intersect $Q$, the adversary lets the probe miss the object completely. 
While in some cases it is unavoidable for the adversary to reveal new information about $\obj$, it is always possible to reveal only one new piece of information per probe (except for the first two valid probes).

\vspace{0.5em}
We show that for every algorithm, there is a convex polygon, such that, except for the first two probes, in each probe the algorithm either discovers a new vertex or verifies an edge or does not discover any new information. So, at least $2n-2$ $\omega$-probes are necessary to determine a convex $n$-gon when $0 < \omega < \pi /2$. When $\omega = \pi /2$, the adversary forces at least $2n-3$ $\omega$-probes. 

We assume that the algorithm is deterministic and does not repeat the same probe during one probing session. 

\vspace{0.5em} 
The adversary begins by defining a circle $\Psi$ (where the polygon $\obj$ will reside), a point $p$ as the center of $\Psi$ and chooses $n \geq 4$ ($n \geq 5$ when $\omega = \pi /2$) to avoid narrow vertices. The adversary maintains a closed convex curve $A$ that is initially the circle $\Psi$. During queries the adversary will change the shape of the curve $A$ such that $A$ is the boundary of the intersection of all $\omega$-probes made so far with $\Psi$. Once the adversary reveals a vertex of the polygon $\obj$, it remains fixed on $A$, defining sections of the curve that cannot be changed. So, at the end of a query session the curve $A$ becomes the convex polygon $\obj$.

Let $\overrightarrow{L_x}$ be a direction of the probe, such that $\overrightarrow{L_x}$ does not intersect $\obj$. The probe along this line is considered to be non-valid (refer to the Definition~\ref{def:probe_along_L}) and reported as such to the algorithm. But, if $\overrightarrow{L_x}$ intersects $A$, we update $A$ to be the intersection of $A$ and the half-plane, that contains $\overrightarrow{L_x}$ on its boundary and the point $p$. 

\vspace{0.5em}
Notice that, during the construction of the polygon, the adversary avoids to locate vertices at positions that allow existence of two edges, whose extensions form an angle of $\omega$ (refer to Figure~\ref{fig:points_of_contact_2}). Let $Q$ be the convex hull of the vertices of $\obj$ that have been discovered, and let $e = (v_1, v_2)$ be an edge of $Q$ but not an edge of $\obj$ (meaning that the interior of $e$ is not on the boundary of $A$ and thus $F_{Q,e} \setminus l_e \neq \emptyset$ , where $F_{Q,e}$ is the feasible region (see Definition~\ref{def:feasible}) and $l_e$ is the line that contains $e$). To create a new vertex $v$ between $v_1$ and $v_2$ on the boundary of $\obj$, the adversary chooses the position for $v$ in $F_{Q,e} \setminus l_e$, such that for every edge $e_Q$ of $Q$ (but not $e$): $\angle (l_{e_Q}, l_{v_1, v}) \neq \omega$ and $\angle (l_{e_Q}, l_{v_2, v}) \neq \omega$ (where $l_{e_Q}$  is the line that contains $e_Q$ and $l_{v_1, v}$ (respectively $l_{v_2, v}$) is the line that contains $v_1$ and $v$ (respectively $v_2$ and $v$)). It is always possible because $Q$ has a finite number of edges (this number is smaller than $n$) and there are infinitely many positions for $v$ inside $F_{Q,e} \setminus l_e$. Refer to Lemma~\ref{lem:feasible_triangle}. 

\vspace{0.5em} 
The strategy of the adversary can be described by four different stages. The initial stage lasts until the algorithm discovers four distinct vertices $v_1$, $v_2$, $v_3$ and $v_4$ of $\obj$. We assumed at the beginning that $\obj$ has no narrow vertices. This means that all the internal angles of the polygon $\obj$ should be larger than $\omega$. To comply with this assumption, the adversary has to place $v_1$, $v_2$, $v_3$ and $v_4$ in such a way that there exists a convex polygon $\obj$ whose vertices include $v_1$, $v_2$, $v_3$ and $v_4$, that has no internal angle smaller or equal to $\omega$ (refer to the Figure~\ref{fig:adversary_initialization}). Any algorithm will require at least two probes to discover the vertices $v_1$, $v_2$, $v_3$ and $v_4$.

Assume that the direction of the first valid probe $\overrightarrow{L}$ made by the algorithm contains $p$. Otherwise, the adversary can make an object small enough for the probe to miss the object. (If there were non-valid probes that intersected $A$, then the adversary would create a circle inside $A$ with center $p$ and restricts its working space to this circle. For simplicity, we call this circle $\Psi$ again and update $A$ to be equal to $\Psi$.) The adversary stops the probe at any moment after the apex $q$ of the probe enters the interior of $\Psi$, but before it reaches $p$. Let $(q,H_1,H_2,p_1,p_2)$ be the outcome of the probe along $\overrightarrow{L}$ returned by the adversary, such that $H_1$ (respectively $H_2$) makes a negative (respectively positive) angle of $\omega / 2$ with $\overrightarrow{L}$ (refer to Figure~\ref{fig:adversary_initialization}). 

\begin{figure}[h]
    \begin{center}
        \subfigure[Initial steps of adversary's strategy for $0 < \omega < \pi /2$.]{
            \label{fig:square_initialization}
            \includegraphics[width=0.5\textwidth]{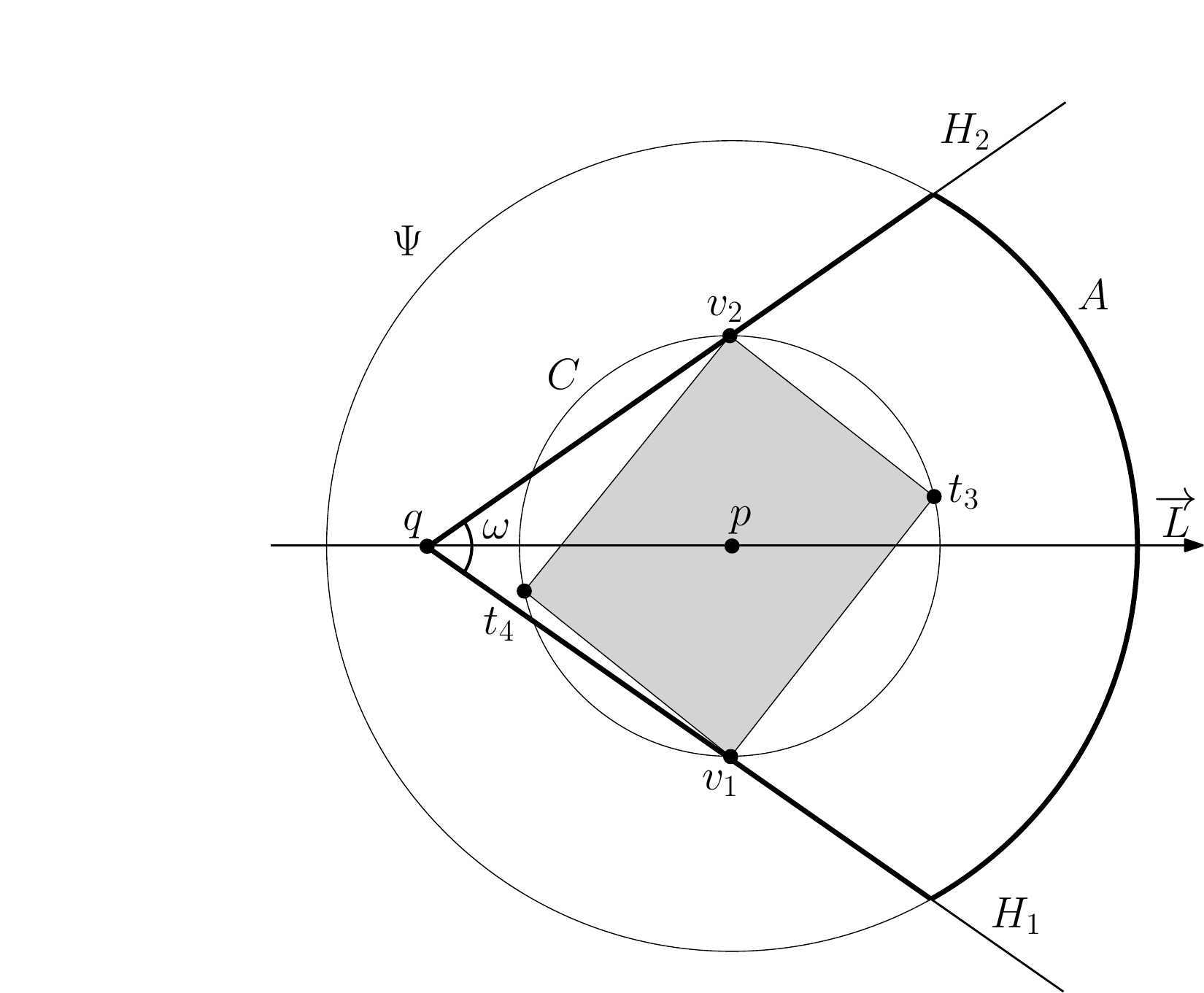}
        }%
        \hspace{0.3cm}%
        \subfigure[Initial steps of the adversary, given $\omega = \pi /2$.]{
            \label{fig:pentagon_initialization}
             \includegraphics[width=0.38\textwidth]{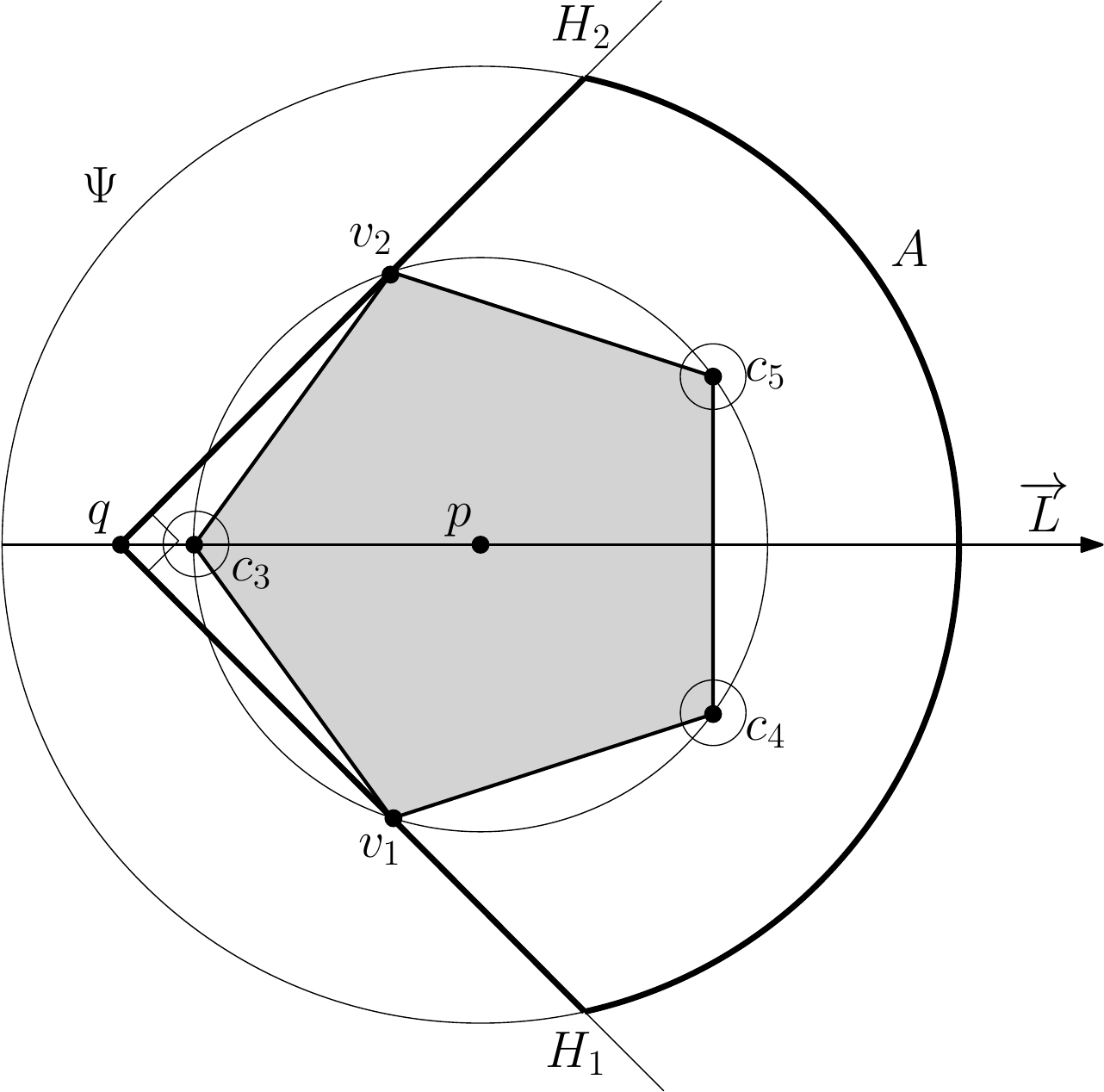}
        }%
    \end{center}
    \vspace{-1.0em}
    \caption{Initialization. After the first valid $\omega$-probe, the points $v_1$ and $v_2$ are fixed. The convex curve $A$ is shown in bold.}
\label{fig:adversary_initialization}
\end{figure}

For now, assume that $0 < \omega < \pi /2$. Given $\overrightarrow{L}$, we define an $\omega$-wedge in such a way that there exists a rectangle $\{v_1, t_3, v_2, t_4 \}$ in the interior of the $\omega$-wedge such that it is inscribed into a circle $C$ with the center at $p$ and vertices $v_1$ and $v_2$ belong to the arms of the probe (without loss of generality assume that $v_1 \in H_1$ and $v_2 \in H_2$). The point $p$ is the midpoint of the line-segment $v_1v_2$ (refer to Figure~\ref{fig:square_initialization}). The adversary returns $p_1 = v_1$ and $p_2 = v_2$ as points of contact. 

The vertices $v_1$ and $v_2$ of the rectangle are fixed. The vertices $t_3$ and $t_4$ are not fixed yet. They can be anywhere in the interior of $A$ (excluding boundary), as long as all the internal angles of the quadrilateral $\{v_1$, $t_3$, $v_2$, $t_4 \}$ are equal to $\pi /2$. Thus, the vertices $t_3$ and $t_4$ have temporary positions. If the direction of the next probe contains any of the four edges of the rectangle $\{v_1$, $t_3$, $v_2$, $t_4 \}$, the points $t_3$ and $t_4$ move to avoid confirming an edge. If the probe reveals $t_3$ (respectively $t_4$), it becomes fixed and we rename it into $v_3$ (respectively $v_4$). If the probe reveals only one temporary vertex $t_3$ or $t_4$, then we only rename the revealed vertex, but \textbf{fix both vertices}. Then the probe outputs only the vertices it found (in this case the output will include vertex $v_1$ or $v_2$ already known by the algorithm) leaving the fourth, not revealed vertex of the rectangle, fixed by the adversary but still hidden from the algorithm. The algorithm, of course, can guess its position now. This concludes the initialization step for $0 < \omega < \pi /2$.

\vspace{0.5em}
Let us explore the case $\omega = \pi /2$. Given $\overrightarrow{L}$, we define an $\omega$-wedge in such a way that there exists a regular pentagon $\{v_1, c_4, c_5, v_2, c_3\}$ in the interior of the $\omega$-wedge such that it is inscribed into a circle with center at $p$ and exactly two non-adjacent vertices $v_1$ and $v_2$ belong to the arms of the probe (without loss of generality assume that $v_1 \in H_1$ and $v_2 \in H_2$) (refer to Figure~\ref{fig:pentagon_initialization}). The adversary returns $p_1 = v_1$ and $p_2 = v_2$ as points of contact. 

The vertices $v_1$ and $v_2$ of the pentagon are fixed. As in the previous case, the other three vertices can move (within some constraints) to avoid edge confirmation during probing. We create a circle $C_3$ (respectively, $C_4$, $C_5$) with center $c_3$ (respectively, $c_4$, $c_5$) and diameter $\frac{1}{2\sqrt{2}}dist(v_1,v_2)(1-\tan \frac{\pi}{5})$, where $dist(v_1,v_2)$ is the distance between $v_1$ and $v_2$. The vertex $v_3$ (respectively, $v_4$, $v_5$) is fixed during later probes of the algorithm and stays in the interior of $C_3$ (respectively, $C_4$, $C_5$). For any point $p_3'$ (respectively, $p_4'$, $p_5'$) within the interior of $C_3$ (respectively, $C_4$, $C_5$) all the interior angles of the pentagon $\{v_1, p_4', p_5', v_2, p_3'\}$ are strictly bigger than $\dfrac{\pi}{2}$, meaning that no narrow vertex can be created. None of the circles degenerates into a point since their diameter is strictly bigger than zero.

The adversary maintains an $\omega$-cloud $\Omega$ of the pentagon $\{v_1, c_4, c_5, v_2, c_3\}$. If, for example, the direction of the probe enters the cloud through the ark that is supported by $c_3$ and $c_4$, then the adversary fixes two new vertices $v_3$ and $v_4$, such that $v_3 = c_3$ and $v_4 = c_4$. Every query of the algorithm will be answered by the adversary (until the algorithm discovers two new vertices) considering $\Omega$, unless the apex of the probe coincides with the pivot point on the cloud (refer to Definition~\ref{def:omega_cloud}), meaning that the direction of the probe coincides with one of the edges of the pentagon. Several cases are possible in this situation:

\begin{enumerate}
\item Assume that only $v_1$ and $v_2$ are fixed and that the vertices $v_3$, $v_4$ and $v_5$ are unknown. There are several cases to consider. The direction of the probe $\overrightarrow{L}$ may contain:
\begin{itemize}
\item $\overrightarrow{v_1c_3}$ (respectively, $\overrightarrow{v_1c_4}$) - choose two vertices $v_3$ and $v_4$, such that $v_3 \in C_3$, $v_4 \in C_4$; $v_3$ and $v_4$ are in the interior of $A$ and $v_3 \notin \overrightarrow{L}$ (respectively, $v_4 \notin \overrightarrow{L}$). Delete $c_3$, $C_3$, $c_4$ and $C_4$, and update $\Omega$. As the two points of contact return $v_1$ and $v_3$ (respectively, $v_1$ and $v_4$).
\item $\overrightarrow{v_2c_3}$ (respectively, $\overrightarrow{v_2c_5}$) - similar to the previous case. 
\item $\overrightarrow{c_3v_1}$ (respectively, $\overrightarrow{c_3v_2}$) - choose a vertex $v_3$ furthest from the segment $v_1v_2$, such that $v_3 \in C_3$, $v_3$ is in the interior of $A$ and $v_3 \notin \overrightarrow{L}$. Delete $c_3$ and $C_3$, and update $\Omega$. Return $v_3$ and $v_1$ (respectively, $v_3$ and $v_2$) as the two points of contact.
\item $\overrightarrow{c_4c_5}$ - choose a vertex $v_4$ furthest from the segment $v_1c_5$, such that $v_4 \in C_4$, $v_4$ is in the interior of $A$ and $v_4 \notin \overrightarrow{L}$. Delete $c_4$ and $C_4$, and update $\Omega$. Return $v_4$ and $v_1$ as the two points of contact.
\item $\overrightarrow{c_5c_4}$ - similar to the previous case.
\item $\overrightarrow{c_4v_1}$ (respectively, $\overrightarrow{c_5v_2}$) - choose two vertices $v_4$ and $v_5$, such that $v_4 \in C_4$, $v_5 \in C_5$, $v_4$ and $v_5$ are in the interior of $A$ and $v_4 \notin \overrightarrow{L}$ (respectively, $v_5 \notin \overrightarrow{L}$). Delete $c_4$, $C_4$, $c_5$ and $C_5$, and update $\Omega$. Return $v_1$ and $v_4$ (respectively, $v_2$ and $v_5$) as the two points of contact. 
\end{itemize}

\item If  $\overrightarrow{L}$ contains two fixed vertices, meaning, that at least three vertices of the pentagon are known to the algorithm, then the adversary proceeds similarly to the second step of its protocol (read further) by revealing a new vertex of $\obj$, that will separate the two fixed vertices on the boundary of $\obj$. The second arm of the probe can touch an already known vertex. If this is not the case, then, since $\omega = \pi /2$, we always have freedom to adjust the wedge in such a way, that the second arm will intersect the circle of one of not yet revealed vertices of the pentagon. The adversary then fixes the vertex in the center of the corresponding circle. Note, that this probe may reveal up to two new pieces of information about $\obj$.
\end{enumerate}

After every probe the adversary updates the curve $A$, and if new vertices were discovered, it also updates the $\omega$-cloud $\Omega$. 

\vspace{0.5em} 
\begin{itemize}
\item If the algorithm spent $2$ probes and revealed $4$ vertices of the pentagon, then the fifth vertex is fixed by the adversary (as the center of the corresponding circle), but not revealed yet to the algorithm. The algorithm may guess its position though. 
\item If the algorithm spent $3$ probes but there is still at least one not fixed vertex of the pentagon, then the adversary fixes all of them.
\end{itemize}

Either one of the two cases terminates the initialization step, and the adversary proceeds to the second stage of its protocol.

\vspace{0.5em} 
During the second stage of the strategy the adversary answers each query with already revealed vertices unless the direction of the probe contains two vertices $u$ and $v$ known to the algorithm, that are adjacent on $A$. In this situation one new vertex (between $u$ and $v$ on the boundary of $\obj$) should be revealed or an edge $\overline{u v}$ must be confirmed. The adversary's strategy would not confirm any edge on this type of query until $n-1$ vertices have been revealed. Since we are not allowed narrow vertices, the adversary positions new vertices in such a way that their internal angle in the convex hull of already revealed vertices of $\obj$ is bigger than $\omega$. Note, that for the particular case, when $\omega = \pi /2$, fixed vertices that are unknown to the algorithm will be considered as revealed/known vertices.

\vspace{0.5em}
When the algorithm knows $n-1$ vertices of $\obj$ the adversary proceeds to the third stage of its strategy, during which the algorithm will be able to confirm edges. The adversary will not reveal the last, $n^{th}$ vertex, forcing the algorithm to inspect every consecutive pair of vertices on $A$ to check whether they are connected by an edge or there is a ``missing" vertex between them. As a result, the adversary will reveal already known vertices and confirm edges.

\vspace{0.5em}
By the final, fourth stage of the adversary's strategy, the algorithm knows $n-1$ vertices and $n-2$ edges of $\obj$. There is left only one pair of adjacent vertices $u$ and $v$ for which the edge $\overline{uv}$ has not been verified. When the algorithm sets the direction of the probe to be the line $L_{uv}$, the adversary has no choice but to reveal the last vertex. The algorithm should spend at least two more probes to complete the polygon's boundary, because the number of vertices $n$ is unknown to the algorithm. The query session then is complete and the convex polygon is found.

\vspace{0.5em}
The query for the adversary is a direction of the probe $\overrightarrow L$. The adversary asks the algorithm about its next probe and outputs the result.

\vspace{0.5em}
The input for the algorithm is the circle $\Psi$ and the point $p$. The algorithm also knows that the point $p$ is inside the unknown object $\obj$. It then performs probes (i.e., gives the adversary the direction of the probe and gets results) and outputs a convex polygon $\obj'$. If $\obj' = \obj$ then the output is correct.

\vspace{0.5em}
\rule{0.88\textwidth}{1pt}

\noindent 
{\bf Invariant (maintained by the adversary):}
\begin{enumerate}
\item $A$ is a closed convex curve that is the intersection of the circle $\Psi$ and all $\omega$-probes made so far by the algorithm. $A$ is a circle or it consists of circular arcs and line segments.
\item $P$ is a set of points on $A$.
\item Each line segment of $A$ contains at least one and at most two points of $P$.
\item $F$ is a variable with value $F = | \{ u \in P : u$ and $ccw(u)$ are not on the same line segment of $A \} |$, where $ccw(u)$ is the counter-clockwise successor of $u$ on $A$.
\item There exists a convex $n$-gon $\obj$, such that each point of $P$ is a vertex of $\obj$ and all internal angles of $\obj$ are larger than $\omega$. If $\obj$ is the object to be reconstructed, then the algorithm makes the same sequence of probes it has made so far.
\item After the initialization, $P$ contains at least four points (for $0< \omega < \pi /2$), and five or six points (for $\omega = \pi /2$). All the internal angles of the polygon that have points of $P$ as vertices, are larger than $\omega$.   
\end{enumerate}

\vspace{1em}
\noindent 
{\bf Adversary's strategy:}
\begin{enumerate}
\item Initialization: Choose $n \geq 5$, where $n$ denotes the number of vertices of the object to be reconstructed. On the first valid $\omega$-probe the adversary reveals two vertices $v_1$ and $v_2$, that are positioned as described above.
Refer to Figure~\ref{fig:adversary_initialization}.

We consider two cases:
\begin{enumerate}
\item $0< \omega < \pi /2$: If the direction of the probe contains any of the four edges of the rectangle $\{v_1$, $t_3$, $v_2$, $t_4 \}$, the points $t_3$ and $t_4$ move to avoid edge confirmation (as described above). If the probe reveals $t_3$ (or $t_4$), we fix both vertices and we rename $t_3$ into $v_3$ (and $t_4$ into $v_4$). We output only vertices found by the probe. Set $P = \{v_1, v_2, v_3, v_4\}$, $F=4$ and $A = \Psi \cap \omega$-wedge.

\item $\omega = \pi /2$: If the direction of the probe contains any of the five edges of the pentagon $\{v_1, c_4, c_5, v_2, c_3\}$, the points $c_3$, $c_4$ and $c_5$ move to avoid edge confirmation (as described above). If both points are fixed, reveal additional vertex $v$. Otherwise, answer queries with already known vertices or not fixed vertices of the pentagon until one of the following two happens:
\begin{enumerate}
\item If the algorithm spent $2$ probes and revealed $4$ vertices of the pentagon, then fix/rename the fifth vertex (as the center of the corresponding circle). Set $P = \{v_1, v_2, v_3, v_4, v_5\}$.
\item If the algorithm spent $3$ probes, then fix/rename all the vertices of the pentagon. Set $P = \{v_1, v_2, v_3, v_4, v_5\}$. If an additional vertex $v$ was revealed then add $v$ to $P$.
\end{enumerate}
Set $F=|P|$ and update $A$.
\end{enumerate}

To answer all subsequent queries the adversary will use the following protocol:

\item While $|P| < n-1$ do the following: 
\begin{quote}
Answer each query with already revealed vertices, unless the direction of the probe contains two known vertices $u, v \in P$, that are adjacent on $A$. In this situation the answer will be one new vertex $v'$ (between $u$ and $v$ on the boundary of $\obj$, such that $\angle (v, v', u) > \omega$) and one vertex $v''$ that is already known to the algorithm. 

Add $v'$ to $P$; set $F = F+1 $. Update $A$.
\end{quote}

\item While $F >1$ do the following: 
\begin{quote}
Answer each query with vertices from $P$, unless the direction of the probe contains $u, v \in P$, such that $v$ is a counter-clockwise successor of $u$ on $A$, and $u$, $v$ do not belong to the same straight-line edge of $A$. In the latter case confirm the edge $\overline{u v}$, set $F= F-1$ and update $A$. 
\end{quote}

\item Answer each query with vertices from $P$, unless the direction of the probe contains $u, v \in P$, such that $u$ and $v$ are adjacent on $A$ and they are not on the same straight-line edge of $A$ (in other words, the edge $\overline{uv}$ has not been verified). In this situation the answer would be one new vertex $v'$ (between $u$ and $v$ on the boundary of $\obj$) and one vertex $v''$ that is already known to the algorithm. 
\end{enumerate}

The algorithm should spend at least two more probes to complete the boundary of the polygon, because the number of vertices $n$ is unknown to the algorithm. The query session is then complete and the polygon is found.

\rule{0.88\textwidth}{1pt}

\vspace{0.5em}
It follows from this protocol that the invariant is correctly maintained. With the above strategy any algorithm will spend at least $2n-2$ probes for any $\omega$ in the range $0 < \omega < \pi/2$, and at least $2n-3$ probes for $\omega = \pi/2$. 

\begin{theorem}
\label{theo:model_1_sum-up_theorem_lower_bound}
Given an $\omega$-wedge, where $0 < \omega < \pi/2$, for every algorithm, there exists a convex polygon $\obj$ (whose angles are strictly bigger than $\omega$), such that $2n - 2$ $\omega$-probes are required to determine its shape. For the case when $\omega = \pi/2$, $2n-3$ $\omega$-probes are required to determine $\obj$.
\end{theorem}
\begin{proof}
To prove this, consider the quantity 
\[ \Phi = 2|P| - F . 
\]
At the start of the strategy $P$ is empty and $F = 0$, so $\Phi = 0$. 

\vspace{0.5em}
Let us first analyze the case when $0 < \omega < \pi/2$. At the initial step any algorithm will spend at least 2 probes to reveal 4 vertices. After initialization $|P| = 4$ and $F = 4$, thus $\Phi = 4$.

Every iteration of the while-loop in the second step of the protocol either does not change the values of $|P|$ and $F$, or increases $|P|$ by $1$ and increases $F$ by $1$. Therefore, $\Phi$ increases by at most one on every probe made by the algorithm. This while-loop makes at least $n-5$ iterations (forcing the algorithm to spend at least $n-5$ probes) resulting in $\Phi = n-1$.

Every iteration of the while-loop in the third step of the protocol does not change the value of $|P|$, but decreases the value of $F$ by at most $1$. Therefore, $\Phi$ increases by at most one on every probe made by the algorithm. This while-loop makes at least $n-2$ iterations (forcing the algorithm to spend at least $n-2$ probes) resulting in $\Phi = 2n-3$.

In the final step of the protocol, any algorithm will spend at least $3$ probes to complete the boundary of the polygon. We prove below (see Claims $1$ and $2$) that, after termination we have $|P| = n$, $F = 0$ and $\Phi = 2n$.
  
To conclude, the number of probes made by any algorithm is at least $2 + (n-5) + (n-2) + 3 = 2n - 2$.  
  
\vspace{0.5em}
For the case when $\omega = \pi/2$, the initialization step can have two outcomes: 
\begin{enumerate}
\item At most two probes were spent, that resulted in $|P| = 5$, $F = 5$ and thus $\Phi = 5$.
\item At most three probes were spent, that resulted in $|P| = 6$, $F = 6$ and thus $\Phi = 6$.
\end{enumerate}
The first while-loop makes at least $n-6$ iterations for the first outcome (forcing the algorithm to spend at least $n-6$ probes), and at least $n-7$ iterations for the second outcome (forcing the algorithm to spend at least $n-7$ probes) resulting in $\Phi = n-1$. The analysis for the third and fourth step of the protocol are identical to the one given above. So, the number of probes made by any algorithm is at least $2 + (n-6) + (n-2) + 3 = 2n - 3$ (for the first outcome of the initialization), and $3 + (n-7) + (n-2) + 3 = 2n - 3$ (for the second outcome).
\qed
\end{proof} 

\begin{figure}[h]
\centerline{\resizebox{!}{5.5cm}{\includegraphics{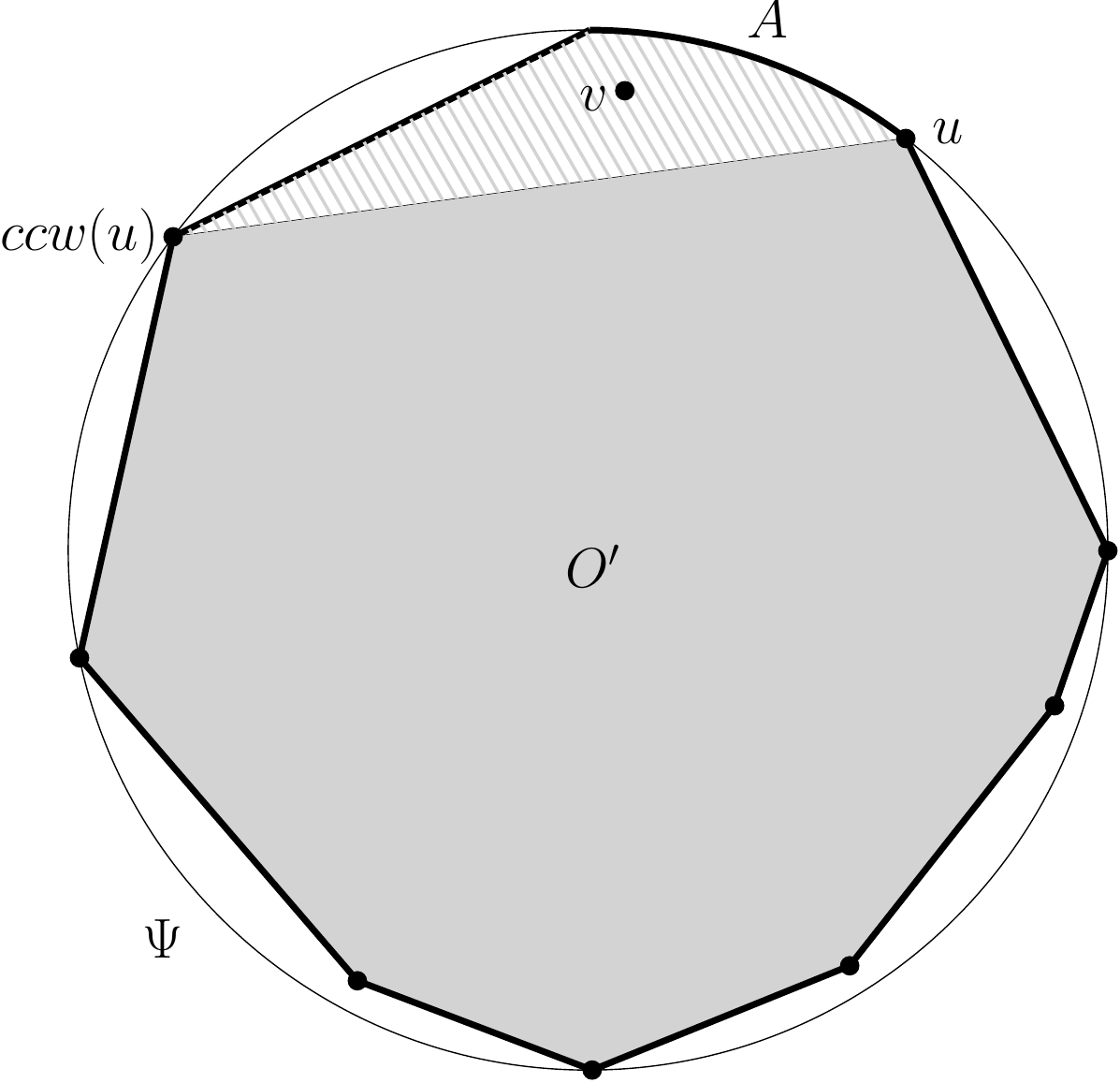}}}
\caption{The figure shows an area (tiled with grey lines) not explored by the algorithm. The curve $A$ is shown in bold; the object $\obj'$, returned by the algorithm, is depicted in gray.}
\vspace{-1.0em}
\label{fig:No_B_claim1}
\end{figure}  

\textit{Claim 1:} When the algorithm terminates: $F = 0$.
\begin{proof}
Assume to the contrary that the algorithm returned an object $\obj'$ (where $\obj'$ is a convex hull of $P$), but $F \geq 1$. When $F$ is not empty, there exists a vertex $u \in P$ such that $u$ and $ccw(u)$ are not on the same line segment of $A$. Since $A$ is convex, the line segment $\overline{u, ccw(u)}$ is completely inside $A$ and its intersection with $A$ is the set of points $\{u, ccw(u) \}$. Refer to Figure~\ref{fig:No_B_claim1}. There exists a vertex $v$ in the interior of the region, consisting of all points in $A$, that are to the right of $\overline{u, ccw(u)}$. The adversary can claim that the object $\obj$ is the convex hull of $P \bigcup \{v \}$. If the adversary gives the convex hull of $P \bigcup \{v \}$ to the algorithm as an object, the algorithm will make exactly the same set of probes as it did for the convex hull of $P$, and therefore will give the same output $O'$ for different objects $\obj$ and $\obj'$. Since $\obj \neq \obj'$ the algorithm is not correct.
\qed
\end{proof}

\vspace{0.5em} 

\textit{Claim 2:} When the algorithm terminates: $|P| = n$.
\begin{proof}
Assume that the algorithm returned an object $\obj'$ such that $|P| < n$ (notice, that $P$ cannot be bigger than $n$). According to the invariant, there exists a convex $n$-gon $\obj$ such that each point of $P$ is a vertex of $\obj$ and the algorithm would make the same sequence of probes on $\obj$ as it has made on $\obj'$. It is clear that $\obj \neq \obj'$ and thus the algorithm is not correct. 
\qed
\end{proof}

\section{Polygons with Narrow Vertices} 
\label{sec:Polygons_with_B-vertices}

This section addresses the reconstruction of convex polygons that have narrow vertices. We showed in Section~\ref{sec:preliminaries} that the number of narrow vertices cannot exceed $3$ (refer to Observation~\ref{observ:N_B_general}) and if $\omega <\pi/3$ then the number of narrow vertices can be at most $2$ (refer to Observation~\ref{observ:N_B_small_omega}).

\vspace{0.5em}
Recall that the result of a valid probe is the coordinates of the apex of the probe. Notice that if the ray in the direction of the probe enters $\obj$ via narrow vertex, then the apex of this probe belongs to the $\omega$-cloud $\Omega$ of the polygon. If the polygon does not have any narrow vertices then the $\omega$-cloud and the boundary of the polygon are disjoint.

\vspace{0.5em}
Let us first investigate the problem of reconstructing polygons that have exactly one narrow vertex.

\subsection{Polygons with Exactly One Narrow Vertex}
\label{subsec:Polygons_with_one_B-vertex}

Let $v_B$ be the only narrow vertex of $\obj$. Let $a$ and $b$ be two points and let $L_{ab}$ be the line through $a$ and $b$. Let $\overrightarrow{L_{ab}}$ be the direction of a probe $W$, directed from $a$ to $b$. Notice, that if we shoot a probe along $\overrightarrow{L_{v_Ba}}$, then the $\omega$-probe will stop with the apex $q$ touching the vertex $v_B$. The arms of the probe may or may not touch the boundary of the polygon (see Figures~\ref{fig:points_of_contact_3} and~\ref{fig:one_B-case}).

\begin{figure}[h]
    \begin{center}
        \subfigure[$v_B$ is the only narrow vertex of $\obj$.]{
            \label{fig:one_B-case}
            \includegraphics[width=0.4\textwidth]{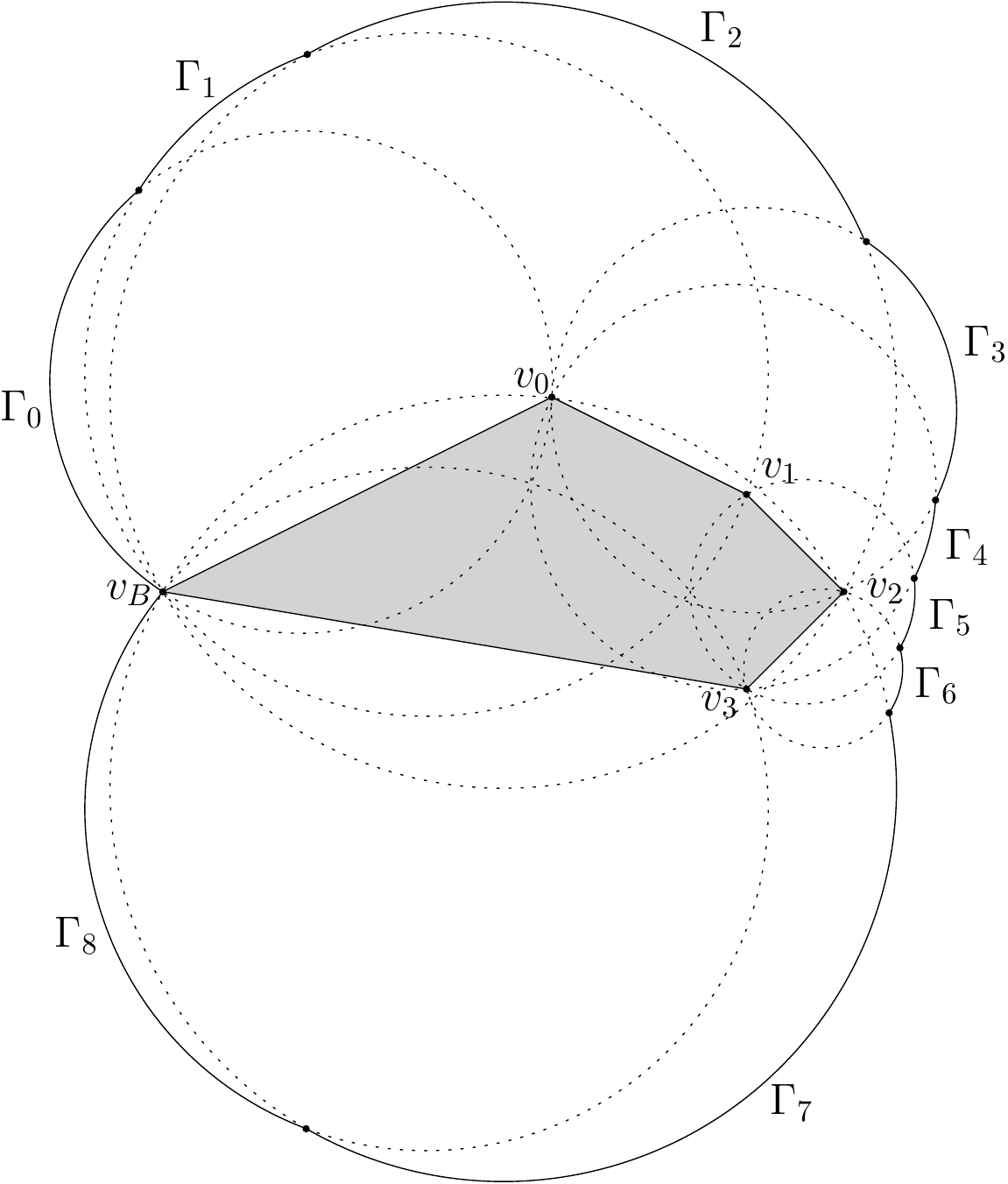}
        }%
        \hspace{0.1cm}%
        \subfigure[$\obj$ has two narrow vertices. The triangle $\triangle (v_{B_1}, v_{B_2}, u)$ may or may not contain vertices of $\obj$, other that $v_{B_1}$ and $v_{B_2}$.]{
            \label{fig:two_B-case}
            \includegraphics[width=0.45\textwidth]{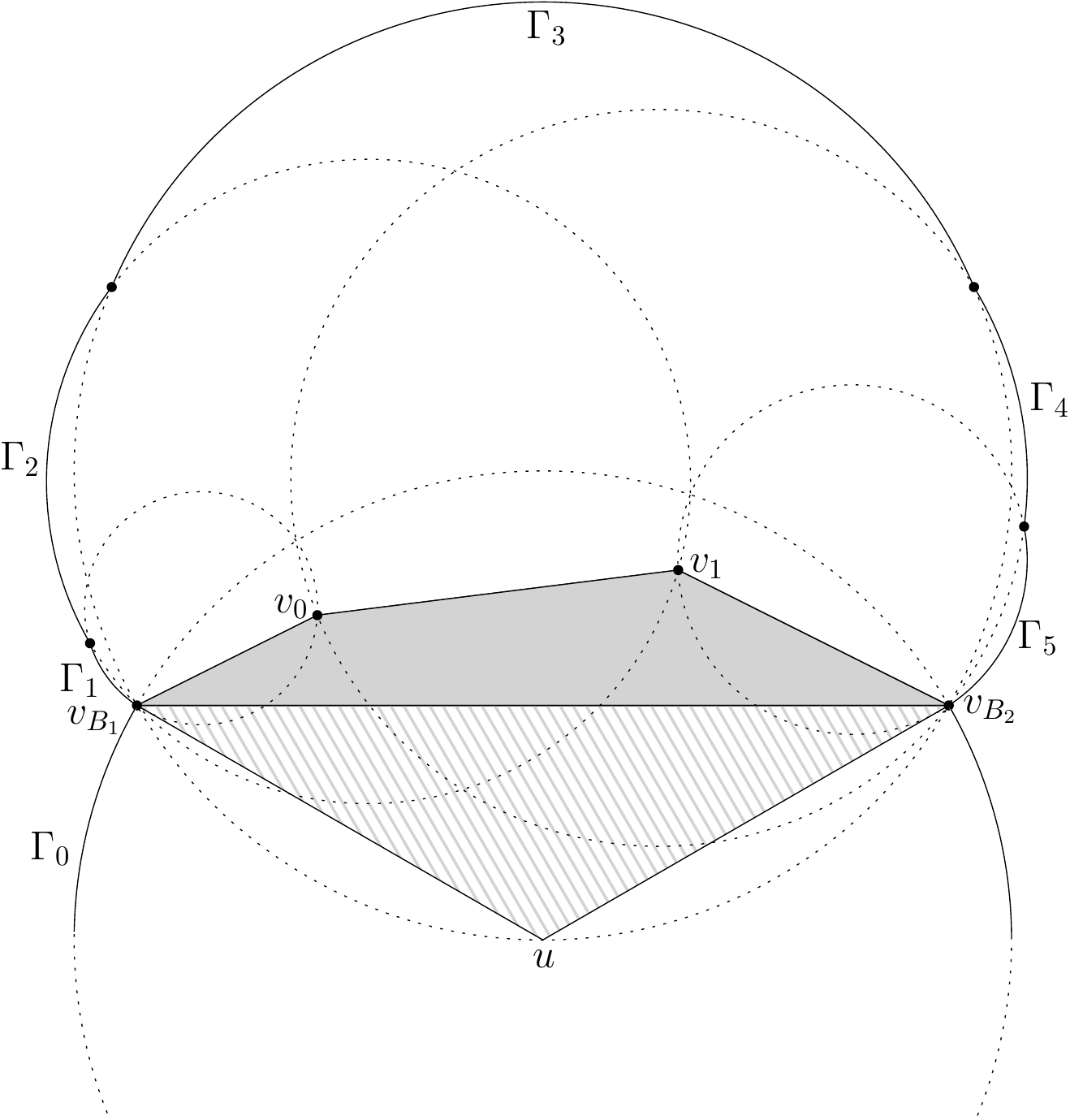}
        }%
    \end{center}
    \vspace{-1.5em}
    \caption{Case studies of narrow vertices of $\obj$. In both examples, $\omega = 60^\circ$.}
\label{fig:B-cases}
\end{figure}

\vspace{0.5em}
Look at the example of Figure~\ref{fig:one_B-case}. Assume that we have found the vertices $v_B$ and $v_0$, but the vertices $v_1$, $v_2$ and $v_3$ are still unknown. Suppose that our next step is to check whether $\overline{v_Bv_0}$ is an edge. Since we have no control over the rotation of the arms, the probe in the direction $\overrightarrow{L_{v_Bv_0}}$ might give us no new information other than to confirm that $v_B$ is a narrow vertex (in case we did not know already). The solution is to probe in the opposite direction, namely $\overrightarrow{L_{v_0v_B}}$. This will confirm the edge $\overline{v_Bv_0}$ (and reveal $v_2$). 

\vspace{0.5em}
Consider again the example of Figure~\ref{fig:one_B-case}, but now assume that we have found $v_B$ and $v_1$, but the vertices $v_0$, $v_2$ and $v_3$ are still unknown. Our goal is to deny the existence of the edge $\overline{v_Bv_1}$ by revealing the vertex $v_0$. But we cannot rely on the probing in the direction $\overrightarrow {L_{v_Bv_1}}$. In the worst case it will not return $v_0$. As in the previous example, the probe in the opposite direction will work. The probe along $\overrightarrow {L_{v_1v_B}}$ returns two new vertices, namely $v_0$ and $v_2$.

\vspace{0.5em}
We may conclude, that having one narrow vertex in the polygon is not a problem. Once a narrow vertex is identified, all other elements of $\obj$ can be found by using algorithm from Section~\ref{subsec:algorithm} while probing in the directions other then those whose first point of intersection with $\obj$ is the narrow vertex. Note that we may discover a narrow vertex by probing in different directions. But only the probe along $\overrightarrow {L_{v_Ba}}$ (where $a$ is any point interior to $\obj$) can identify $v_B$ as a narrow vertex. Hence, we may spend one additional probe per narrow vertex to reconstruct the polygon $\obj$. The algorithm that reconstructs convex polygon with exactly one narrow vertex is given below in the Section~\ref{subsec:Algorithm_B_vertex}. The algorithm reconstructs $\obj$ by using not more than $2n-1$ probes. This upper bound is proved in Theorem~\ref{theo:B-case_sum-up_Upper_Bound}.

\subsection{Polygons with Two or Three Narrow Vertices}
\label{subsec:Polygons_with_many_B-vertices} 

The problem arises when a polygon has more than one narrow vertex. Refer to Figure~\ref{fig:two_B-case}. Assume we have found two narrow vertices $v_{B_1}$ and $v_{B_2}$, but all other elements of $\obj$ are still unknown. Our previous strategy to shoot a probe along $\overrightarrow {L_{v_{B_1}v_{B_2}}}$ or in the opposite direction (along $\overrightarrow {L_{v_{B_2}v_{B_1}}}$) does not work.

If there is a vertex $v$ of $\obj$ that is between $v_{B_1}$ and $v_{B_2}$ in a clockwise order around the boundary of $\obj$, then there are at least two arcs $\Gamma_v$ and $\Gamma_v'$ (supported by $v$) that belong to the part of the $\Omega$-cloud enclosed between two pivots $v_{B_1}$ and $v_{B_2}$ in a clockwise order around the cloud. (Recall from Definition~\ref{def:omega_cloud} that a pivot point is an intersection point of two consecutive circular arks of the $\Omega$-cloud). If we probe in the direction toward the arc $\Gamma_v$ then we will find $v$. We know that $\Gamma_v$ is on the cloud between $v_{B_1}$ and $\Gamma_v'$ (clockwise), but we do not know how big it is. In other words, we do not know how close to $v_{B_1}$ we should aim.

\vspace{0.5em}
Let $\varepsilon > 0$ be a fixed real number. We suppose, that for every pair of narrow vertices $v_{B_1}$ and $v_{B_2}$ of $\obj$, if $v_{B_1} v_{B_2}$ is not an edge that belongs to $\obj_{v_{B_1} v_{B_2}}$\footnote{For the rest of the paper, $\obj_{v_1 v_2}$ denotes the part of the boundary of $\obj$ in a counter-clockwise order from $v_1$ to $v_2$ including $v_1$ and $v_2$.}, then there exist a vertex $v \in \obj_{v_{B_1} v_{B_2}}$ such that $\angle (v_{B_2}, v, v_{B_1})  \leq \pi-\varepsilon$. According to this hypothesis, one of the angles $\angle (v, v_{B_1}, v_{B_2})$ or $\angle (v_{B_1}, v_{B_2}, v)$ is at least $\varepsilon / 2$. Assume, without loss of generality, that $\angle (v_{B_1}, v_{B_2}, v) \geq \varepsilon / 2$. Thus, if we shoot a probe along the line $\overrightarrow {L_{x v_{B_2}}}$ that makes a positive angle $\leq \varepsilon / 2$ with the line $L_{v_{B_1}v_{B_2}}$, we will hit some vertex (not necessarily $v$) that belongs to $\obj_{v_{B_1} v_{B_2}}$. Knowing that vertex allows us to discover all other vertices on the boundary $\obj_{v_{B_1} v_{B_2}}$ by using the same strategy as before when we had only one narrow vertex.

Since we do not know which one of the two angles $\angle (v_{B_1}, v_{B_2}, v)$ or $\angle (v, v_{B_1}, v_{B_2}) $ is bigger than $\varepsilon / 2$, an additional probe (along the line $\overrightarrow {L_{x v_{B_1}}}$ that makes a negative angle $\leq \varepsilon / 2$ with the line $L_{v_{B_2}v_{B_1}}$) may be required to identify $v$. If both probes return no new information, then we confirm an edge from $v_{B_2}$ to $v_{B_1}$ in the counter clockwise  order around the boundary of $\obj$.

\vspace{0.5em}
But what if $\varepsilon$ is unknown? We may try to guess $\varepsilon$. We may choose an initial (very small) value for $\varepsilon$ and probe along the line $\overrightarrow {L_{x v_{B_2}}}$ that makes a positive angle $\leq \varepsilon / 2$ with the line $L_{v_{B_1}v_{B_2}}$. And also probe along the line $\overrightarrow {L_{x' v_{B_1}}}$ that makes a negative angle $\leq \varepsilon / 2$ with the line $L_{v_{B_2}v_{B_1}}$. If no vertex is discovered, we reduce the value of $\varepsilon$ by a factor of two and we repeat the above procedure until we hit some vertex that belongs to $\obj_{v_{B_1} v_{B_2}}$. This strategy, unfortunately, can take an infinite number of probes in the worst case if $\overline{v_{B_1}v_{B_2}}$ is an edge of $\obj$ that belongs to $\obj_{v_{B_1} v_{B_2}}$. Look at the example of Figure~\ref{fig:two_B-case} and assume that the only vertices of $\obj$ inside the triangle $\triangle (v_{B_1}, v_{B_2}, u)$ are $v_{B_1}$ and $v_{B_2}$ and we probe along a line $\overrightarrow {L_{x v_{B_2}}}$ that approaches the edge $\overline{v_{B_1}v_{B_2}}$ from below. 

\vspace{0.5em}
In other words, if $\varepsilon$ is unknown, the edge $\overline{v_{B_1}v_{B_2}}$ cannot be verified and thus $\obj$ cannot be reconstructed precisely.

Now we are ready to present an algorithm for the reconstruction of a convex polygon whose internal angles may be smaller than $\omega$.

\subsection{Algorithm}
\label{subsec:Algorithm_B_vertex}

In this section we show an algorithm that reconstructs general convex polygons. For polygons with two or three narrow vertices the algorithm requires a real number $\varepsilon > 0$ as a part of an input. Otherwise the polygon may not be reconstructed precisely. Upper bound on the number of probes spent by the algorithm is given in Theorem~\ref{theo:B-case_sum-up_Upper_Bound}.

\vspace{0.5em}
\rule{0.88\textwidth}{1pt}

\noindent 
{\bf Invariant:} 
\begin{enumerate} 
\item $Q$ is a convex polygon. 
\item Each vertex of $Q$ is a vertex of $\Ob$ (and, therefore, $Q$ is contained in $\Ob$). 
\item Each vertex $u$ of $Q$ stores a Boolean variable $\flag(u)$; if $\flag(u) = \TRUE$, then $u$ and its counter-clockwise successor in $Q$ form an edge of $\Ob$. 
\item Each vertex $u$ of $Q$ stores a Boolean variable $Bvertex(u)$; if $Bvertex(u) = \TRUE$, then $u$ is a narrow vertex in $\obj$. 
\item $F$ is a variable with value $F = | \{ u \in Q : \flag(u) = \FALSE \} |$. 
\end{enumerate}  

\vspace{0.5em} 

\noindent 
{\bf Algorithm:} 

To initialize the algorithm, choose an arbitrary line 
$\overrightarrow{L}$ that contains $p$. Let $(q,H_1,H_2,p_1,p_2)$ be the 
outcome of the probe along $\overrightarrow{L}$. 
\begin{itemize}
\item If $p_1 \neq p_2$, initialize $Q$ as the polygon consisting of the vertices $p_1$ and $p_2$, set $\flag(p_1) = \FALSE$, $\flag(p_2) = \FALSE$, $Bvertex(p_1) = \FALSE$, $Bvertex(p_2) = \FALSE$ and $F = 2$. 
\item Otherwise, if $p_1 = p_2 = q$, initialize $Q$ as the 
      polygon consisting of one vertex $q$, set $\flag(q) = \FALSE$, $Bvertex(q) = \TRUE$, and $F = 1$. Let $(q',H'_1,H'_2,p'_1,p'_2)$ be the outcome of the probe along the bisector of the first probe in the direction from inside the first wedge toward $q$. 
      \begin{itemize}
      \item If $p'_1 \neq p'_2$ then: 
            \begin{itemize}
            \item if $p'_1 \neq q$, replace $Q$ by the convex hull of $Q \cup \{p'_1\}$ and set $\flag(p'_1) = \FALSE$, $Bvertex(p'_1) = \FALSE$ and $F = F + 1$.
             \item if $p'_2 \neq q$, replace $Q$ by the convex hull of $Q \cup \{p'_2\}$ and set $\flag(p'_2) = \FALSE$, $Bvertex(p'_2) = \FALSE$ and $F = F + 1$.    
            \end{itemize}
      \item Otherwise, if $p'_1 = p'_2 = q'$, do the following: if $q = q'$ set $F=0$. Otherwise, replace $Q$ by the convex hull of $Q \cup \{q'\}$ and set $\flag(q') = \FALSE$, $Bvertex(q') = \TRUE$ and $F = F + 1$.
      \end{itemize} 
\end{itemize}

\vspace{0.5em} 

\noindent 
While $F \neq 0$, do the following:    
\begin{enumerate} 
\item 
      \begin{enumerate}
      \item If there is a vertex $u$ in $Q$ such that $\flag(u)=Bvertex(u)=\FALSE$, let $u$ be such a vertex (give priority to vertices whose internal angles in $Q$ are bigger than $\omega$), and let $v$ be the counter-clockwise successor of $u$ in $Q$.      
      \item Else, if there is a vertex $u$ in $Q$ for which $Bvertex(u)=\FALSE$ and whose clockwise successor $v$ in $Q$ satisfy the following properties: $\flag(v)=\FALSE$ and $Bvertex(v)=\TRUE$, then let $u$ be such a vertex, and let $v$ be its clockwise successor in $Q$.
      \item Else, do the following:
      
      If we are not given an $\varepsilon$, stop the algorithm and return $Q$. Otherwise, take a pair of arbitrary adjacent vertices $u$ and $v$ in $Q$ for which $\flag(u)=\FALSE$, $Bvertex(u) = Bvertex(v) =\TRUE$ and $v$ is the counter-clockwise successor of $u$ in $Q$. 
      Let $\overrightarrow{L_{uv}}$ be the line through $u$ and $v$ that is directed from $u$ to $v$. Let $\overrightarrow{L}$ be the line $\overrightarrow{L_{uv}}$ that is rotated to the left by $\varepsilon /2$ degrees around vertex $v$. Let $(q,H_1,H_2,p_1,p_2)$ be the outcome of the probe along $\overrightarrow{L}$.
      \begin{enumerate}
      \item If $p_1 \neq v$, then insert $p_1$ into $Q$ between $u$ and $v$, and set $\flag(p_1) = \FALSE$, $Bvertex(p_1) = \FALSE$ and $F=F+1$. 
      \item If $p_1 = v$, then let $\overrightarrow{L'}$ be the line $\overrightarrow{L_{vu}}$ that is rotated to the right $\varepsilon /2$ degrees around vertex $u$. Let $(q',H'_1,H'_2,p'_1,p'_2)$ be the outcome of the probe along $\overrightarrow{L'}$.
            \begin{itemize}
            \item If $p'_2 \neq u$ then insert $p'_2$ into $Q$ between $u$ and $v$, and set $\flag(p'_2) = \FALSE$, $Bvertex(p'_2) = \FALSE$ and $F=F+1$.                         
            \item If $p'_2 = u$ then set $\flag(u) = \TRUE$ and $F=F-1$.
            \end{itemize}
      \end{enumerate}  
      Skip steps $2$, $3$ and $4$ and continue to the next iteration of the while-loop.
      \end{enumerate}
\item Let $\overrightarrow{L}$ be the line through $u$ and $v$ that is directed from $u$ to $v$. Let $(q,H_1,H_2,p_1,p_2)$ be the outcome of the probe along $\overrightarrow{L}$. 
      
      If $\flag(u)=\TRUE$, rename $p_1$ into $p_2$ and $p_2$ into $p_1$.
      
      If $p_1 = p_2 = q = u$, then set $Bvertex(u) = \TRUE$ and go to step 1.
\item \begin{enumerate} 
      \item If $p_1 = u$, then 
      \begin{itemize}
      \item if $\flag(u) = \TRUE$, then set $\flag(v) = \TRUE$ and $F=F-1$.
      \item otherwise, set $\flag(u) = \TRUE$ and $F=F-1$.
      \end{itemize} 
      \item If $p_1 \neq u$, insert $p_1$ into $Q$ between $u$ and $v$, and set $\flag(p_1) = \FALSE$, $Bvertex(p_1) = \FALSE$ and $F=F+1$. 
      \end{enumerate}  
\item \begin{enumerate} 
      \item If $p_2 = u$, set $\flag(v) = \TRUE$ and $F=F-1$. 
      \item If $p_2 \neq u$ and $p_2$ is not a vertex of $Q$, then replace $Q$ by the convex hull of $Q \cup \{p_2\}$, and set $\flag(p_2) = \FALSE$, $Bvertex(p_2) = \FALSE$ and $F=F+1$. 
      \end{enumerate}
\end{enumerate} 

\rule{0.88\textwidth}{1pt} 
\vspace{1em} 

\noindent 
It follows from the algorithm that the invariant is correctly maintained. We show that the algorithm terminates, by proving the following theorem.

\begin{theorem}[Upper Bound]
\label{theo:B-case_sum-up_Upper_Bound}
Given $\omega$-wedge, with $0 < \omega \leq \pi/2$, a convex polygon $\obj$, and $\varepsilon > 0$ (such that for every pair of narrow vertices $v_{B_1}$ and $v_{B_2}$ of $\obj$, if $v_{B_1} v_{B_2}$ is not an edge that belongs to $\obj_{v_{B_1} v_{B_2}}$, then there exists a vertex $v \in \obj_{v_{B_1} v_{B_2}}$ such that $\angle (v_{B_2}, v, v_{B_1})  \leq \pi-\varepsilon$). Let $N_B$ be the number of narrow vertices of $\obj$ and let: 

\begin{equation}
\label{equation P_B}
P_B = 
\begin{cases}
-1, &  N_B = 0, \\
-1, &  N_B = 1, \\
2, &  N_B = 2, \\
3, &  N_B = 3.
\end{cases}
\end{equation}

The above algorithm reconstructs $\obj$ using at most $2n - 1 + N_B + P_B$ $\omega$-probes. 
\end{theorem}

\begin{proof}
Let $n$ denote the number of vertices of $\Ob$. Consider the quantity 
\[ \Phi = 2|Q| + N_B' - F , 
\]
where $|Q|$ denotes the number of vertices of $Q$, $N_B'$ denotes the number of narrow vertices of $Q$, $0 \leq N_B' \leq 3$ (refer to Observation~\ref{observ:N_B_general}). The quantity $\Phi$ represents the number of units of information we need to obtain in order to successfully reconstruct $\obj$. These units of information are vertices, edges and narrow vertices of $\obj$ that were identified as such by the algorithm. At the initial step of the algorithm, after the first probe, we have $\Phi = 2$: either the probe returned two vertices (in which case $|Q| = 2$, $N_B' = 0$, $F = 2$), or it returned one narrow vertex ($|Q| = 1$, $N_B' = 1$, $F = 1$). In the first case, the algorithm proceeds to the while-loop step, while in the second case, one more probe is required. This additional probe reveals either two new vertices, or one new vertex, or one new narrow vertex. 

\vspace{0.5em}
During every iteration of the while-loop, the algorithm reveals at least one vertex or one edge of $\obj$ or marks an already known vertex as a narrow vertex. Therefore, in every iteration, the value of $\Phi$ increases by at least one. Notice also that if $\epsilon$ is known, then the step $(1.c)$ of the algorithm increases the value of $\Phi$ by exactly one, but may use up to two probes. Fortunately, this step of the algorithm cannot be executed more than three times (refer to Observation~\ref{observ:N_B_general}). 

Since at any moment we have $\Phi \leq 2n + N_B$ (where $N_B$ stands for the number of narrow vertices of $\obj$, $0 \leq N_B \leq 3$ (refer to Observation~\ref{observ:N_B_general})), it follows that the while-loop makes at most $2n - 2 + N_B$ iterations and thus, the algorithm terminates.  

\vspace{0.5em} 
After the very first probe $\Phi = 2$. If this probe revealed two vertices, then the first iteration of the while-loop increases the value of $\Phi$ either by two (in which case two new vertices or one new vertex and one new edge were revealed), or by one (in which case an already known vertex was marked as a narrow vertex, making the number of non-revealed narrow vertices of $\obj$ to decrease by one). Thus, the algorithm terminates after at most $2n - 3 + N_B$ additional iterations, making a total of at most $2n - 1 + N_B + P_B$ probes. Notice, that if $\obj$ has no narrow vertices, then the first iteration of the while-loop always increases $\Phi$ by two, and thus the while-loop makes at most $2n - 4$ additional iterations, that results in at most $2n - 2$ probes. This is why $P_B$ is set to $-1$ for polygons without narrow vertices. Similarly, if the polygon has exactly one narrow vertex, that was discovered by the first probe and marked as a narrow vertex during the first iteration of the while-loop, then the second iteration of the loop reveals two new pieces of information. Those two new pieces can be either two new vertices or one new vertex and one new edge. Thus, the while-loop makes at most $2n-4$ additional iterations. Hence, the algorithm in this case makes at most $2n-1$ probes to reconstruct the polygon with one narrow vertex. 

\vspace{0.5em} 
If the first probe of the algorithm revealed a narrow vertex, then an additional probe is required during the initialization step. This second probe may reveal:
\begin{itemize}
\item two new vertices, making $\Phi = 4$ ($|Q| = 3$, $N_B' = 1$ and $F = 3$). 
\item one new vertex, making $\Phi = 3$ ($|Q| = 2$, $N_B' = 1$ and $F = 2$).
\item one new narrow vertex, making $\Phi = 4$ ($|Q| = 2$, $N_B' = 2$ and $F = 2$).
\end{itemize} 

Notice, that in either case, the number of non-revealed narrow vertices of $\obj$ decreases by $N_B'$. 

It then requires at most $2n - 3 + N_B$ iterations of the while-loop to reconstruct $\obj$, which results in a total of at most $2n - 1 + N_B + P_B$ probes. Notice that if $\obj$ has exactly one narrow vertex that was discovered and marked by the first probe, then one of the following probes will discover two new pieces of information. This probe can be either the second probe of the initialization step or the first iteration of the loop. Two new pieces of information can be either two new vertices or one new vertex and one new edge. Thus, the algorithm makes at most $2n-4$ additional probes, and reconstructs the polygon with one narrow vertex by using at most $2n-1$ probes. 
 
After termination, we have $F=0$. It then follows from the invariant 
that, at that moment, $Q$ is equal to $\Ob$. Finally, the number of 
probes made by the algorithm is at most $2n - 1 + N_B + P_B$. 
\qed
\end{proof}

\subsection{Lower Bound for General Polygons}
\label{subsec:Lower_Bound_General_Polygons} 

We show a lower bound on the number of probes needed to reconstruct a convex $n$-gon by presenting an adversarial argument. In our proof, the adversary sets the number of narrow vertices in $\obj$ before the first probe of the algorithm and cannot change its decision during the run of the algorithm. Thus, the adversary have to induce the predetermined number of narrow vertices in $\obj$ despite the probing strategy of the algorithm. The algorithm, in its turn, may correctly guess the number of narrow vertices and use its best strategy for that specific number of narrow vertices. Thus it may reconstruct $\obj$ by using less probes, than required for any other algorithm that does not know about constraints on $\obj$. 

\vspace{0.5em}
We prove that, if the algorithm knows in advance the exact number of narrow vertices $0 \leq N_B \leq 3$ in $\obj$, the lower bound on the number of probes required to reconstruct $\obj$ is 
\begin{list}{•}{•}
\item $2n - 2$ for $N_B = 0$,
\item $2n - 1$ for $N_B = 1$,
\item $2n + 2$ for $N_B = 2$ and  $N_B = 3$.
\end{list}

The lower bound for zero or one narrow vertex is tight with the upper bound (refer to Theorems~\ref{theo:model_1_sum-up_theorem_upper_bound},~\ref{theo:model_1_sum-up_theorem_upper_bound_90},~\ref{theo:model_1_sum-up_theorem_lower_bound} and~\ref{theo:B-case_sum-up_Upper_Bound}). The lower bound for $N_B = 2$ (respectively $N_B = 3$) is smaller by one probe (respectively, by 3 probes) than the upper bound given in Theorem~\ref{theo:B-case_sum-up_Upper_Bound}. We will show later in this section how to improve our algorithm for the case when the number of narrow vertices of $\obj$ is provided to the algorithm. In this case, our lower bounds are tight. Alternatively, if the adversary is allowed to change its decision about the number of narrow vertices in $\obj$ during the run of the algorithm, it can make the lower bound equal to the upper bound for any number of narrow vertices.

\vspace{0.5em}
For now, assume that the adversary sets the number of narrow vertices in $\obj$ before the first probe of the algorithm and does not change its decision during the execution of the algorithm. The adversary's strategy that forces at least $2n - 2$ $\omega$-probes from any algorithm which reconstructs $\obj$ with no narrow vertices is given in Subsection~\ref{subsec:lower_bound_model1}. Here we describe strategies for polygons with one, two or three narrow vertices. We assume, as before, that the algorithm is deterministic and does not repeat the same probe during same reconstruction session. 

\subsubsection{Lower Bound for Polygons with One Narrow Vertex}
\label{subsubsec:Lower_Bound_NB=1}

${}$

\vspace{0.5 em}

We begin with the case where $N_B = 1$. We will prove a lower bound of $2n - 1$ $\omega$-probes. The adversary defines a circle $\Psi$, a point $p$ as a center of $\Psi$, chooses $n \geq 3$ and sets $N_B = 1$. The adversary maintains a closed convex curve $A$, that is the boundary of an intersection of all $\omega$-probes made by the algorithm and the circle $\Psi$. Initially, $A$ is the circle $\Psi$. The vertices of $\obj$ revealed during the execution of the algorithm remain fixed, defining sections of the curve that cannot be changed. At the end of a query session $A$ is a convex polygon $\obj$.

\vspace{0.5em}
The query made by the algorithm consists of a direction $\overrightarrow{L}$ for the probe. Assume that the probe is valid and $\overrightarrow{L}$  contains $p$. (We treat non-valid probes similarly to the strategy described in Subsection~\ref{subsec:lower_bound_model1}). The adversary stops the probe at any moment after the apex $q$ of the probe enters the interior of $\Psi$, but before it reaches $p$. Let $(q,H_1,H_2,p_1,p_2)$ be the outcome of the probe along $\overrightarrow{L}$ returned by the adversary, such that $H_1$ (respectively $H_2$) makes a negative (respectively positive) angle of $\omega / 2$ with $\overrightarrow{L}$; $q = p_1 = p_2$. In other words, on the first valid probe of the algorithm, the adversary reveals the narrow vertex $v_B = q$ of $\obj$. Refer to Figure~\ref{fig:adversary_NB=1}.

\begin{figure}[h]
\centerline{\resizebox{!}{6.3cm}{\includegraphics{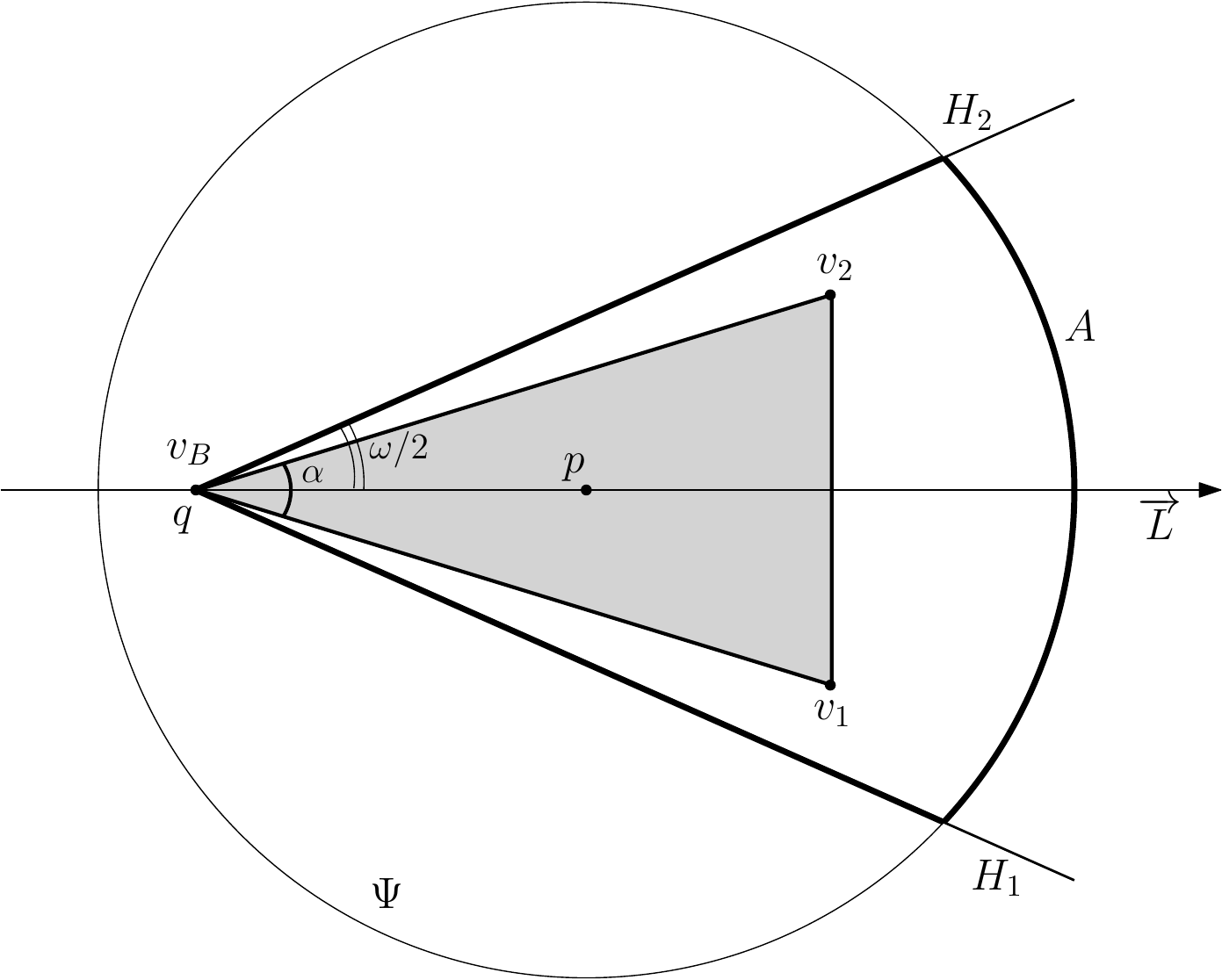}}}
\caption{Initialization step for $0 < \omega < \pi/2$. Adversary induces exactly one narrow vertex in $\obj$ by revealing it at the first valid probe of the algorithm. The closed curve $A$ is highlighted in bold. $\obj$ is fully contained within $A$.}
\vspace{-1.0em}
\label{fig:adversary_NB=1}
\end{figure}

To ensure that $v_B$ is the only narrow vertex of $\obj$, the adversary creates a \textit{frame} polygon (that satisfies $N_B = 1$) enclosed in $\obj$ and reveals it during the first probes of the algorithm. Depending on the value of $\omega$ and probing strategy of the algorithm, three types of frames are possible: isosceles triangle with apex at $v_B$ (for $0 < \omega <\pi/2$), quadrilateral or pentagon (for $\omega = \pi/2$).

\vspace{0.5em}
Let us start with the case when $0 < \omega <\pi/2$. The adversary creates an isosceles triangle with apex at $v_B$ and base vertices $v_1$ and $v_2$. The triangle $\triangle (v_B, v_1, v_2)$ contains $p$; $v_1$ and $v_2$ are inside $A$, but not on the boundary; and the angle $\alpha$ at the apex of the triangle is such that $0 < \alpha < \omega$ and $\alpha < \pi - 2 \omega$. Refer to Figure~\ref{fig:adversary_NB=1}. Notice that $v_1$ and $v_2$ are not fixed yet. The adversary can move them, depending on the the second query of the algorithm, as long as all the above constraints on the triangle hold.

The second valid probe of the algorithm is answered with respect to the $\omega$-cloud of the triangle, unless the direction of the probe coincides with one of its edges. In this case, the adversary slightly moves either $v_1$ or $v_2$, such that this probe will \textbf{not} confirm an edge. The adversary creates an $\omega$-cloud for the corrected triangle and answers the query according to the new $\omega$-cloud. This probe may reveal up to two new vertices of $\obj$. If the current probe does not reveal $v_1$ or $v_2$, the adversary fixes both $v_1$ and $v_2$, but returns only information revealed by the probe to the algorithms. This concludes the initialization step for the case $0 < \omega <\pi/2$.

\vspace{0.5em}
Let us explore the initialization step when $\omega = \pi/2$. The adversary stops the first valid probe at any moment after the apex $q$ of the probe enters the interior of $\Psi$, but before it reaches $p$. Let $(q,H_1,H_2,p_1,p_2)$ be the outcome of the probe along $\overrightarrow{L}$ returned by the adversary, such that $H_1$ (respectively $H_2$) makes a negative (respectively positive) angle of $\omega / 2$ with $\overrightarrow{L}$; $p_1 \neq q$; $p_2 \neq q$; $p \in \overline{p_1 p_2}$. In other words, on the first valid probe of the algorithm, the adversary reveals two vertices of $\obj$. Refer to Figure~\ref{fig:adversary_NB=1_pi}.

\begin{figure}[h]
    \begin{center}
        \subfigure[The adversary reveals two vertices $v_1$ and $v_2$ of $\obj$ on the first valid probe of the algorithm.]{
            \label{fig:adversary_NB=1_pi}
            \includegraphics[width=0.395\textwidth]{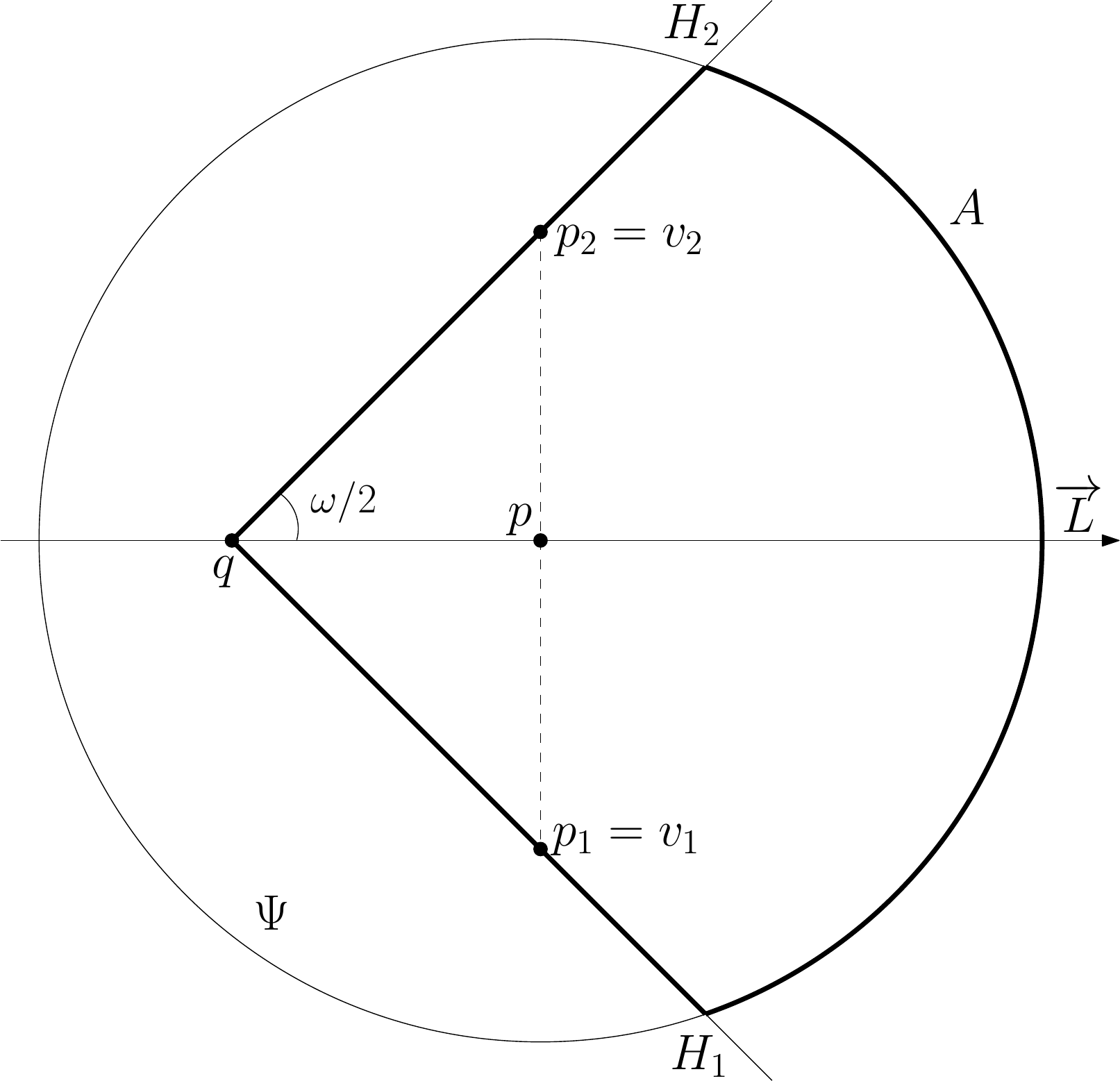}
        }%
        \hspace{0.5cm}%
        \subfigure[The second valid probe of the algorithms is shot along the same line as its first valid probe but in the opposite direction.]{
            \label{fig:adversary_NB=1_pi_square}
            \includegraphics[width=0.4\textwidth]{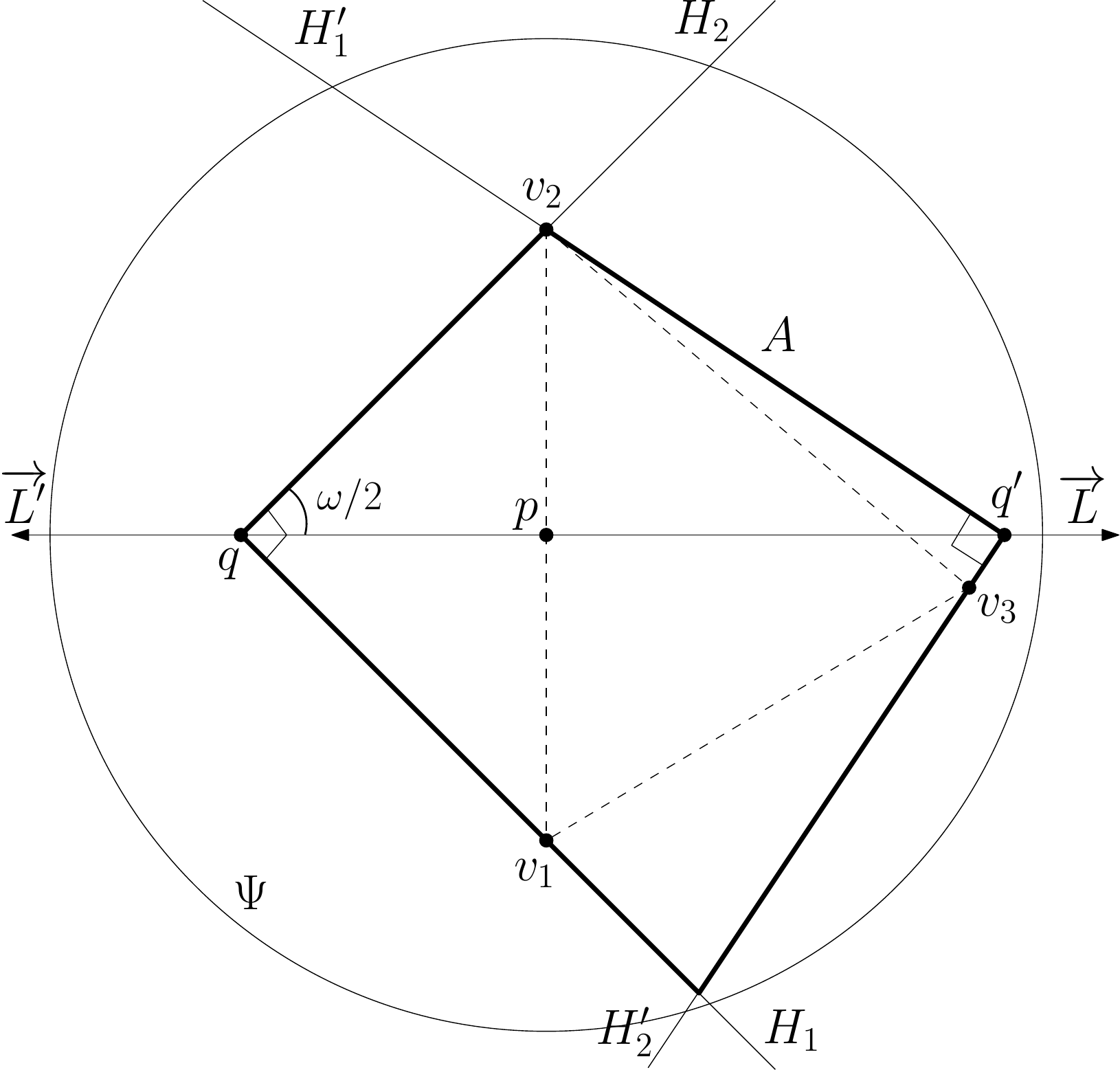}
        }%
    \end{center}
    \vspace{-1.5em}
    \caption{Initialization step for $\omega = \pi/2$. The closed curve $A$ is highlighted in bold.}
    \vspace{-1.0em}
\end{figure}

Let $\overrightarrow{L'}$ be the direction of the second valid probe of the algorithm. Three cases are possible:
\begin{enumerate}
\item If $v_1, v_2 \in \overrightarrow{L'}$, then assume $\overrightarrow{L'}$ intersects $v_1$ before it intersects $v_2$ (the other case is symmetrical). Let $(q',H_1',H_2',p_1',p_2')$ be the outcome of the probe along $\overrightarrow{L'}$ returned by the adversary, such that $q' = p_1' = p_2' = v_1$, $H_1'$ (respectively $H_2'$) makes a negative (respectively positive) angle of $\omega / 2$ with $\overrightarrow{L'}$.  In other words, the adversary declares the already revealed vertex $v_1$ as a narrow vertex of $\obj$ and positions the arms of the second probe such that they are not coincide with the line segment $\overline{v_1v_2}$.

\item Else, if $\overrightarrow{L'} \neq \overleftarrow{L}$, return the already discovered vertices $v_1$ and $v_2$ of $\obj$. In other words, the adversary constructs the $\omega$-cloud of the only vertices $v_1$ and $v_2$ and returns the outcome of the probe along the given $\overrightarrow{L'}$ on this cloud. Assume, that the apex of the second probe is closer to $v_2$ than to $v_1$ (the other case is symmetrical). The adversary marks $v_1$ as a narrow vertex (the algorithm may correctly guess this fact). 

\item If $\overrightarrow{L'} = \overleftarrow{L}$ ($\overrightarrow{L'}$ coincides with $\overrightarrow{L}$ but has the opposite direction), then the adversary reveals the new vertex $v_3$ of $\obj$ and returns the following outcome: $(q',H_1',H_2',p_1',p_2')$, where $p_1' = v_2$; $p_2' = v_3$, $q'$ and $v_3$ are close to each other, but $q' \neq v_3$; $|qp| < |pq'|$ and $|pq'|$ is smaller then the radius of $\Psi$. Refer to Figure~\ref{fig:adversary_NB=1_pi_square}. 
\end{enumerate}

\vspace{0.5em}
Let $\overrightarrow{L''}$ be the direction on the third valid probe of the algorithm. The answer of the adversary depends also on the direction of the previous probe $\overrightarrow{L'}$.

For the first two cases of $\overrightarrow{L'}$, the adversary defines two vertices: $v_3$ and $v_4$ (interior to $A$ but not on the boundary), such that three of the internal angles of the resulting quadrilateral $\{ v_1, v_3, v_2, v_4\}$ are strictly bigger than $\pi/2$ and the fourth internal angle, rooted at the narrow vertex, is strictly smaller than  $\pi/2$. Moreover, none of the edges of the quadrilateral is contained within $\overrightarrow{L''}$. Note that $v_3$ and $v_4$ become fixed after the adversary returns the outcome of the third valid probe of the algorithm (even if that probe does not discover one or both of them). Note also that there are infinitely many possibilities for the adversary to choose $v_3$ and $v_4$ to satisfy the above conditions. Refer to Lemma~\ref{lem:feasible_triangle}.

\vspace{0.5em}
Let us explore the third case of $\overrightarrow{L'}$. Note that in this case, the algorithm already discovered three vertices of $\obj$: $v_1$, $v_2$ and $v_3$. Refer to Figure~\ref{fig:adversary_NB=1_pi_square}. 
\begin{itemize}
\item If $\overrightarrow{L''}$ coincides with one of the edges of the triangle $\triangle (v_1, v_2, v_3)$, the adversary declares the vertex of entrance $\overrightarrow{L''}$ into $A$ as a narrow vertex and places the arms such that they maximize $A$ and do not touch other vertices of $\obj$. For example, if $\overrightarrow{L''}$ contains $\overline{v_1 v_2}$ and is directed from $v_1$ to $v_2$, the outcome of the probe $(q'',H_1'',H_2'',p_1'',p_2'')$ is as follows: $q'' = p_1'' = p_2'' = v_1$; $v_2, v_3 \notin H_1''$ and $v_2, v_3 \notin H_2''$.
\item Otherwise, the adversary creates an $\omega$-cloud of the triangle $\triangle (v_1, v_2, v_3)$ and answers the query of the algorithm according to this cloud. In other words, the probe does not reveal any new information about $\obj$, but may have changed the shape of $A$. 
\end{itemize}

To finalize the initialization step and to create a frame for $\obj$ with exactly one narrow vertex, the adversary creates one new vertex $v_4$ in the interior of $A$, such that the convex quadrilateral $\{v_1, v_2, v_3, v_4\}$ has exactly one acute angle. It is always possible in the case when $\overrightarrow{L''}$ contains edge of the triangle $\triangle (v_1, v_2, v_3)$. Unfortunately, it is not always true for general $\overrightarrow{L''}$. The adversary may require one more vertex $v_5$ (in the interior of $A$), such that the resulting convex pentagon $\{v_1, v_2, v_3, v_4, v_5\}$ has exactly one acute angle (and $v_4$ and $v_5$ are not adjacent on the boundary of the pentagon). During the following probes the algorithm may not discover $v_4$ or $v_5$. As soon as discovered part of $\obj$ satisfies frame constraints (frame is a polygon with exactly one acute angle), the not discovered helping vertices can be discarded. Thus, during the second stage of his strategy, the adversary places a new vertex in a way that the internal angles of $\obj$ are maximized (as long as the internal angle corresponding to the narrow vertex is smaller than $\omega$). Notice that the vertices $v_4$ and $v_5$ can be discovered by one probe of the algorithm. The vertex at the acute angle of the pentagon or quadrilateral frame is marked as a narrow vertex (whether or not the algorithm discovered this fact).  
This concludes the initialization stage for $N_B = 1$ and $0 < \omega \leq \pi/2$.

\vspace{0.5em}
The rest of the adversary's strategy is similar to the one described in the  Subsection~\ref{subsec:lower_bound_model1}. The only difference is that when a new vertex is inserted between two existing ones (one of which is a narrow vertex), the angle at the narrow vertex should stay smaller than $\omega$. These precautions are taken since the narrow vertex may have been discovered by the  algorithm but not necessarily identified as a narrow vertex, and thus the internal angle at the corresponding vertex of $A$ (where the boundary of $A$ touches the narrow vertex) can be bigger than $\omega$.

\vspace{0.5em}
We analyze the number of probes spent by the algorithm in the following theorem.

\begin{theorem}[Lower Bound $N_B=1$]
\label{theo:B-case_sum-up_Lower_Bound_NB=1}
Let an $\omega$-wedge with $0 < \omega \leq \pi/2$ be given. For every algorithm, there exists a convex polygon $\obj$ with one narrow vertex, such that $2n - 1$ $\omega$-probes are required to determine its shape. 
\end{theorem}
\begin{proof}
If $0 < \omega < \pi/2$, the initialization step results in $3$ vertices of $\obj$ (one of which is a narrow vertex) and requires at least $2$ probes. The rest of the polygon is reconstructed in at lest $2n-3$ probes. (Refer to the Subsection~\ref{subsec:lower_bound_model1}.) Thus, in total, $2n-1$ probes are required to reconstruct $\obj$.

If $\omega = \pi/2$, the algorithm spends at least one probe to discover the first $2$ vertices. Depending on the strategy of the algorithm, the initialization step terminates after at least one more probe, that returns no new vertices or edges; or after at least $2$ additional probes that result in one new vertex. Notice that because of the constraints induced by the frame, at most one of the consecutive probes may return 2 new pieces of information. The rest of $\obj$ requires at least one probe per edge and at least one probe per vertex. (Refer to the Subsection~\ref{subsec:lower_bound_model1}). This results in at least $2n-1$ probes and proves the following theorem.
\qed
\end{proof}

\subsubsection{Lower Bound for Polygons with Two or Three Narrow Vertices}
\label{subsubsec:Lower_Bound_NB>=2}

${}$

\vspace{0.5em}

For $N_B \geq 2$ the adversary should provide the algorithm with a small real number $\varepsilon > 0$. For every pair of narrow vertices of $\obj$: $v_{B_1}$ and $v_{B_2}$, the following is true: if $\overline{v_{B_1} v_{B_2}}$ is not an edge that belongs to $\obj_{v_{B_1} v_{B_2}}$, then there exists a vertex $v \in \obj_{v_{B_1} v_{B_2}}$ such that $\angle (v_{B_2}, v, v_{B_1} ) \leq \pi-\varepsilon$. This ensures that one of the angles $\angle (v, v_{B_1}, v_{B_2})$ or $\angle (v_{B_1}, v_{B_2}, v)$ is at least $\varepsilon / 2$. Without $\varepsilon$, the polygon may not be reconstructed. 

\vspace{0.5em}
Let us describe the adversary's strategy for convex polygons with exactly two narrow vertices. The adversary chooses $\varepsilon = \omega / 10$, defines a circle $\Psi$, a point $p$ as a center of $\Psi$ and reveals all this information to the algorithm. It then chooses $n \geq 3$, sets $N_B = 2$, initialises $A$ to be $\Psi$ and sets $\flag = \FALSE$.

\begin{figure}
    \begin{center}
        \subfigure[The apex of the first valid probe of the algorithms coincides with $p$.]{
            \label{fig:adversary_NB=2}
            \includegraphics[width=0.4\textwidth]{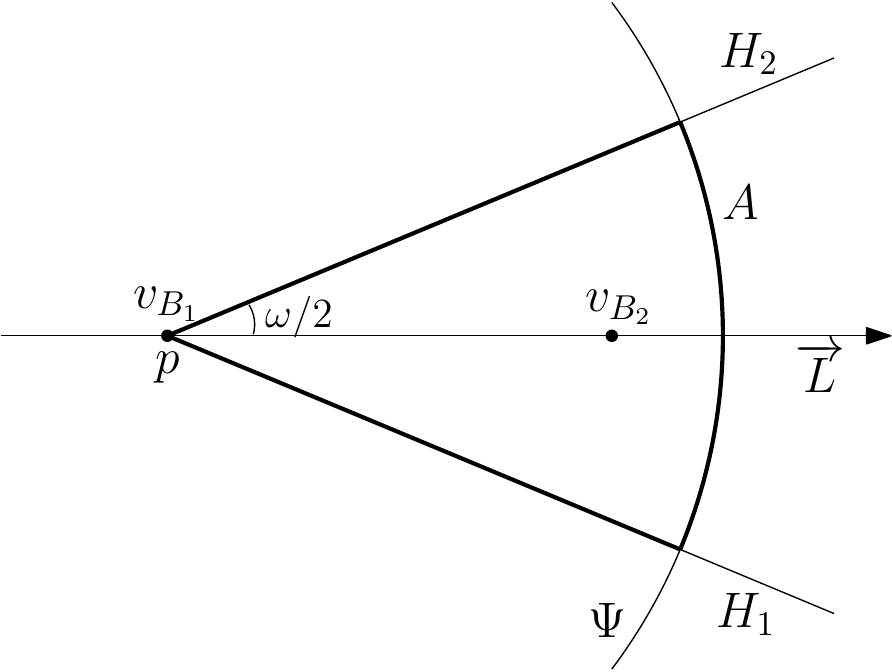}
        }%
        \hspace{0.3cm}%
        \subfigure[Possible result of the second valid probe. For the sake of demonstration the size of the angle $\angle (q', p, v_{B_2})$ is exaggerated. The vertex $v$ is the possible new position for $v_{B_2}$.]{
            \label{fig:adversary_NB=2_secondprobe}
            \includegraphics[width=0.4\textwidth]{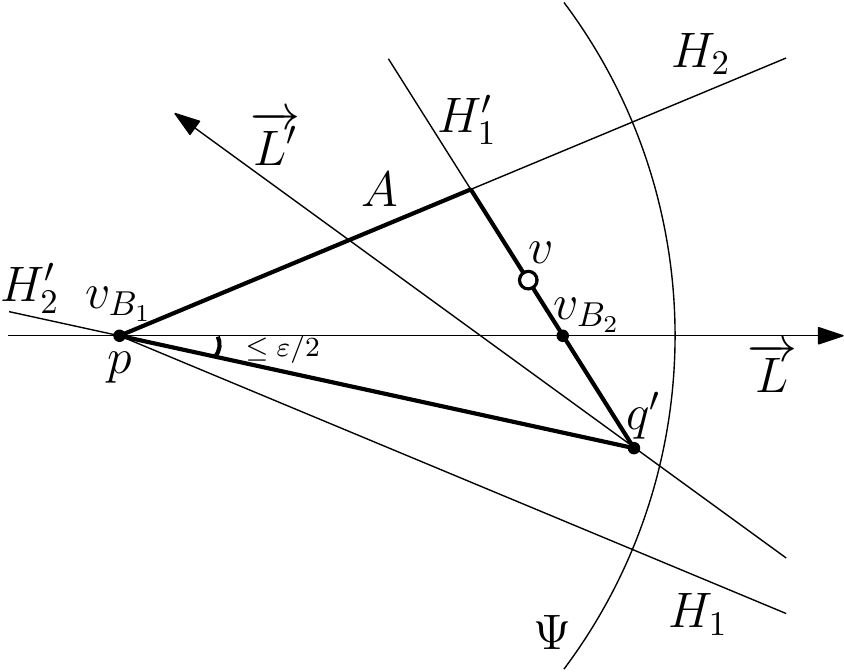}
        }%
    \end{center}
    \vspace{-1.5em}
    \caption{Initialization step for $N_B = 2$. The closed curve $A$ is highlighted in bold.}
    \vspace{-1.0em}
\end{figure}

Let $\overrightarrow{L}$ be the direction of the first valid probe of the algorithm such that $p \in \overrightarrow{L}$. Otherwise, we let this probe miss the object and perform updates to our working space and $A$ as described in Subsection~\ref{subsec:lower_bound_model1}. The adversary stops the probe when its apex coincides with $p$. Let $(q,H_1,H_2,p_1,p_2)$ be the outcome of the probe along $\overrightarrow{L}$ returned by the adversary, such that: $q = p = v_{B1} = p_1 = p_2$, $H_1$ (respectively $H_2$) makes a negative (respectively positive) angle of $\omega / 2$ with $\overrightarrow{L}$. Refer to Figure~\ref{fig:adversary_NB=2}. Notice that the shape of the polygon with two narrow vertices can be skinny. To ensure $p \in \obj$, we let $v_{B1} = p$. 

\vspace{0.5em}
We create another narrow vertex $v_{B2}$. Initially, the adversary places it such that $v_{B2} \in \overrightarrow{L}$ and $v_{B2} \in A$, but not on the boundary of $A$, which also ensures that $v_{B1} \neq v_{B2}$. Let $\overrightarrow{L'}$ be the direction of the second valid probe made by the algorithm, such that $\overrightarrow{L'}$ intersects the line segment $\overline{v_{B_1} v_{B_2}}$. Let $(q',H_1',H_2',p_1',p_2')$ be the outcome of the probe on the polygon $\{v_{B1}, v_{B2}\}$. If $\overrightarrow{L'}$ does not intersect $\overline{v_{B_1} v_{B_2}}$, then the adversary does the following: declares the probe as not valid (meaning, the probe missed the object), shrinks its workspace such that the updated $A$ does not have common points with $\overrightarrow{L'}$, and creates $v_{B2}$ as described earlier.
 
\begin{itemize}
\item If $q' = v_{B_2}$ or $q' = v_{B_1}$, we rotate the $\omega$-wedge around its apex such that $\overrightarrow{L}$ bisects the wedge and $\obj$ stays inside the wedge. The adversary fixes $v_{B_2}$ on its current position, marks it as a narrow vertex and returns the result of the second probe to the algorithm.
\item Otherwise, if none of the arms $H_1'$ or $H_2'$ creates an angle of $\varepsilon/2$ or smaller with $\overrightarrow{L}$, then fix $v_{B_2}$ on its current position and return the output of the second probe to the algorithm. The adversary marks $v_{B_2}$ as a narrow vertex (notice that in this case, the algorithm did not discover this fact, but may correctly guess it).
\item Otherwise, if $H_1'$ or $H_2'$ creates an angle of $\varepsilon/2$ or smaller with $\overrightarrow{L}$, then two cases are possible:
\begin{enumerate}
\item If the angle between $\overrightarrow{L}$ and $\overrightarrow{L'}$ is bigger than $\pi/2$ (refer to Figure~\ref{fig:adversary_NB=2_secondprobe}) the adversary moves $v_{B_2}$ along the supporting arm of the second probe away from $q'$, such that $\angle (q', p, v_{B_2}) > \varepsilon/2$ and the angles between the line segment $\overline{v_{B_1} v_{B_2}}$ with $H_1$ and $H_2$ are bigger than $\varepsilon/2$. (Consider, for example, Figure~\ref{fig:adversary_NB=2_secondprobe}. The possible new position for $v_{B_2}$ is marked by $v$.)

\item If the angle between $\overrightarrow{L}$ and $\overrightarrow{L'}$ is smaller than $\pi/2$, we consider the two possible values of $\flag$: 
\begin{itemize}
\item If $\flag = \FALSE$, the adversary declares the current probe as non valid, shrinks its working space, updates $A$, creates $v_{B_2}$ and sets $\flag = \TRUE$. 
\item If $\flag = \TRUE$, the adversary fixes $v_{B_2}$ on its current position and returns the output of the current probe to the algorithm.
\end{itemize}

\end{enumerate}
\end{itemize}

This concludes the initialization step. 

\vspace{0.5em}
The other stages of the adversary's strategy are similar to those used for $N_B = 0$ (refer to Subsection~\ref{subsec:lower_bound_model1}) and $N_B = 1$ (refer to the beginning of the current subsection). The only difference is due to the use of $\varepsilon$.

\vspace{0.5em}
As in the previous strategies, during the second stage, we answer the probes with already revealed vertices, unless the direction of the probe contains two known vertices and the probe does not enter $\obj$ through a narrow vertex. In this case, we reveal one new vertex provided that the number of already revealed vertices is smaller than $n-2$. Similarly to the strategy used for $N_B=1$, all newly created vertices in this stage should not change the number of narrow vertices of $Q \subseteq \obj$. 

The adversary proceeds to the third stage of its strategy when the algorithm knows $n-1$ vertices of $\obj$. During this stage the adversary confirms edges (without revealing any new vertices), until $n-2$ edges are confirmed. Every probe of the algorithm is answered with already revealed vertices of $\obj$, unless the direction of the probe $\overrightarrow{L}$ contains two vertices adjacent on $A$ such that $\overrightarrow{L}$ does not enter $\obj$ through a narrow vertex. In this case, an edge between the pair of vertices is confirmed.

During the final, fourth stage of the strategy, the last unknown vertex is revealed to the algorithm and its two adjacent edges are confirmed.

Without knowing $\varepsilon$, the algorithm may not be able to reconstruct $\obj$ with $N_B \geq 2$ precisely. Let $\overrightarrow{L_{v_{B_1}v_{B_2}}}$ be the line through a pair of narrow vertices $v_{B_1}$ and $v_{B_2}$, directed from $v_{B_1}$ to $v_{B_2}$. Suppose that there are no revealed vertices of $\obj$ to the left of $\overrightarrow{L_{v_{B_1}v_{B_2}}}$. Let $F$ be a feasible region of $Q \subseteq \obj$ to the left of $\overrightarrow{L_{v_{B_1}v_{B_2}}}$.
Assume that the probe of the algorithm stops when $v_{B_1} \in H_1$ and $v_{B_2} \in H_2$ (notice, that the apex of the probe is to the left of $\overrightarrow{L_{v_{B_1}v_{B_2}}}$). If the angle between $H_2$ (respectively, $H_1$) and $\overrightarrow{L_{v_{B_1}v_{B_2}}}$ is smaller than $\varepsilon/2$ and the angle at $v_{B_1}$ (respectively, $v_{B_2}$) of $F$ is smaller than $\varepsilon/2$, then:

\begin{enumerate}
\item If we are in the second stage of the strategy, one new vertex $v$ is created in $F$ (but not on its boundary). The adversary makes the probe touch $v$ and returns the output to the algorithm. Notice that in the second stage, $\obj$ has no confirmed edges.
\item If we are in the third stage, the adversary outputs the result as it is and the algorithm can interpret it as an edge between $v_{B_1}$ and $v_{B_2}$. In this case, $F$ is the line segment $\overline{v_{B_1} v_{B_2}}$.
\item If we are in the fourth stage and the algorithm did not reveal the last $n$-th vertex of $\obj$, the adversary proceeds similarly to Step $1$. If the algorithm knows all the vertices but did not confirm all the edges, we proceed according to Step $2$. 
\end{enumerate}

\vspace{1em}
We conclude the adversary's strategy with the case when $N_B = 3$. Notice that this case can happen only when $\pi/3 < \omega \leq \pi/2$. The last three stages of the strategy are identical to the case when $N_B = 2$, so we describe in short the initialization step only. 

The adversary chooses $\varepsilon = \frac{1}{10}(\omega - \pi/3)$, defines a circle $\Psi$, a point $p$ as a center of $\Psi$ and reveals all this information to the algorithm. It then chooses $n \geq 3$, sets $N_B = 3$ and initialises $A$ to be $\Psi$.

On the first valid probe of the algorithm, whose direction $\overrightarrow{L}$ contains $p$, the adversary returns two vertices of $\obj$: $v_{B_1}$ and $v_{B_2}$. Let $(q,H_1,H_2,p_1,p_2)$ be the outcome of the probe along $\overrightarrow{L}$ returned by the adversary, such that: $p_1 = v_{B_1}$, $p_2 = v_{B_2}$, $q \neq p_1$, $q \neq p_1$, $|q p_1| = |q p_2|$, $H_1$ (respectively $H_2$) makes a negative (respectively positive) angle of $\omega / 2$ with $\overrightarrow{L}$. Refer to Figure~\ref{fig:adversary_NB_3}. The adversary marks $v_{B_1}$ and $v_{B_2}$ as narrow vertices.

\begin{figure}[ht]
\centerline{\resizebox{!}{8cm}{\includegraphics{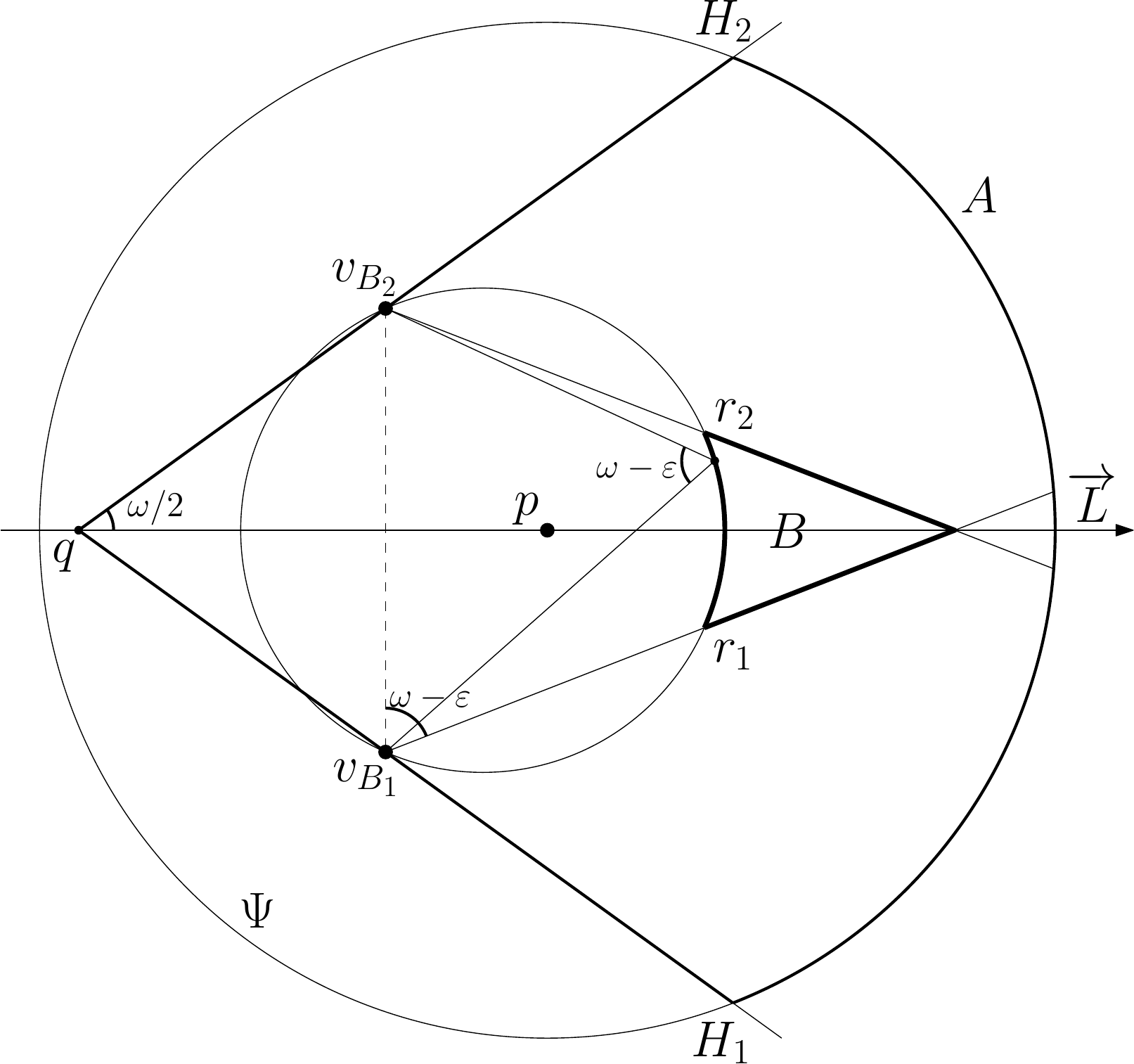}}}
\caption{Initialization step for $N_B = 3$. The closed curve $A$ is highlighted in bold. The third narrow vertex $v_{B_3}$ can be positioned inside the region $B$, marked in extra-bold.}
\label{fig:adversary_NB_3}
\end{figure}

Let $r_1$ (respectively, $r_2$) be a ray emanating from $v_{B_1}$ (respectively, $v_{B_2}$) and making a negative (respectively, positive) angle of $\omega - \varepsilon$ with the line segment $\overline{v_{B_1} v_{B_2}}$.
Let $p'$ be a point interior to the $\omega$-wedge of the first valid probe, such that $\angle (v_{B_2}, p', v_{B_1}) = \omega - \varepsilon$, and let $C$ be a disk defined by $v_{B_1}$, $v_{B_2}$ and $p'$. Let $H_{r_1}$ (respectively, $H_{r_2}$) be a half-plane containing $r_1$ (respectively, $r_2$) on its boundary and containing $q$. Let $H'$ be a half-plane containing $\overline{v_{B_1} v_{B_2}}$ on its boundary and not containing $q$. The region $B$, where the third narrow vertex $v_{B_3}$ can be positioned, is defined as $\Psi \cap H_{r_1} \cap H_{r_2} \cap H' \setminus C$. Refer to Figure~\ref{fig:adversary_NB_3}.

The adversary stops the first valid probe in such a way that: $B$ is non-empty; $v_{B_1}$ and $v_{B_2}$ are interior to $\Psi$; and for every point $t \in B$, the triangle $\triangle (v_{B_1}, v_{B_2}, t)$ contains $p$.

\vspace{0.5em}
Let $L'$ be the direction of the next valid probe of the algorithm.  The vertex $v_{B_3}$ is created and revealed when it is possible to position the $\omega$-wedge of the probe along $L'$ such that one of its arms, or the apex itself, intersects $B$. The vertex $v_{B_3}$ should be created such that none of the arms of the probe makes an angle of $\varepsilon / 2$ or smaller with the line segment $\overline{v_{B_1} v_{B_2}}$, $\overline{v_{B_2} v_{B_3}}$ or $\overline{v_{B_1} v_{B_3}}$. Since $v_{B_1}$ and $v_{B_2}$ are narrow vertices, $v_{B_3}$ cannot be revealed simultaneously (via the same probe) with a new vertex or an edge. 

Notice that the region $B$ can shrink as a result of the vertex $v$ being revealed between $v_{B_1}$ and $v_{B_2}$ on the boundary of $\obj$ in a counter-clockwise direction from $v_{B_2}$ to $v_{B_1}$. We prevent $B$ from disappearing by positioning $v$ such that $\pi - 2 \varepsilon \leq \angle (v_{B_1}, v, v_{B_2}) \leq \pi - \varepsilon$.

\begin{theorem}[Lower Bound $N_B=2$, $N_B=3$]
\label{theo:B-case_sum-up_Lower_Bound}
Given $\omega$-wedge ($0 < \omega \leq \pi/2$), convex polygon $\obj$, and a real number $\epsilon > 0$ (such that for every pair of narrow vertices of $\obj$: $v_{B_1}$ and $v_{B_2}$, if $\overline{v_{B_1} v_{B_2}}$ is not an edge that belongs to $\obj_{v_{B_1} v_{B_2}}$, then there exists a vertex $v \in \obj_{v_{B_1} v_{B_2}}$ such that $\angle (v_{B_2}, v, v_{B_1}) \leq \pi-\varepsilon$). For every algorithm, there exists a convex polygon $\obj$ with two or three narrow vertices such that $2n + 2$ $\omega$-probes are required to determine its shape.
\end{theorem}

\begin{proof}
Suppose that $N_B = 3$. During the initialization step, at least one probe is used and two vertices are revealed. One of the subsequent probes reveals $v_{B_3}$ and no additional information. The algorithm wastes at least one probe per every pair of narrow vertices, which results in at least 3 probes and no new information. The rest of the probes return each at most one new piece of information about $\obj$, making it at least $2n-3$ probes. Thus, a total of $2n+2$ probes are necessary to reconstruct $\obj$.

The analysis of the strategy for polygons with $N_B = 2$ is similar. The initialization step requires at least two probes and results in two vertices of $\obj$. At least two probes bring no new information due to the usage of $\varepsilon$. Not less than $2n-2$ additional probes are required to reveal $n-2$ vertices and to confirm $n$ edges. This adds up to at least $2n+2$ probes.
\qed
\end{proof}

\section{Conclusion and Future Work}
\label{sec:Conclusion}

In this paper, we presented an $\omega$-wedge device together with optimal probing algorithms for convex polygons with no acute angles or when at most one angle is smaller then $\omega$. When the number of such angles is more than one (possibly only two or three), the performance is weaker, yet we show it to be {\it almost} optimal. The main reason why our results are not tight is due to the fact that the number of {\it bad} angles is assumed a priory unknown. Knowing this extra information about polygon $\obj$ beforehand would improve our strategy and bring down the upper bounds. In particular, Lemma~\ref{lem:every_probe_touches_B} and Observation~\ref{observ:different} show that it is possible to detect that a particular vertex is narrow without specifically confirming this fact, that is, saving one probe. This means that we could save up to $3$ probes (for polygons with exactly $3$ narrow vertices), which would make our strategy optimal.

\begin{lemma}
\label{lem:every_probe_touches_B}
Assume that we are given a convex polygon $\obj$ with exactly $3$ narrow vertices. Every valid $\omega$-probe of $\obj$ touches a narrow vertex.
\end{lemma}
\begin{proof}
Assume to the contrary that the outcome ($q$, $H_1$, $H_2$, $p_1$, $p_2$) of the valid probe returns two vertices $p_1$ and $p_2$ such that none of them is a narrow vertex. Refer to Figure~\ref{fig:lemma_6}. By definition, a narrow vertex belongs to the $\omega$-cloud of $\obj$, meaning that there is no arc of the $\omega$-cloud above the narrow vertex. Thus, there are no narrow vertices between $p_1$ and $p_2$ in a counter-clockwise direction from $p_2$ to $p_1$ around the boundary of $\obj$.

Let $v_{B_1}$, $v_{B_2}$ and $v_{B_3}$ be the three narrow vertices of $\obj$. Let $v$ be an intersection of $H_2$, and the line through $v_{B_1}$ and $v_{B_2}$; and let $u$ be an intersection of $H_1$, and the line through $v_{B_1}$ and $v_{B_3}$. Consider the quadrilateral $\{q, u, v_{B_1}, v\}$. By definition, the internal angles of $\obj$ corresponding to the narrow vertices are smaller or equal to $\omega$. Thus, the internal angles of the quadrilateral at $v$ and $u$ are strictly smaller than $\omega$. The angle at $q$ equals $\omega$. Thus, the sum of the internal angles of the quadrilateral is smaller then $4 \omega$ and thus smaller then $2 \pi$, which is a contradiction.
\qed
\end{proof}

\begin{figure}[h]
\centerline{\resizebox{!}{4cm}{\includegraphics{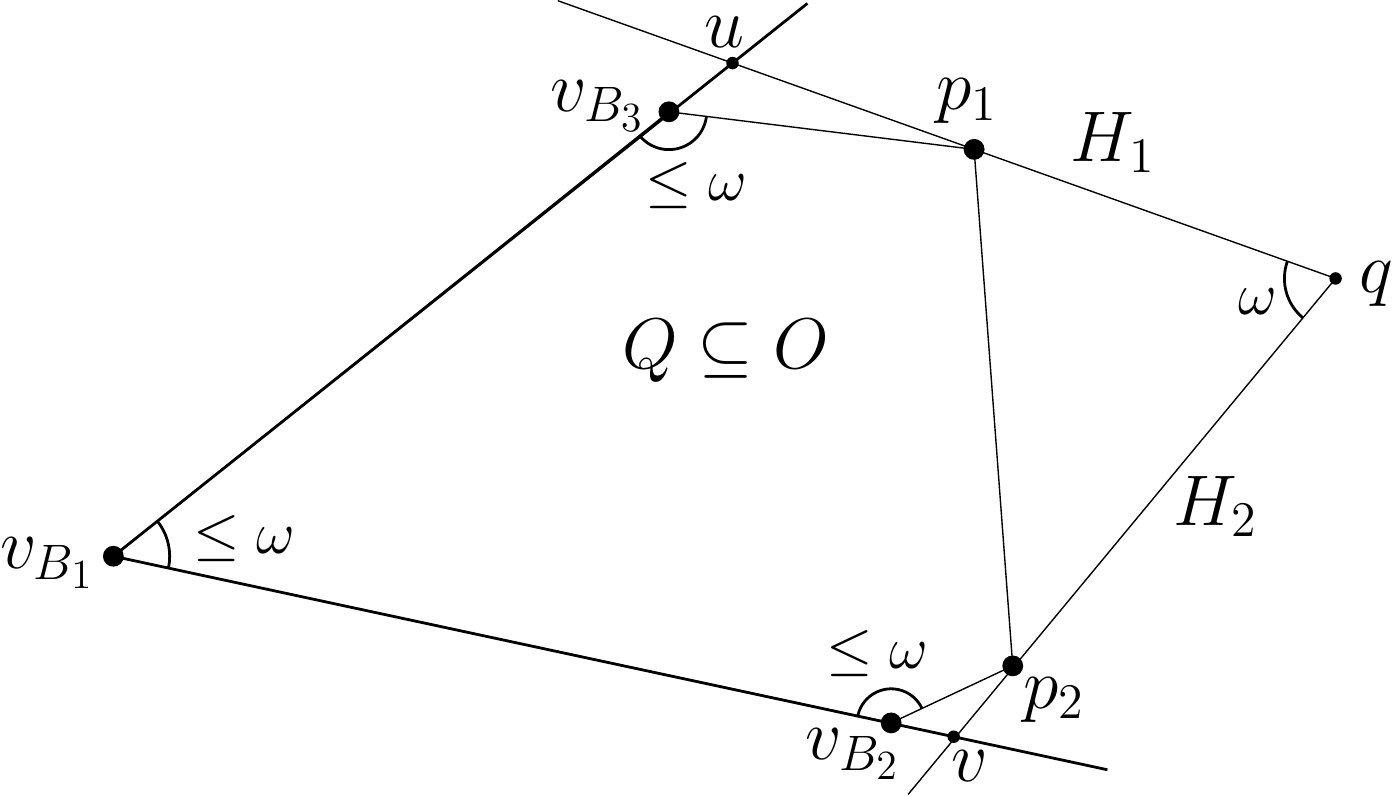}}}
\caption{Illustration of the proof of Lemma~\ref{lem:every_probe_touches_B}. $v_{B_1}$, $v_{B_2}$ and $v_{B_3}$ are narrow vertices.}
\label{fig:lemma_6}
\end{figure}

The following observation follows from Lemma~\ref{lem:every_probe_touches_B}.

\begin{observation}
\label{observ:different}
Assume that we are given a convex polygon $\obj$ with exactly $3$ narrow vertices. Let ($q$, $H_1$, $H_2$, $p_1$, $u$) and ($q'$, $H'_1$, $H'_2$, $u$, $p_2$) be two outcomes of two valid $\omega$-probes of $\obj$, respectively. Suppose that $q$ and $q'$ belong to different arcs of $\omega$-cloud of $\obj$ and $p_1 \neq p_2$. If $\angle (p_1, u, p_2) > \omega$, then $p_1$ and $p_2$ are narrow vertices. Otherwise, if $\angle (p_1, u, p_2) \leq \omega$, $u$ is a narrow vertex.
\end{observation}

\vspace{0.5em}
Alternatively, we can take advantage of composite probing strategy that incorporates usage of up to $3$ finger probes. Whenever we have a pair of narrow vertices $v_{B_1}$, $v_{B_2}$ (with no additional information discovered between them on the boundary of $\obj$), we can shoot a finger probe orthogonally to $\overline{v_{B_1} v_{B_2}}$. For convex polygons with two or three narrow vertices we will no longer require the help of $\varepsilon$ and no probes will be wasted on every pair of narrow vertices. This will result in optimal strategy without the advantage of knowing number of narrow vertices beforehand.

\vspace{0.5em}
We would like to explore different models of $\omega$-probing. For instance, one in which the outcome of a probe does not give us the contact points with the polygon. We know that $\obj$ belongs to the interior of a valid $\omega$-probe. We are given the orientation of the arms $H_1$ and $H_2$ and we know that they touch $\obj$, but we do not know the exact location of these contact points.

\vspace{0.5em}
In our current model, we reconstruct $\obj$ starting from within of $\obj$. We form a polygon $Q$ by connecting known vertices of $\obj$, so $Q$ completely resides inside $\obj$. We expand $Q$ by adding newly discovered vertices of $\obj$, so at the end of the reconstruction, $Q$ equals $\obj$.

In the model we just described, we should use a different approach. The convex polygon $Q$ is defined as a boundary of an intersection of the circle $\Psi$ and all $\omega$-probes made so far, implying that $\obj$ completely resides inside $Q$. By performing additional probes, we cut away parts of $Q$ so it shrinks closer to $\obj$.

\vspace{0.5em}
Another, interesting model of $\omega$-probing to explore would be one where the outcome of a probe is simply the apex.

\vspace{0.5em}
Our paper investigates $\omega$-probing with small $\omega$'s: $0 < \omega \leq \pi/2$. It would be interesting to consider cases when $\pi /2 < \omega < \pi$. In the latter case, the structure of the $\omega$-cloud is more complicated, since we no longer have the fact that each arc participates in $\omega$-cloud at most once.

\bibliographystyle{alpha}
\bibliography{AngularProbingOrientedArms}

\begin{thebibliography}{BMSS11}

\bibitem[ABK98]{DBLP:conf/siggraph/AmentaBK98}
Nina Amenta, Marshall~W. Bern, and Manolis Kamvysselis.
\newblock A new voronoi-based surface reconstruction algorithm.
\newblock In {\em Proceedings of the 25th Annual Conference on Computer
  Graphics and Interactive Techniques, {SIGGRAPH} 1998, Orlando, FL, USA, July
  19-24, 1998}, pages 415--421, 1998.

\bibitem[BC10]{DBLP:conf/cccg/BoseC10}
Prosenjit Bose and Jean-Lou~De Carufel.
\newblock Minimum enclosing area triangle with a fixed angle.
\newblock In {\em CCCG}, pages 171--174, 2010.

\bibitem[BL91]{DBLP:journals/amai/BrucksteinL91}
Alfred~M. Bruckstein and Michael Lindenbaum.
\newblock Reconstruction of polygonal sets by constrained and unconstrained
  double probing.
\newblock {\em Ann. Math. Artif. Intell.}, 4:345--361, 1991.

\bibitem[BMSS11]{journals/ijcga/BoseMSS09}
P.~Bose, M.~Mora, C.~Seara, and S.~Sethia.
\newblock On computing enclosing isosceles triangles and related problems.
\newblock {\em Int. J. Comput. Geometry Appl.}, 21(1):303--318, 2011.

\bibitem[CY87]{DBLP:journals/jal/ColeY87}
R.~Cole and C.K. Yap.
\newblock Shape from probing.
\newblock {\em J. Algorithms}, 8(1):19--38, 1987.

\bibitem[DEY86]{DBLP:conf/stoc/DobkinEY86}
David~P. Dobkin, Herbert Edelsbrunner, and Chee-Keng Yap.
\newblock Probing convex polytopes.
\newblock In {\em STOC}, pages 424--432, 1986.

\bibitem[ES88]{DBLP:journals/siamcomp/EdelsbrunnerS88}
Herbert Edelsbrunner and Steven Skiena.
\newblock Probing convex polygons with x-rays.
\newblock {\em SIAM J. Comput.}, 17(5):870--882, 1988.

\bibitem[FW09]{DBLP:conf/isaac/FleischerW09}
R.~Fleischer and Y.~Wang.
\newblock On the camera placement problem.
\newblock In {\em ISAAC}, pages 255--264, 2009.

\bibitem[Gar92]{Gardner92}
R.~J. Gardner.
\newblock X-rays of polygons.
\newblock {\em Discrete Comput. Geom.}, 7:281--293, March 1992.

\bibitem[Li88]{DBLP:journals/ipl/Li88a}
Shuo-Yen~Robert Li.
\newblock Reconstruction of polygons from projections.
\newblock {\em Inf. Process. Lett.}, 28(5):235--240, 1988.

\bibitem[MS96]{MeijerSkiena96}
Henk Meijer and Steven~S. Skiena.
\newblock Reconstructing polygons from x-rays.
\newblock {\em Geometriae Dedicata}, 61:191--204, 1996.
\newblock 10.1007/BF00151583.

\bibitem[RG94]{DBLP:journals/ijrr/RaoG94}
Anil~S. Rao and Kenneth~Y. Goldberg.
\newblock Shape from diameter: Recognizing polygonal parts with a parallel-jaw
  gripper.
\newblock {\em I. J. Robotic Res.}, 13(1):16--37, 1994.

\bibitem[Rom95]{RomanikSurvey}
Kathleen~A. Romanik.
\newblock Geometric probing and testing - a survey.
\newblock Technical report, DIMACS, 95-42, September 1995.

\bibitem[Ski88]{Skiena88Dissertation}
Steven~S. Skiena.
\newblock {\em Geometric Probing}.
\newblock PhD thesis, University of Illinois, Urbana, IL, 1988.

\bibitem[Ski89]{DBLP:journals/algorithmica/Skiena89}
Steven~S. Skiena.
\newblock Problems in geometric probing.
\newblock {\em Algorithmica}, 4(4):599--605, 1989.

\bibitem[Ski91]{DBLP:journals/jal/Skiena91}
Steven~S. Skiena.
\newblock Probing convex polygons with half-planes.
\newblock {\em J. Algorithms}, 12(3):359--374, 1991.

\end{thebibliography}

 

\end{document}